\documentclass{article}
\usepackage{authblk}
\usepackage{amsmath,amssymb,amsthm,amsfonts,bm}
\usepackage{geometry}
\usepackage{color,enumerate}
\usepackage{hyperref}
\usepackage{breqn}
\usepackage{graphicx,subfigure,float}
\usepackage[capitalize,nameinlink]{cleveref}

\newcounter{subeq}
\renewcommand{\thesubeq}{\theequation.\arabic{subeq}}
\newcommand{\newsubeqblock}{\setcounter{subeq}{0}\refstepcounter{equation}}
\newcommand{\mysubeq}{\refstepcounter{subeq}\tag{\thesubeq}}
\renewcommand{\Re}{\operatorname{Re}}
\renewcommand{\Im}{\operatorname{Im}}
\renewcommand{\i}{\mathrm{i}}
\DeclareMathOperator{\sech}{sech}

\newtheorem{theorem}{Theorem}
\newtheorem{lemma}{Lemma}
\newtheorem{remark}{Remark}%

\usepackage{natbib}

\begin{document}

\title{Soliton solutions to the coupled Sasa-Satsuma equation under mixed boundary conditions}

\author[1]{Changyan Shi}
\author[2]{Xiyao Chen}
\author[2]{Guangxiong Zhang}
\author[2,3]{Chengfa Wu}
\author[1,*]{Bao-Feng Feng}

\affil[1]{School of Mathematical and Statistical Sciences, The University of Texas Rio Grande Valley, Edinburg, TX 78539, USA}
\affil[2]{Institute for Advanced Study, Shenzhen University, Shenzhen, 518060, People's Republic of China}
\affil[3]{School of Mathematical Sciences, Shenzhen University, Shenzhen, 518060, People's Republic of China}

\affil[*]{Corresponding author: \href{mailto:baofeng.feng@utrgv.edu}{baofeng.feng@utrgv.edu}}

\date{} 

\maketitle

\abstract{In this paper, we derive general bright-dark soliton solutions to the coupled Sasa-Satsuma (CSS) equation using the Kadomtsev-Petviashvili (KP) reduction method. Since the CSS equation is a special case of the four-component Hirota equation, our approach begins with the construction of two-bright-two-dark soliton solutions for the four-component Hirota equation. By imposing specific parameter constraints, these solutions are subsequently reduced to the bright-dark soliton solutions of the CSS equation. Finally, the dynamical behaviors of the one- and two-bright-dark soliton solutions are thoroughly analyzed and illustrated.}

\section{Introduction}
Solitons have been extensively studied in both mathematics and physics due to their unique property of preserving shape and velocity after interactions with other solitary waves of the same type. The first numerical observation of soliton behavior in media governed by the Korteweg-de Vries (KdV) equation \cite{Korteweg1895Vries} was carried out by Zabusky and Kruskal \cite{Zabusky1965Kruskal}. A few years later, Gardner, Greene, Kruskal, and Miura \cite{Gardner1967GreeneKruskalMiura} introduced the inverse scattering transform (IST) method, which provided exact soliton solutions to the KdV equation with rapidly decaying initial conditions. Lax \cite{Lax1968} extended this approach, demonstrating that the IST method could be applied to any nonlinear wave equation equivalent to the compatibility condition of two linear problems, now known as the Lax pair. Following this, Zakharov and Shabat \cite{Zakharov1971Shabat} showed that the nonlinear Schr{\"o}dinger (NLS) equation \cite{Benney1967Newell} also had a Lax pair, enabling its solution  through the IST method. Subsequently, Ablowitz, Kaup, Newell, and Segur (AKNS) \cite{Ablowitz1973KaupNewellSegur,Ablowitz1974KaupNewellSegur} discovered a broad class of integrable equations, including the sine-Gordon equation and the modified Korteweg-de Vries (mKdV) equation, that could also be solved or linearized using the IST method. These foundational contributions sparked widespread research into soliton solutions of nonlinear evolution equations.

The NLS equation governs the evolution of weakly nonlinear wave packets in various media \cite{malomed2005nonlinear}, including nonlinear optical fibers \cite{hasegawa1973transmission}, planar waveguides \cite{fibich2015nonlinear}, Bose-Einstein condensates \cite{dalfovo1999theory,pitaevskii2003bose}, plasmas \cite{zakharov1972collapse,kato2005nonlinear} and water \cite{Benney1967Newell}. A natural extension of this equation is the coupled NLS equation, a multi-component model introduced by Manakov \cite{manakov1974theory}, given by
\begin{align}\label{Manakov}
    \begin{split}
        \i q_{1,t} &= q_{1,xx} - 2 \left(\sigma_1 |q_1|^2 + \sigma_2 |q_2|^2\right) q_1,\\
        \i q_{2,t} &= q_{2,xx} - 2 \left(\sigma_1 |q_1|^2 + \sigma_2 |q_2|^2\right) q_2,
    \end{split}
\end{align}
where \(\sigma_1, \sigma_2\) are real coefficients, and the system reduces to the standard NLS equation if \(q_2 = q_1\), \(\sigma_2 = \sigma_1\). Known as the Manakov system, \cref{Manakov} has been widely applied in nonlinear optics, particularly to model optical pulses in birefringent optical fibers \cite{maimistov2013nonlinear, agrawal2000nonlinear}, with the two components representing different light polarizations \cite{gelash2023vector}. 
Like the NLS equation, the Manakov system has focusing and defocusing cases. 
In the focusing regime, where \(\sigma_1 < 0, \sigma_2 < 0 \), the Manakov system \eqref{Manakov} supports bright-bright solitons \cite{manakov1974theory,kanna2003exact}, which can undergo shape-changing collisions and exhibit breather-like oscillations. In the defocusing regime, where 
\(\sigma_1 > 0, \sigma_2 > 0\), the system admits dark-dark \cite{sheppard1997polarized} and bright-dark solitons \cite{radhakrishnan1995integrability}. 
Additionally, the system can exhibit a mixed regime where \(\sigma_1\) and \(\sigma_2\) 
have opposite signs, resulting in the existence of bright-bright \cite{wang2010integrable,kanna2006soliton}, dark-dark \cite{ohta2011general}, and bright-dark soliton solutions \cite{vijayajayanthi2008bright}. 
Both the NLS equation and the Manakov system are considered integrable, possessing an infinite number of conserved quantities \cite{hitchin2013integrable}. 
Consequently, various methods have been developed to obtain their solutions, including Darboux transformation \cite{matveev1991darboux}, Hirota's bilinear method \cite{hirota2004direct}, and the Kadomtsev–Petviashvili (KP) reduction method \cite{jimbo1983solitons,sato1989kp}.

In the study of nonlinear optical waves, Sasa and Satsuma \cite{sasa1991new} discovered an integrable higher order NLS equation based on the work of Hasegawa and Kodama \cite{Kodama1987Hasegawa}. This generalized NLS equation, now called the Sasa-Satsuma (SS) equation, incorporates higher-order effects such as third-order dispersion, self-steepening, and stimulated Raman scattering. Motivated by the integrability and the concrete physical relevance of the SS equation, extensive research has focused on exploring its properties. For example, soliton solutions to the SS equation have been derived in several studies \cite{sasa1991new, Mihalache1993TornerMoldoveanuPanoiuTruta, gilson2003sasa, ohta2010dark, Xu_2015, guo2019darboux, yang2019high}. Additionally, its rogue wave solutions have been derived by various methods, such as Darboux transformation \cite{Chen2013, MU2020QinRogerNail, Mu2016Qin, Ling2016} and the KP reduction method \cite{feng2022higher,wu2022general}.
To describe the propagation of optical pulses in birefringent fibers, the Sasa-Satsuma equation has been extended into a coupled form \cite{PhysRevLett.80.1425}, given by
\begin{align}
    u_{1,t}&=u_{1,xxx} - 6 \left(\epsilon_1|u_1|^2 + \epsilon_2|u_2|^2\right) u_{1,x} - 3 u_1 \left(\epsilon_1|u_1|^2 + \epsilon_2|u_2|^2\right)_x,\label{css_1}\\
    u_{2,t}&=u_{2,xxx} - 6 \left(\epsilon_1|u_1|^2 + \epsilon_2|u_2|^2\right) u_{2,x} - 3 u_2 \left(\epsilon_1|u_1|^2 + \epsilon_2|u_2|^2\right)_x,\label{css_2}
\end{align}
which was originally proposed by Porsezian, Shanmugha, and Mahalingam in \cite{porsezian1994coupled}.  Multiple studies are devoted to investigating the exact solutions and their dynamics for the above coupled Sasa-Satsuma equation, including bright-bright and bright-dark soliton solutions \cite{LU20143969,PhysRevE.87.032913,liu2018vector,liu2023riemann,liu2018dark}, as well as dark-dark soliton and rogue wave solutions \cite{zhang2017binary,zhao2014localized}.
Very recently, by employing the KP reduction method, the authors derived the general dark-dark soliton, breather, and rogue wave solutions to the CSS equation \cite{zhang2024dark,zhang2025rogue}.

Early investigations into the Manakov system revealed that bright-bright \cite{manakov1974theory,kanna2003exact} and bright-dark \cite{radhakrishnan1995integrability,vijayajayanthi2008bright} soliton solutions exhibit shape-changing collisions and breather-like oscillations. As the coupled Sasa-Satsuma (CSS) equation is a higher-order extension of the Manakov system, similar behavior of the bright-bright soliton solution was observed for the CSS equation  \cite{PhysRevE.87.032913,LU20143969,liu2018vector,ma2023sasa,ma2024soliton}.
However, for the bright--dark soliton solutions of the CSS equation, the existing literature remains incomplete. 
In \cite{liu2018dark}, only breather-type oscillatory behaviors were reported, and the possibility of shape-changing collisions between bright--dark solitons was not addressed. 
Moreover, although first- and second-order bright--dark soliton solutions were obtained via Darboux transformation, a general formulation of bright--dark soliton solutions has not yet been established, to the best of our knowledge.
These observations naturally lead to the following questions:
\begin{itemize}
     
\item [(i)] How can we construct the general bright-dark soliton solution to the CSS equation using the Kadomtsev-Petviashvili reduction method? 



\item [(ii)] Do the bright-dark soliton solutions to the CSS equation exhibit shape-changing collisions similar to those observed in the Manakov system? 
\end{itemize}
The present work addresses these issues and provides a unified treatment of the CSS equation under mixed boundary conditions. 
Our main contribution is the derivation of explicit general bright--dark soliton solutions in a determinant form. 



We should point out that there exists a direct connection between the Sasa-Satsuma or coupled Sasa-Satsuma equations \eqref{css_1}-\eqref{css_2} and the vector Hirota equation
\begin{align}  
    v_{j,t}&=v_{j,xxx} - 3 \left(\sum_{k=1}^n c_k|v_k|^2\right) v_{j,x} - 3 v_j \left(\sum_{k=1}^n c_k v_k^* v_{k,x}\right),\label{4cmkdv}
\end{align} 
where \(j = 1,2, \cdots,n\). The two-component Hirota equation can be reduced to the Sasa-Satsuma equation, through which both the bright- and dark-soliton solutions are rederived \cite{shi2025general}.  For $n=4$, if we impose the following conditions
\begin{equation}\label{cond_4cmkdv_to_css}
    \epsilon_1 = c_1 = c_2, \quad \epsilon_2 = c_3 = c_4, \quad u_1 = v_1 = v_2^*, \quad u_2 = v_3 = v_4^*,
\end{equation}
then the four-component Hirota equation reduces to the
  CSS equation \eqref{css_1}-\eqref{css_2}. This observation plays a crucial role for us in deriving bright-dark soliton solutions to the CSS equation. Starting from the matrix AKNS hierarchy, several generalized Sasa-Satsuma-type equations have been derived and their bright soliton solutions have been investigated \cite{ma2025combined,ma2026reduced}.
 Due to the complex bilinear form to the CSS equation \cite{zhang2024dark}, the relation between the CSS equation and the four-component Hirota equation
has inspired us to explore an alternative method for deriving the bright-dark soliton solutions to the CSS equation. 
The derivation relies essentially on the KP reduction technique. By formulating the problem within the KP hierarchy, the CSS equation is reduced to a compatible set of bilinear equations, which are related to the (Hirota) bilinear forms of \eqref{4cmkdv} with $n=4$ through the reduction condition \eqref{cond_4cmkdv_to_css}. This four-component Hirota structure is the central mechanism behind the construction of the tau functions and the multi-soliton solutions. In this sense, the KP reduction provides the theoretical bridge, while the four-component Hirota equation constitutes the core computational framework of the paper.


The remainder of the paper is organized as follows. In \cref{sect:main_results}, we present our results on bright-dark soliton solutions to the CSS equation \eqref{css_1}-\eqref{css_2} in \cref{thm:bd_css}. In \cref{sect:bd}, we provide the proof of \cref{thm:bd_css} by first deriving the solutions to the four-component Hirota equation \eqref{4cmkdv} and then reducing it to the solution to the CSS equation \eqref{css_1}-\eqref{css_2}. Finally, the article concludes in \cref{sect:conclusion}.

\section{Main results}\label{sect:main_results}
As discussed in the previous section, the CSS equation \eqref{css_1}-\eqref{css_2} is a special case to the four-component Hirota equation \eqref{4cmkdv}, obtained through the reduction condition \eqref{cond_4cmkdv_to_css}. To derive the two-bright-two-dark soliton solution for \eqref{4cmkdv} with $n=4$, we begin by bilinearizing the equation using Hirota's \(D\)-operator \cite{hirota2004direct} defined by
\begin{eqnarray*}\label{doperator}
D_x^mD_t^nf\cdot g=\left.\left(\frac{\partial}{\partial x}-\frac{\partial}{\partial {x'}}\right)^m\left(\frac{\partial}{\partial t}-\frac{\partial}{\partial{t'}}\right)^n
[f(x,t)g(x',t')]\right|_{x'=x,t'=t}.
\end{eqnarray*}
By the transformation
\begin{align}\label{bbdd_transformation}
        \begin{split}
            &v_1 = \frac{g_1}{f}\exp\left(\i (-\omega_1 t) \right), \quad v_2 = \frac{g_2}{f}\exp\left(\i (-\omega_2 t) \right), \\
            &v_3 = \rho_3\frac{h_3}{f} \exp\left(\i (\alpha_3 x - \omega_3 t)\right), \quad v_4 = \rho_4\frac{h_4}{f} \exp\left(\i (\alpha_4 x - \omega_4 t)\right),
        \end{split}
    \end{align}
where
    \begin{equation}\label{bbdd_disp}
        \begin{split}
            \omega_1 &= \omega_2 = 3 \rho_3^2 c_3 \alpha_3 + 3 \rho_4^2 c_4 \alpha_4, \\ 
            \omega_3 &= \alpha_3^3 + 3(c_3 \rho_3^2 + c_4 \rho_4^2) \alpha_3 + 3(c_3 \rho_3^2 \alpha_3 + c_4 \rho_4^2 \alpha_4),\\
            \omega_4 &= \alpha_4^3 + 3(c_3 \rho_3^2 + c_4 \rho_4^2) \alpha_4 + 3(c_3 \rho_3^2 \alpha_3 + c_4 \rho_4^2 \alpha_4),
        \end{split}
\end{equation}
and applying the equality \(a D_x b \cdot f - b D_x a \cdot f = f D_x b \cdot a\),
the first component of equation \eqref{4cmkdv} can be simplified into
\begin{align*}
& f^2 \left(D_x^3 - D_t - 6(c_3 \rho_3^2 + c_4 \rho_4^2) D_x+ 3\i (\rho_3^2 c_3 \alpha_3 + \rho_4^2 c_4 \alpha_4) \right) g_1 \cdot f \\
& - 3 (D_x g_1 \cdot f) \left[\left(D_x^2 - 2 c_3 \rho_3^2 - 2 c_4 \rho_4^2\right)f\cdot f  + 2c_1 |g_1|^2 + 2c_2 |g_2|^2 + 2\rho_3^2 c_3 |h_3|^2 + 2\rho_4^2 c_4 |h_4|^2\right] \\
& +3 c_2 g_2^* f D_x g_1 \cdot g_2 +3 \rho_3^2 c_3 h_3^* f\left(D_x g_1 \cdot h_3 - \i g_1 h_3 \alpha_3\right) +3 \rho_4^2 c_4 h_4^* f\left(D_x g_1 \cdot h_4 - \i g_1 h_4 \alpha_4\right)= 0.
\end{align*}
Introducing the auxiliary functions \(s_{12},r_{13},r_{14}\), we obtain the bilinear form for the first component of \eqref{4cmkdv}
\begin{align*}
& \left(D_x^3 - D_t - 6(c_3 \rho_3^2 + c_4 \rho_4^2) D_x+ 3\i (\rho_3^2 c_3 \alpha_3 + \rho_4^2 c_4 \alpha_4) \right) g_1 \cdot f \\
&\qquad = -3 c_2 s_{12} g_2^* + 3 \i \rho_3^2 c_3\alpha_3 r_{13} h_3^* + 3 \i \rho_4^2 c_4\alpha_4 r_{14} h_4^*,\\
& \left(D_x^2 - 2 c_3 \rho_3^2 - 2 c_4 \rho_4^2\right)f\cdot f  + 2c_1 |g_1|^2 + 2c_2 |g_2|^2 + 2\rho_3^2 c_3 |h_3|^2 + 2\rho_4^2 c_4 |h_4|^2 = 0,\\
& D_x g_1 \cdot g_2 = s_{12} f,\\
& D_x g_1 \cdot h_3 - \i g_1 h_3 \alpha_3 = -\i\alpha_3 r_{13} f,\\
& D_x g_1 \cdot h_4 - \i g_1 h_4 \alpha_4 = -\i\alpha_4 r_{14} f.
\end{align*}
Similarly, the third component of \eqref{4cmkdv} is transformed as
{\allowdisplaybreaks\begin{align*}
& \left[D_x^3 - D_t + 3\i \alpha_3 D_x^2 - 3\left(\alpha_3^2 + 2 c_3 \rho_3^2 + 2 c_4 \rho_4^2\right) D_x \right. \\ 
& \left. -3\i (c_3 \rho_3^2 + c_4 \rho_4^2) \alpha_3 + 3\i (c_3 \rho_3^2 \alpha_3 + c_4 \rho_4^2 \alpha_4)\right] h_3 \cdot f\\
& \qquad = -3\i c_1 \alpha_3  r_{31} g_1^* -3\i c_2 \alpha_3 r_{32} g_2^* - 3\i c_4 (\alpha_3-\alpha_4) \rho_4^2 r_{34} h_4^*,\\
& \left(D_x^2 - 2 c_3 \rho_3^2 - 2 c_4 \rho_4^2\right)f\cdot f  + 2c_1 |g_1|^2 + 2c_2 |g_2|^2 + 2\rho_3^2 c_3 |h_3|^2 + 2\rho_4^2 c_4 |h_4|^2 = 0,\\
& D_x h_3 \cdot g_1 + \i \alpha_3 g_1 h_3 = \i\alpha_3 r_{31} f,\\
& D_x h_3 \cdot g_2 + \i \alpha_3 g_2 h_3 = \i\alpha_3 r_{32} f,\\
& D_x h_3 \cdot h_4 + \i (\alpha_3 - \alpha_4) h_3 h_4 = \i (\alpha_3 - \alpha_4) r_{34} f,
\end{align*}}
where \(s_{31},s_{32},s_{34}\) are auxiliary functions.
The bilinear forms for the second and fourth components of \eqref{4cmkdv} can be derived in a similar manner. The results are summarized in the following lemma.
\begin{lemma}\label{bilinear_4cmkdv}
    Under the transformation \eqref{bbdd_transformation}, the four-component Hirota equation \eqref{4cmkdv} has the following bilinear form
    {\allowdisplaybreaks
    \begin{align}
    & \left(D_x^3 - D_t - 6(c_3 \rho_3^2 + c_4 \rho_4^2) D_x+ 3\i (\rho_3^2 c_3 \alpha_3 + \rho_4^2 c_4 \alpha_4) \right) g_1 \cdot f \nonumber \\
    &\qquad = -3 c_2 s_{12} g_2^* + 3 \i \rho_3^2 c_3\alpha_3 r_{13} h_3^* + 3 \i \rho_4^2 c_4\alpha_4 r_{14} h_4^*,\label{4H_BL_1}\\
    & \left(D_x^3 - D_t - 6(c_3 \rho_3^2 + c_4 \rho_4^2) D_x+ 3\i (\rho_3^2 c_3 \alpha_3 + \rho_4^2 c_4 \alpha_4) \right) g_2 \cdot f \nonumber \\
    &\qquad = -3 c_1 s_{21} g_1^* + 3 \i \rho_3^2 c_3\alpha_3 r_{23} h_3^* + 3 \i \rho_4^2 c_4\alpha_4 r_{24} h_4^*,\label{4H_BL_2}\\
    & \left[D_x^3 - D_t + 3\i \alpha_3 D_x^2 - 3\left(\alpha_3^2 + 2 c_3 \rho_3^2 + 2 c_4 \rho_4^2\right) D_x\right. \nonumber \\
    & \qquad \left.-3\i (c_3 \rho_3^2 + c_4 \rho_4^2) \alpha_3 + 3\i (c_3 \rho_3^2 \alpha_3 + c_4 \rho_4^2 \alpha_4)\right] h_3 \cdot f \nonumber \\
    & \qquad = -3\i c_1 \alpha_3 r_{31} g_1^* -3\i c_2 \alpha_3 r_{32} g_2^* - 3\i c_4 (\alpha_3-\alpha_4) \rho_4^2 r_{34} h_4^*,\label{4H_BL_3}\\
    & \left[D_x^3 - D_t + 3\i \alpha_3 D_x^2 - 3\left(\alpha_3^2 + 2 c_3 \rho_3^2 + 2 c_4 \rho_4^2\right) D_x  \right. \nonumber \\
    & \qquad \left.-3\i (c_3 \rho_3^2 + c_4 \rho_4^2) \alpha_3+ 3\i (c_3 \rho_3^2 \alpha_3 + c_4 \rho_4^2 \alpha_4)\right] h_4 \cdot f \nonumber \\
    & \qquad = -3\i c_1 \alpha_4  r_{41} g_1^* -3\i c_2 \alpha_4 r_{42} g_2^* - 3\i c_4 (\alpha_4-\alpha_3) \rho_4^2 r_{43} h_3^*,\label{4H_BL_4}\\
    & \left(D_x^2 - 2 c_3 \rho_3^2 - 2 c_4 \rho_4^2\right)f\cdot f  + 2c_1 |g_1|^2 + 2c_2 |g_2|^2 + 2\rho_3^2 c_3 |h_3|^2 + 2\rho_4^2 c_4 |h_4|^2 = 0,\label{4H_BL_5}\\
    & D_x g_1 \cdot g_2 = s_{12} f,\label{4H_BL_6}\\
    & D_x g_1 \cdot h_3 - \i g_1 h_3 \alpha_3 = -\i\alpha_3 r_{13} f,\label{4H_BL_7}\\
    & D_x g_1 \cdot h_4 - \i g_1 h_4 \alpha_4 = -\i\alpha_4 r_{14} f,\label{4H_BL_8}\\
    & D_x g_2 \cdot h_3 - \i g_2 h_3 \alpha_3 = -\i\alpha_3 r_{23} f,\label{4H_BL_9}\\
    & D_x g_2 \cdot h_4 - \i g_2 h_4 \alpha_4 = -\i\alpha_4 r_{24} f,\label{4H_BL_10}\\
    & D_x h_3 \cdot h_4 + \i (\alpha_3 - \alpha_4) h_3 h_4 = \i (\alpha_3 - \alpha_4) r_{34} f,\label{4H_BL_11}
    \end{align}
    }where \(s_{12} = -s_{21}\), \(r_{jk} = r_{kj}\) for \(j,k = 1,2,3,4\). 
\end{lemma}

Next, by employing the method of KP reduction, we can obtain the two-bright-two-dark soliton solution to the four-component Hirota equation \eqref{4cmkdv}. The detailed proofs of the following theorems are provided in \cref{sect:bd}.

\begin{theorem}\label{thm:2b2d_4cmkdv}
    The four-component Hirota equation \eqref{4cmkdv} admits the following two-bright-two-dark soliton solution
    \begin{align*}
        &v_1 = \frac{g_1}{f}\exp\left(-3\i (\rho_3^2 c_3 \alpha_3 +  \rho_4^2 c_4 \alpha_4)t \right), \quad v_2 = \frac{g_2}{f}\exp\left(-3\i (\rho_3^2 c_3 \alpha_3 + \rho_4^2 c_4 \alpha_4)t \right), \\
        &v_3 = \rho_3\frac{h_3}{f} \exp\left(\i (\alpha_3 x - (\alpha_3^3 + 3(c_3 \rho_3^2 + c_4 \rho_4^2) \alpha_3 + 3(c_3 \rho_3^2 \alpha_3 + c_4 \rho_4^2 \alpha_4)) t)\right), \\
        &v_4 = \rho_4\frac{h_4}{f} \exp\left(\i (\alpha_4 x - (\alpha_4^3 + 3(c_3 \rho_3^2 + c_4 \rho_4^2) \alpha_4 + 3(c_3 \rho_3^2 \alpha_3 + c_4 \rho_4^2 \alpha_4)) t)\right),
    \end{align*}
    where 
    \begin{align*}
        &f = \begin{vmatrix}
            M_{0,0}
        \end{vmatrix},\quad 
        g_1 = \begin{vmatrix}
            M_{0,0} & \Phi \\
            -\left(\bar{\Psi}\right)^T & 0
        \end{vmatrix},\quad
        g_2 = \begin{vmatrix}
            M_{0,0} & \Phi \\
            -\left(\bar{\Upsilon}\right)^T & 0
        \end{vmatrix},\quad
        h_3 = \begin{vmatrix}
            M_{1,0}
        \end{vmatrix},\quad
        h_4 = \begin{vmatrix}
            M_{0,1}
        \end{vmatrix}.
    \end{align*}
    Here, \(M_{k,l}, k,l = 0,1\), is an \(N \times N\) matrix, \(\Phi\), \(\bar{\Psi}\) and \(\bar{\Upsilon}\) are vectors, whose entries are defined as
    \begin{align*}
        &m_{ij}^{k,l}=\frac{1}{p_i+p_j^*}\left(-\frac{p_i - \i \alpha_3}{p_j^* + \i \alpha_3}\right)^k \left(-\frac{p_i - \i \alpha_4}{p_j^* + \i \alpha_4}\right)^l e^{\xi_i+\xi_j^*}+ \frac{C_i^ * C_j}{q_i+q_j^*} + \frac{D_i^ * D_j}{r_i+r_j^*} ,\\
        &\Phi = \left(e^{\xi _{1}}, e^{\xi _2},\ldots, e^{\xi _{N}}\right)^T,\quad 
        \bar{\Psi}=\left(C_1, C_2,\ldots, C_N \right)^T,\quad
        \bar{\Upsilon}=\left(D_1, D_2,\ldots, D_N \right)^T, 
        \\
        &\xi_i=p_i (x - 3(c_3\rho_3^2+c_4\rho_4^2) t) + p_i^3 t + \xi_{i0}, \\
        &q_i = \frac{1}{c_1} \left(-p_i + \frac{c_3 \rho_3^2}{p_i - \i\alpha_3} + \frac{c_4 \rho_4^2}{p_i - \i \alpha_4}\right), \quad
        r_i = \frac{1}{c_2} \left(-p_i + \frac{c_3 \rho_3^2}{p_i - \i \alpha_3} + \frac{c_4 \rho_4^2}{p_i - \i \alpha_4}\right),
    \end{align*}
    and \(\rho_3,\rho_4, \alpha_3,\alpha_4 \in \mathbb{R}\), \(p_i, \xi_{i,0}, C_i,D_i \in \mathbb{C}\).
\end{theorem}


Finally, by imposing the condition that 
\(v_1\) and \(v_3\)  are complex conjugates of 
\(v_2\) and \(v_4\), respectively, we ensure that the reduction condition \eqref{cond_4cmkdv_to_css} is satisfied. This leads to the bright-dark soliton solution for the coupled Sasa-Satsuma equation \eqref{css_1}-\eqref{css_2}.

\begin{theorem}\label{thm:bd_css}
    The coupled Sasa–Satsuma equation \eqref{css_1}-\eqref{css_2} has the following bright-dark soliton solution
    \begin{align}\label{css_bd_sln}
        u_1 = \frac{g}{f}, \quad u_2 = \rho \frac{h}{f} \exp\left(\i (\alpha x - (\alpha^3 + 6\epsilon_2 \rho^2 \alpha) t)\right)
    \end{align}
    where 
    \begin{align*}
        &f = \begin{vmatrix}
            M_{0}
        \end{vmatrix},\quad g = \begin{vmatrix}
            M_{0} & \Phi \\
            -\left(\bar{\Psi}\right)^T & 0
        \end{vmatrix},\quad
        h = \begin{vmatrix}
            M_{1}
        \end{vmatrix}.
    \end{align*}
    Here, \(M_k, k = 0,1\), is an \(N \times N\) matrix, \(\Phi\) and \(\bar{\Psi}\) are vectors, whose entries are defined as
    \begin{align*}
        &m_{ij}^{k}=\frac{1}{p_i+p_j^*}\left(-\frac{p_i - \i \alpha}{p_j^* + \i \alpha}\right)^k e^{\xi_i+\xi_j^*}+d_{i,j},\\
        &\Phi = \left(e^{\xi _{1}}, e^{\xi _2},\ldots, e^{\xi _{N}}\right)^T,
        \quad \bar{\Psi}=\left(C_1, C_2,\ldots, C_N \right)^T, 
        \quad \xi_i=p_i (x - 6\epsilon_2 \rho^2 t) + p_i^3 t + \xi_{i0},\\
        &d_{i,j}=\frac{\epsilon_1 }{(p_i+p_j^*)}\left(C_i^* C_j + C_{N+1-i} C_{N+1-j}^*\right)\left(-1+ \dfrac{2\epsilon_2 \rho^2 (p_i p_j^* + \alpha^2)}{(p_i^2 + \alpha^2)((p_j^*)^2 + \alpha^2)} \right)^{-1},
    \end{align*}
    and \(\rho, \alpha \in \mathbb{R}\), \(p_i, \xi_{i,0}, C_i \in \mathbb{C}\), with \(p_i,\xi_{i,0}\) satisfying the complex conjugate restrictions 
    \begin{equation}\label{cmp_conj_rest}
        p_i = p_{N+1-i}^*,\quad \xi_{i,0} = \xi_{N+1-i,0}^*.
    \end{equation}
\end{theorem}

\begin{remark}\label{lemma:dij}
    Under the parameter restrictions \eqref{cmp_conj_rest}, the following identities hold for \(d_{i,j}\)
    \begin{align*}
        d_{j,i}^* = d_{i,j},  \quad
        d_{N+1-j,N+1-i} = d_{i,j}.
    \end{align*}
    These identities can facilitate the simplification of the solution expressions for further analysis. 
\end{remark}

\section{Dynamics of the bright-dark soliton solutions}\label{sect:dyn}
In this section, we investigate the dynamic behaviors for the solutions obtained in \cref{thm:bd_css}. For the case of $N=1$, we provide the expression of the solution and illustrated the profiles. For the case of $N=2$,  the analytical expression of the solution is presented. In particular, we analyze and identify the condition under which the solution is either breather or two-hump soliton. For the cases of $N=3$ and $N=4$, the collision behaviors between two solitons are analyzed. Specifically, for $N=3$, collision between a soliton and a breather is identified, while  collisions between breathers are observed for $N=4$. Furthermore, inelastic collisions are presented for both $N=3$ and $N=4$ cases.
\subsection{One bright-dark soliton for $N=1$}
For case $N=1$ in \cref{thm:bd_css}, one bright-dark soliton solution can be derived. By introducing the following parameters,
\begin{equation}\label{N=1 parameters}
\begin{aligned}
    G &= -\frac{2C_1^2(\alpha^2 +p_1^2)\epsilon_1}{\alpha^2+p_1^2-2\rho^2\epsilon_2}, \quad H = \frac{p_1 -\mathrm{i}\alpha}{p_1 + \mathrm{i}\alpha}, \\
    \xi_1 &= p_1x + (p_1^3 - 6p_1\rho^2\epsilon_2)t + \xi_{1,0},\quad \theta = \alpha x - \left(\alpha^3-6\alpha \rho^2\epsilon_2\right)t,
\end{aligned}
\end{equation}
the solution \(u_1, u_2\) can be rewritten as
\begin{equation}\label{b-d N=1 solution}
\begin{aligned}
    u_1 & = \frac{2C_1p_1 \exp(-\xi_1)}{1 + G\exp(-2\xi_1)} = \frac{C_1 p_1}{\sqrt{G}} \sech\left(\xi_1-\log \sqrt{G}\right), \\
    u_2 & = -\rho \exp{\left(\mathrm{i}\theta\right)} \frac{H-G\exp(-2\xi_1)}{1+G\exp(-2\xi_1)}  = -\rho\frac{\exp{\left(\mathrm{i}\theta\right)}}{2}\left((H+1) \tanh{(\xi_1-\log\sqrt{G})} + H -1\right) ,
\end{aligned}
\end{equation}
where $\rho, \alpha, p_1, \xi_{1,0} \in \mathbb{R}$. It is noteworthy that the component \(u_1\) represents a bright soliton, while \(u_2\) corresponds to a dark soliton. Both solitons propagate with the same speed given by \(-(p_1^3 - 6p_1\rho^2\epsilon_2)/p_1 = 6p_1\rho^2\epsilon_2 - p_1^2\).

To ensure the regularity of the bright-dark soliton solution, i.e., \(1 + G\exp(-2\xi_1) \neq 0\), one can find that \(G > 0\) is a necessary condition. Hence, when $\epsilon_1 = 1$ and $\epsilon_2 = -1$, we have $G = -2C_1^2(\alpha^2 + p_1^2)/(\alpha^2 + p_1^2 + 2\rho^2) < 0$, which leads to a singular solution. 

The first order bright-dark soliton solution \eqref{b-d N=1 solution} is illustrated in Fig. \ref{fig:one bright-dark soliton for N=1, case1} with the following parameter choices,
\begin{equation}
    \alpha=2,\quad \rho=1, \quad \epsilon_1=-1,\quad \epsilon_2=1, \quad p_1=1,\quad C_1=1.
\end{equation}

\begin{figure}[H]
    \centering
\subfigure[]{
        \includegraphics[width=30mm]{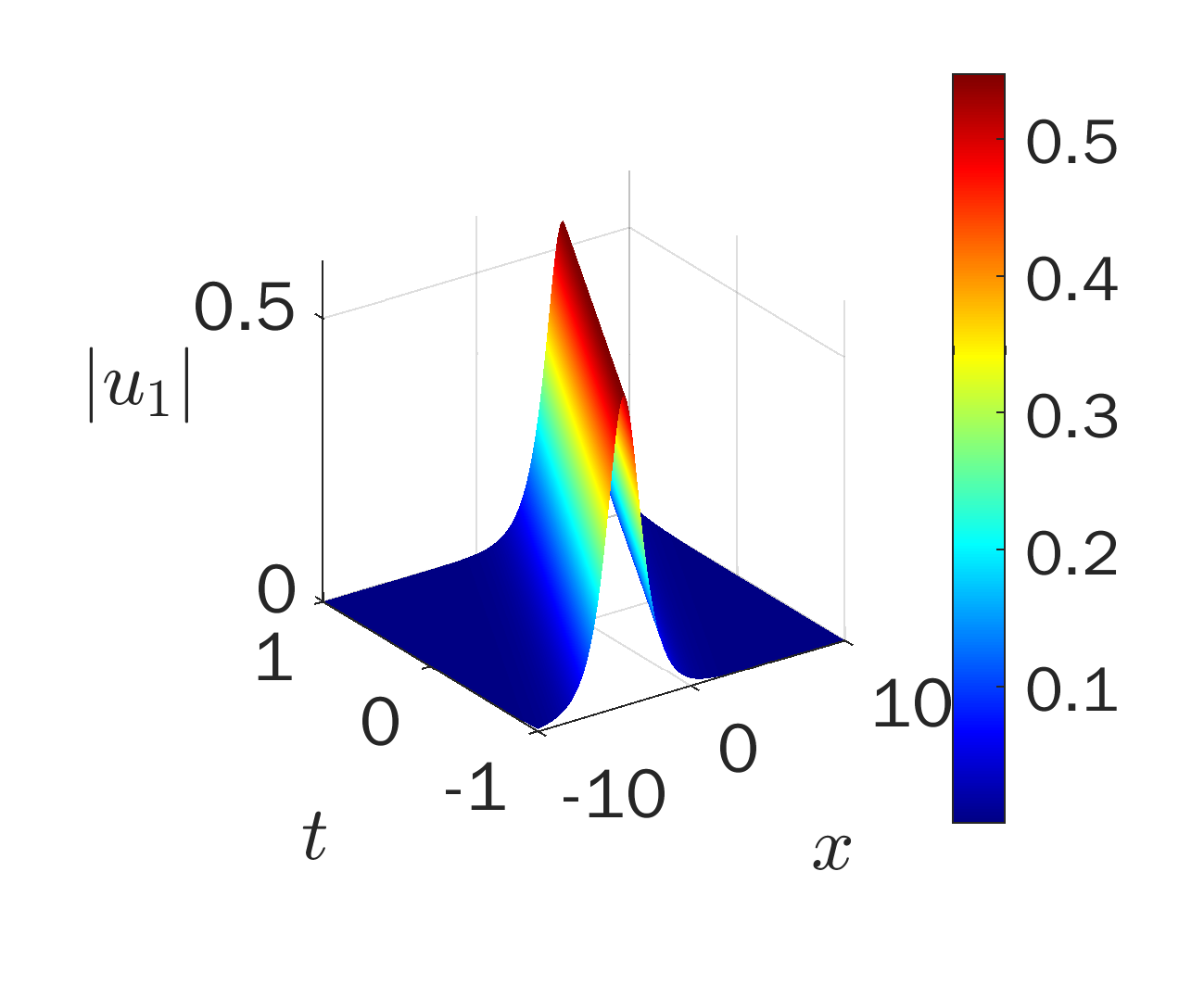}}
\subfigure[]{ 
        \includegraphics[width=30mm]{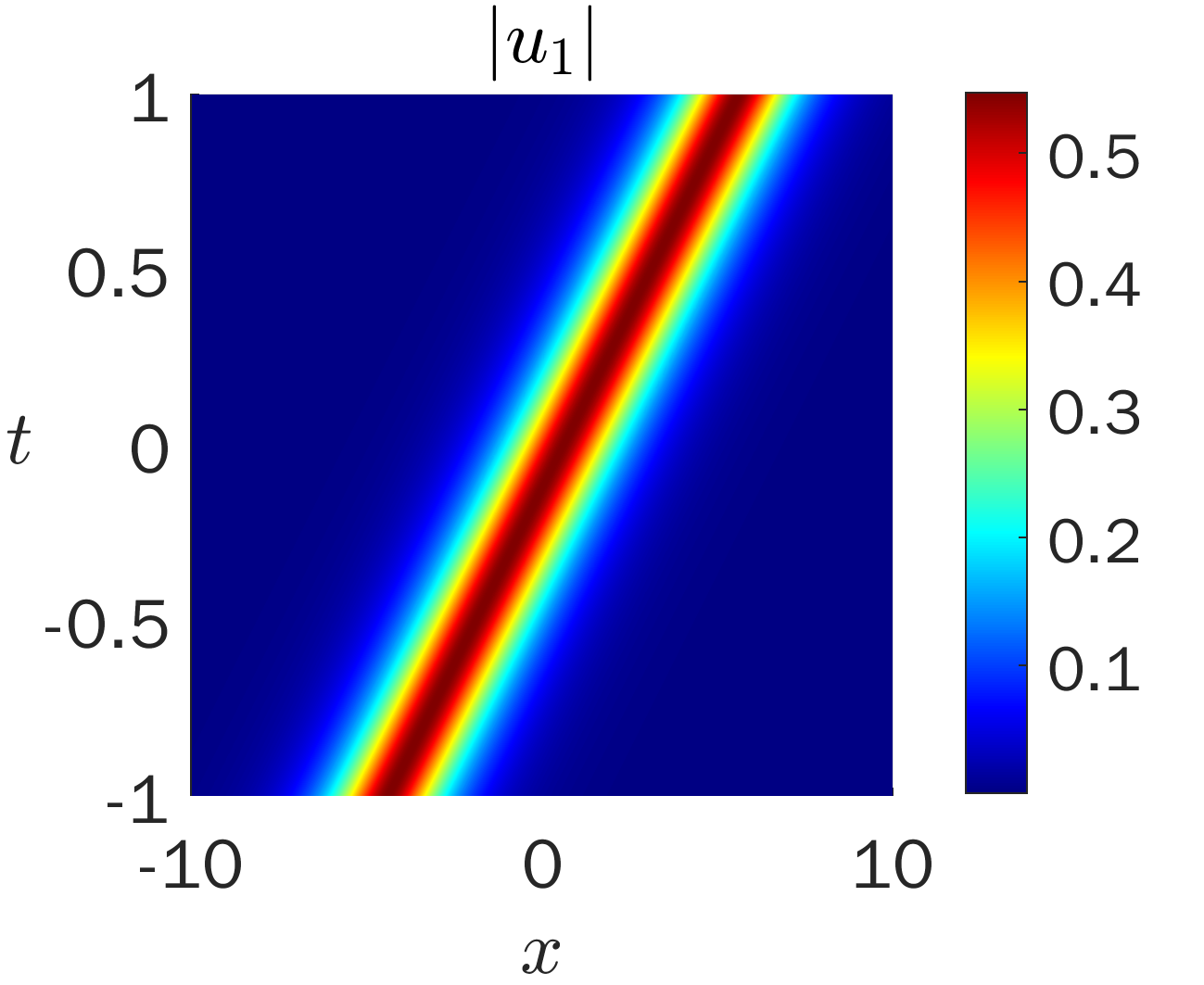}}
\subfigure[]{
        \includegraphics[width=30mm]{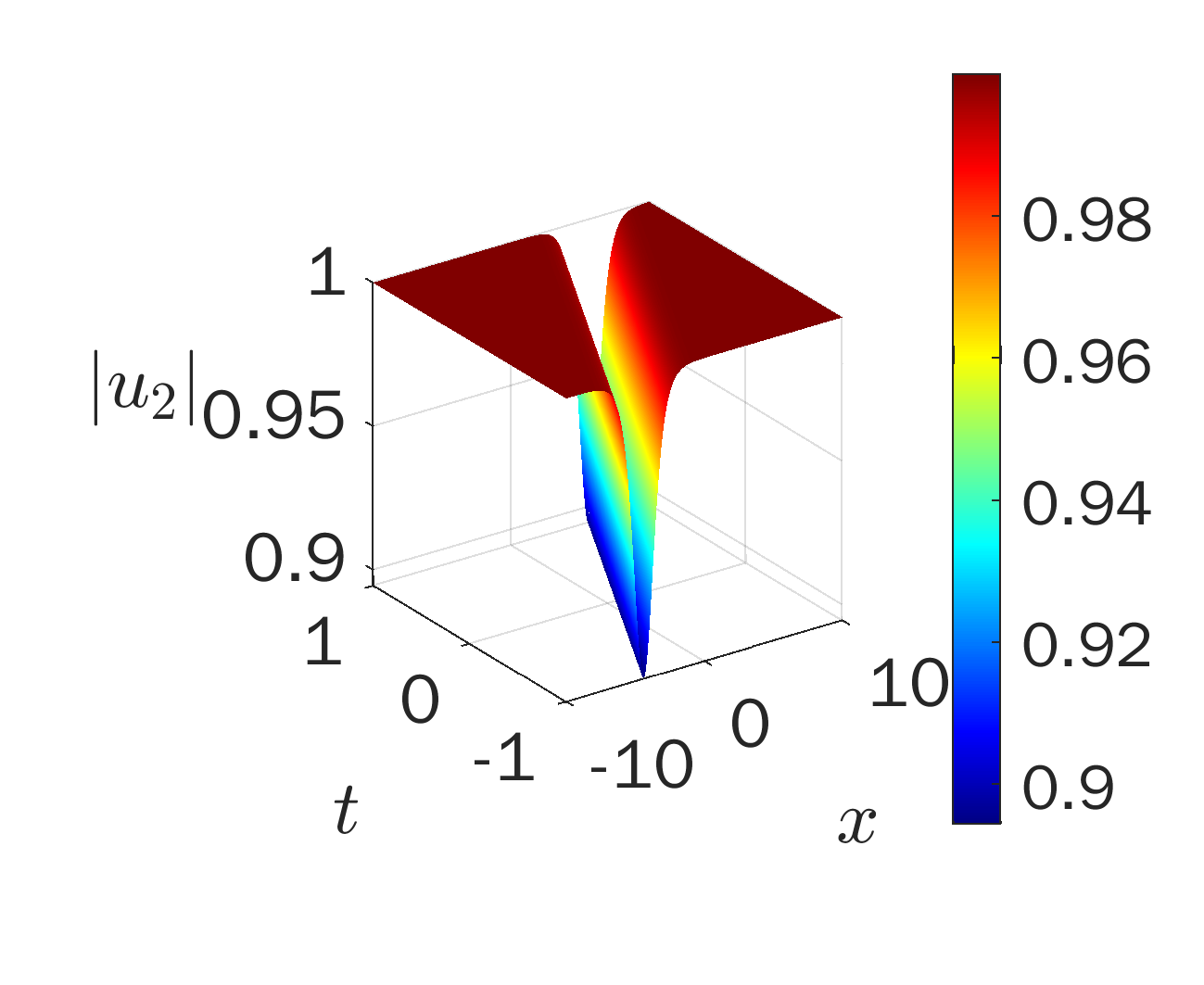}}
\subfigure[]{%
        \includegraphics[width=30mm]{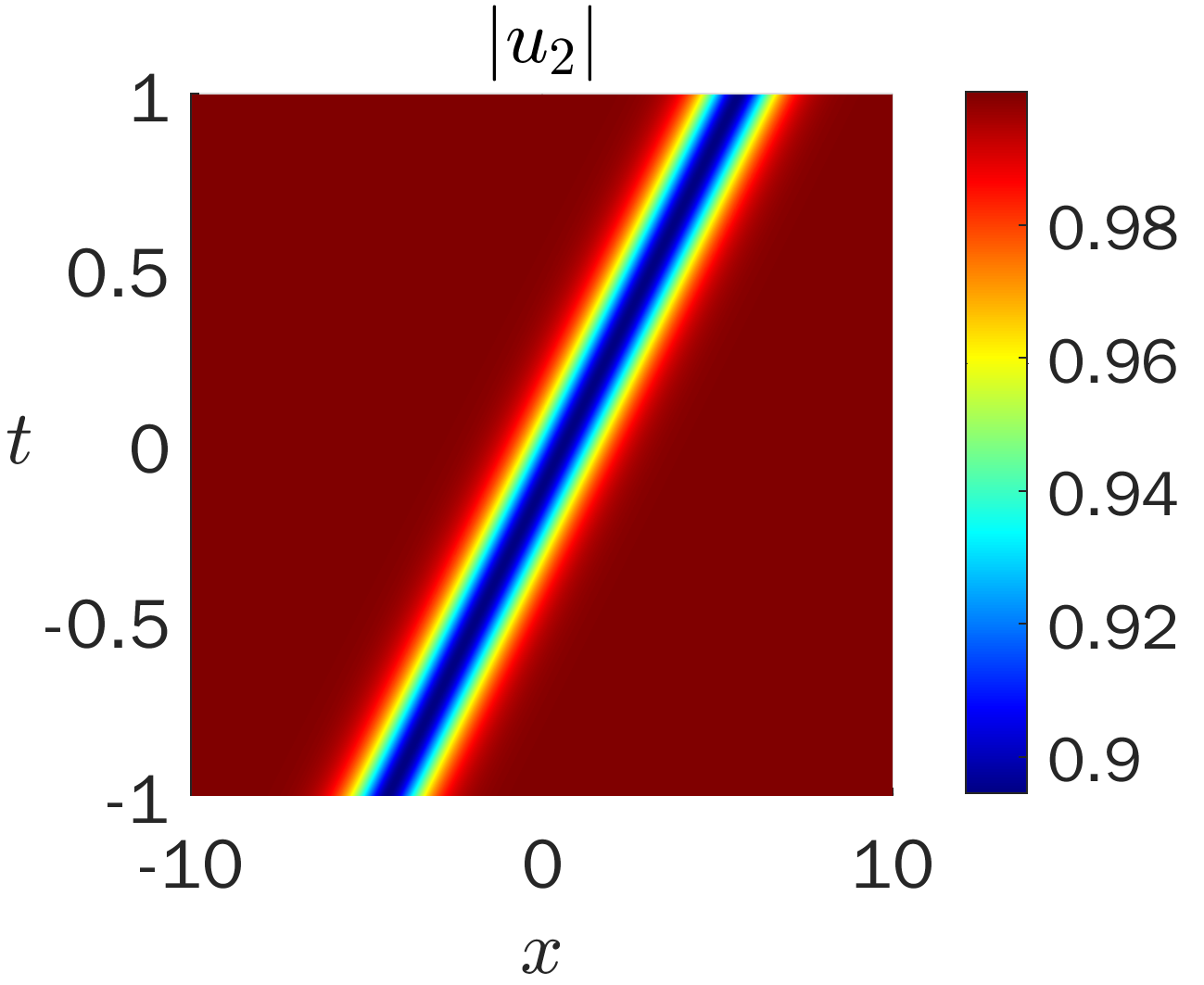}}
    \caption{One bright-dark soliton solution to the coupled Sasa-Satsuma equation under parameters \( N=1, \alpha=2, \rho=1,  \epsilon_1=-1, \epsilon_2=1, p_1=1, C_1=1,\xi_{1,0}=0\). (b) and (d) are the corresponding density plots of (a) and (c), respectively}
    \label{fig:one bright-dark soliton for N=1, case1}
\end{figure}

\begin{figure}[H]
    \centering
\subfigure[]{
        \includegraphics[width=30mm]{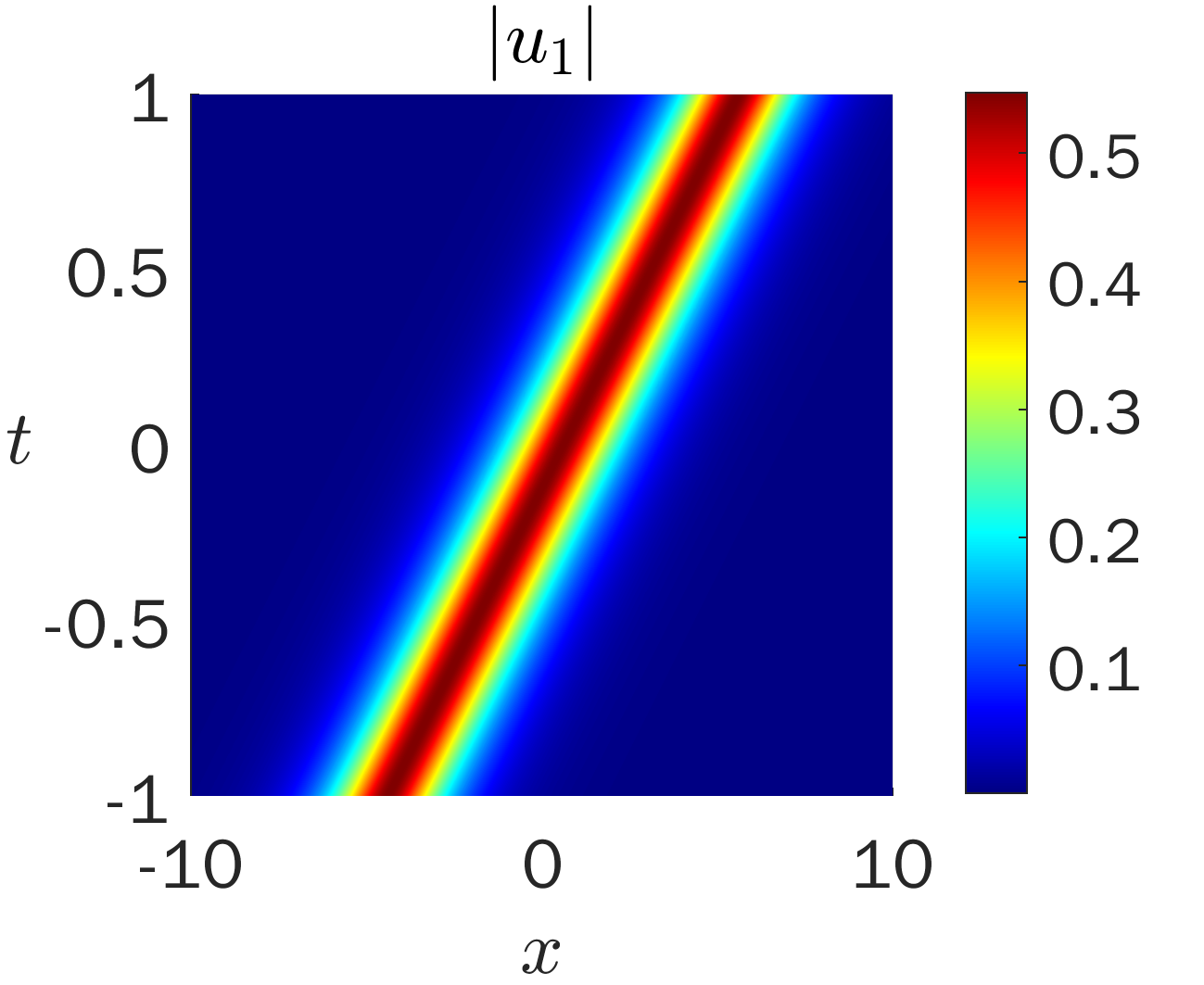}}
\subfigure[]{ 
        \includegraphics[width=30mm]{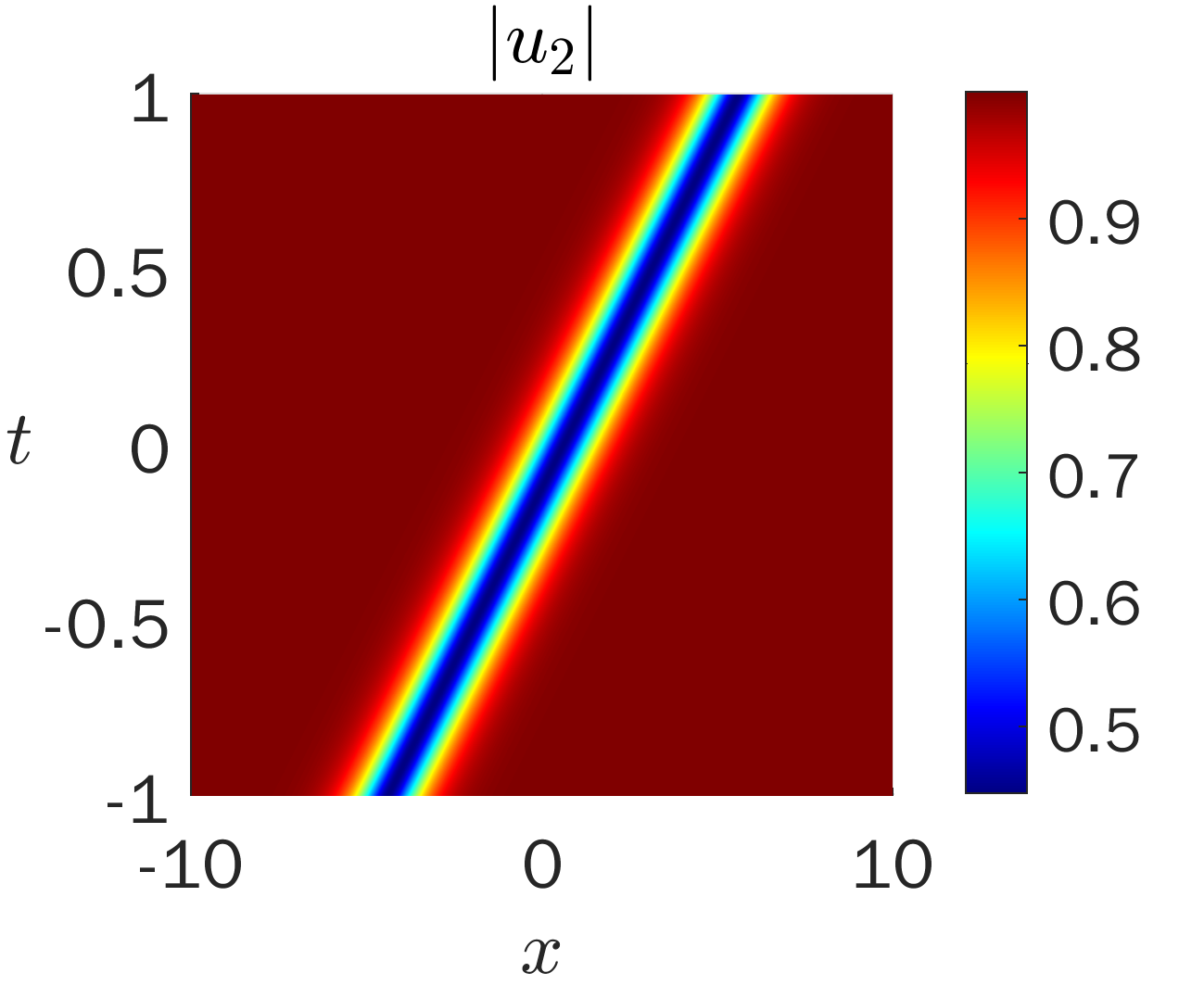}}
\subfigure[]{
        \includegraphics[width=30mm]{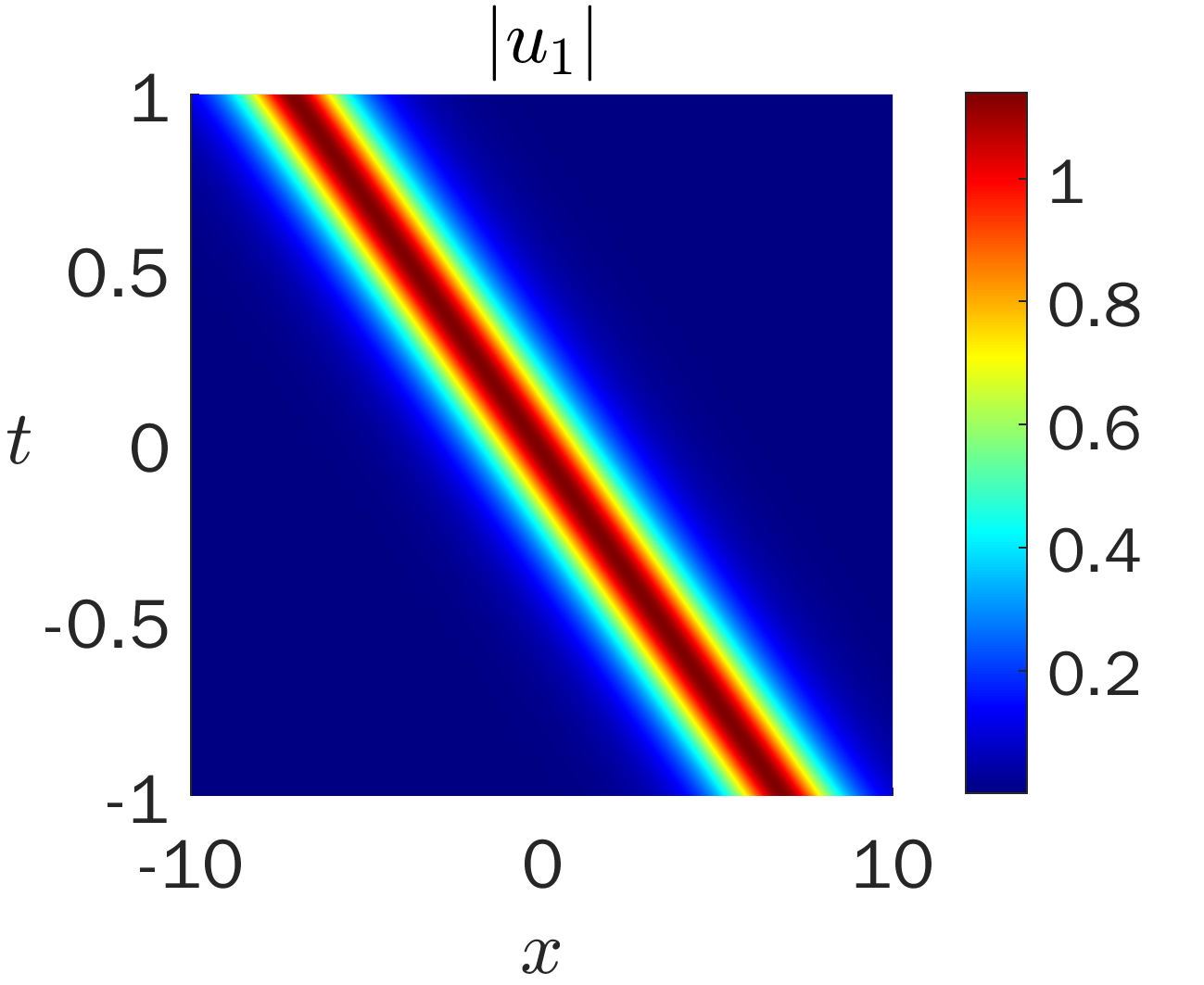}}
\subfigure[]{
        \includegraphics[width=30mm]{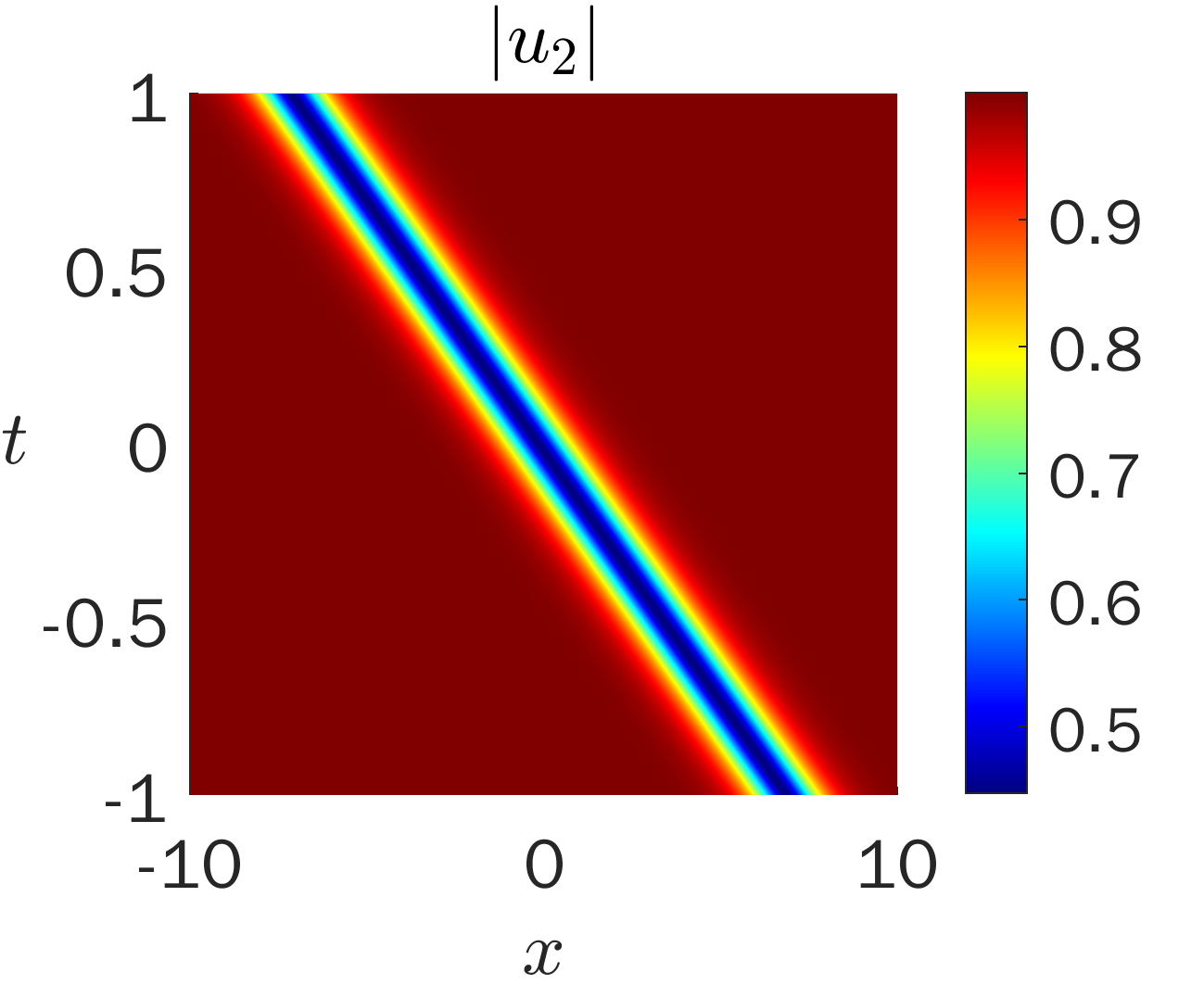}}
    \caption{One bright-dark soliton solution to the coupled Sasa-Satsuma equation under parameters \( N=1, \alpha=1/2, \rho=1, p_1=1, C_1=1,\xi_{1,0}=0\) with (a) and (b): \(  \epsilon_1=\epsilon_2=1\); (c) and (d): \(\epsilon_1= \epsilon_2=-1\).}
    \label{fig:one bright-dark soliton for N=1, case2}
\end{figure}

Furthermore, different sign combinations of \(\epsilon_1\) and \(\epsilon_2\), such as \((\epsilon_1, \epsilon_2) = (1, 1)\) or \((-1, -1)\), can be explored. Corresponding examples are depicted in Fig. \ref{fig:one bright-dark soliton for N=1, case2}.

\subsection{One bright-dark soliton and breather for $N=2$}
When $N=2$ in \cref{thm:bd_css}, the parameters are subject to the constraints $p_1 = p_2^*$ and $\xi_{1,0} = \xi_{2,0}^*$. Assume $p_1 = a + b\mathrm{i}$, where $a,b \in \mathbb{R}$, and recall \(d_{1,1} = d_{2,2}\), and \(d_{1,2} = d_{2,1}^*\) (see \cref{lemma:dij}), then the solution \eqref{css_bd_sln} can be rewritten as
\begin{equation}\label{n2form}
\begin{aligned}
    f ={} & \frac{b^2 e^{2 \xi _1^*+2 \xi _1}}{4 a^2 \left(a^2+b^2\right)}+\frac{d_{1,1} e^{\xi _1^*+\xi _1}}{a}-\frac{d_{1,2}^* e^{2 \xi _1}}{2 a+2\mathrm{i} b}-\frac{d_{1,2} e^{2 \xi _1^*}}{2 a-2 \mathrm{i} b}-|d_{1,2}|^2 +d_{1,1}^2, \\
    g ={} & \frac{b e^{\xi _1^*+\xi _1} \left(C_2 (b-\mathrm{i} a) e^{\xi _1^*}+C_1 (b+\mathrm{i} a) e^{\xi _1}\right)}{2 a \left(a^2+b^2\right)}+\left(C_1 d_{1,1}-C_2 d_{1,2}^*\right) e^{\xi _1}\\
    &+\left(C_2 d_{1,1}-C_1 d_{1,2}\right) e^{\xi _1^*}, \\
    h ={} &\frac{b^2 \left(b^2+(a-\mathrm{i} \alpha )^2\right) e^{2 \xi _1^*+2 \xi _1}}{4 a^2 \left(a^2+b^2\right) \left(b^2+(a+\mathrm{i} \alpha )^2\right)}-\frac{d_{1,1} \left(a^2+\alpha ^2-b^2\right) e^{\xi _1^*+\xi _1}}{a \left(b^2+(a+\mathrm{i}\alpha )^2\right)}\\
    & +\frac{d_{1,2}^* (a+\mathrm{i} (b-\alpha )) e^{2 \xi _1}}{2 (a+\mathrm{i} b) (a+\mathrm{i} (\alpha +b))}+\frac{d_{1,2} (a-\mathrm{i} (\alpha +b)) e^{2 \xi _1^*}}{2 \left(\alpha  (b+\mathrm{i} a)+(a-\mathrm{i} b)^2\right)}-|d_{1,2}|^2 +d_{1,1}^2, \\
\end{aligned}
\end{equation}
with the dispersion relation \(\xi_1 = p_1x + (p_1^3 - 6p_1\rho^2\epsilon_2)t + \xi_{1,0}\).
When $C_1 C_2 \neq 0$, we define,
\begin{equation}
\exp(\phi) = \sqrt{\frac{C_1 d_{1,1} - C_2d_{1,2}^*}{C_2 d_{1,1}-C_1 d_{1,2}}},
\end{equation}
which yields
\begin{eqnarray}
    g &=& \frac{b e^{\xi _1^*+\xi _1} \left(C_2 (b-\mathrm{i} a) e^{\xi _1^*}+C_1 (b+\mathrm{i} a) e^{\xi _1}\right)}{2 a \left(a^2+b^2\right)}\nonumber\\
    &+& 2\sqrt{(C_1 d_{1,1} - C_2d_{1,2}^*)(C_2 d_{1,1}-C_1 d_{1,2})}e^{\Re{(\xi_1)}}\cos(\Im(\xi_1) - \mathrm{i}\phi).
\end{eqnarray}
This expression contains a periodic term if \(\Im(p_1) \neq 0\), leading to a breather solution for $u_1$ (see \cref{fig:f3e,fig:f4e}). A similar analytical approach can be extended to $u_2$ and generate a breather solution as well, which is illustrated in \cref{fig:f3g,fig:f4g}. Moreover, when $p_1\in \mathbb{R}$, it can be shown that the solution degenerates into a solution of $N = 1$.

\begin{figure}[H]
    \centering
\subfigure[]{\label{fig:f4a}
        \includegraphics[width=30mm]{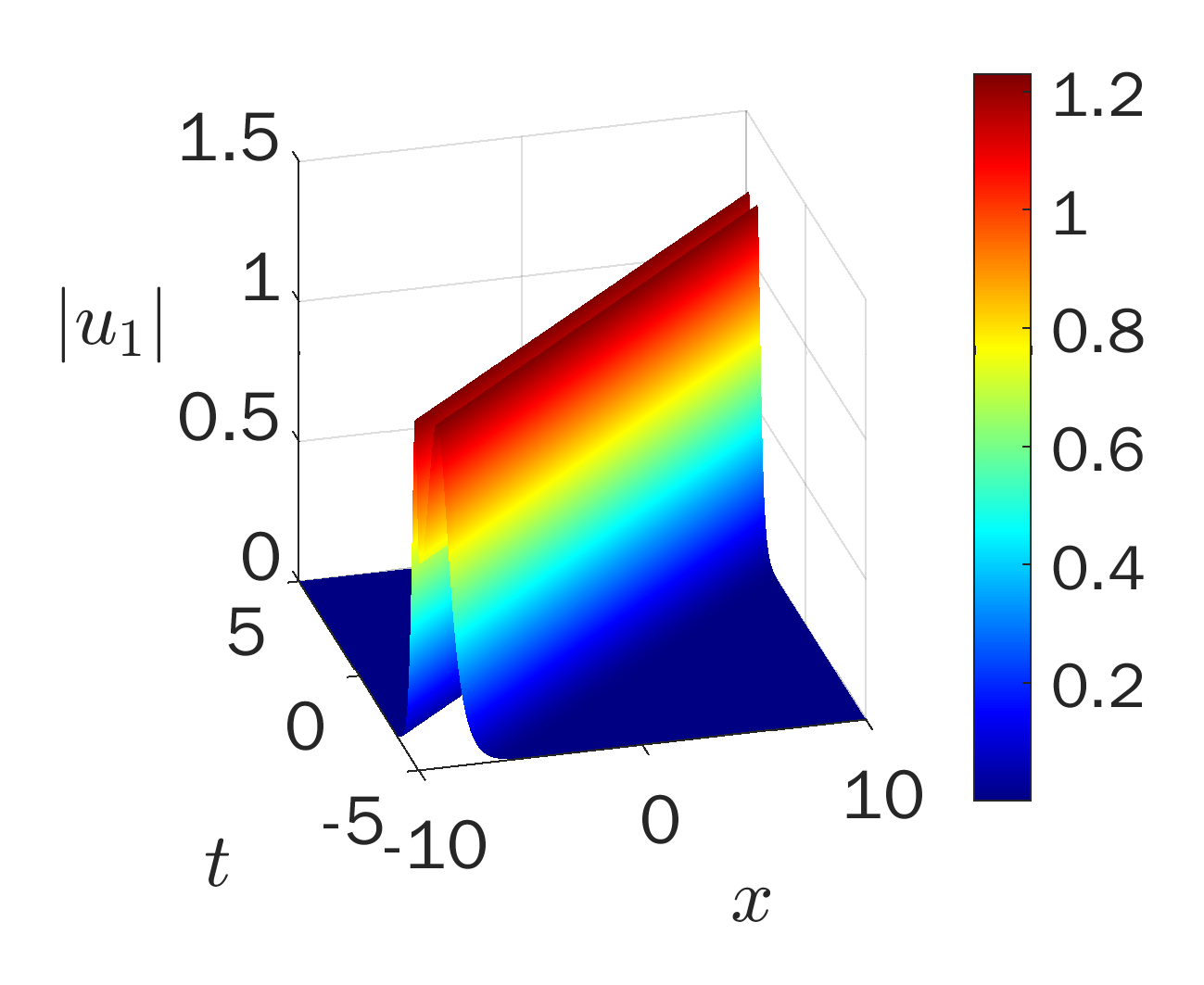}}
\subfigure[]{ 
        \includegraphics[width=30mm]{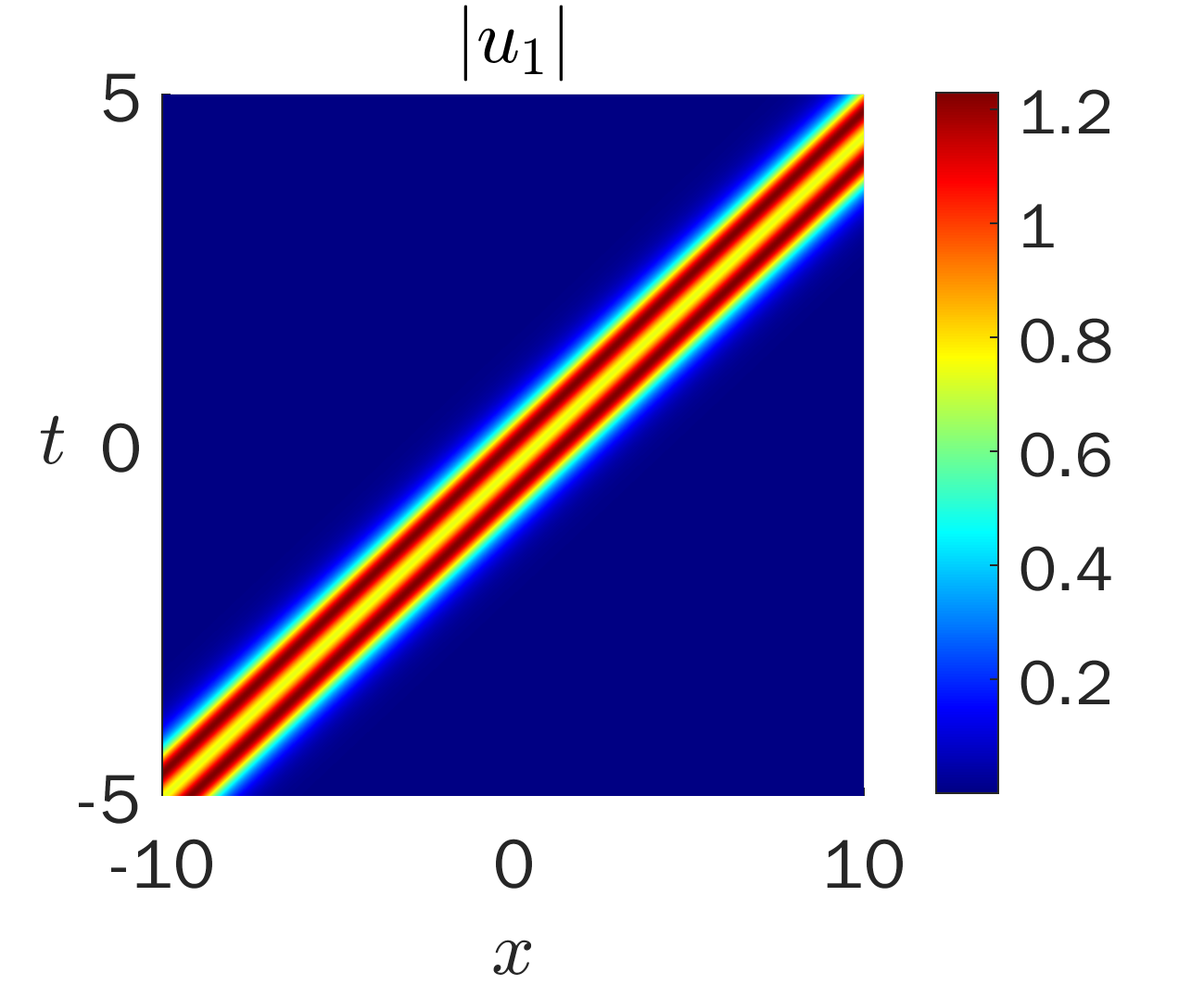}}
\subfigure[]{
        \includegraphics[width=30mm]{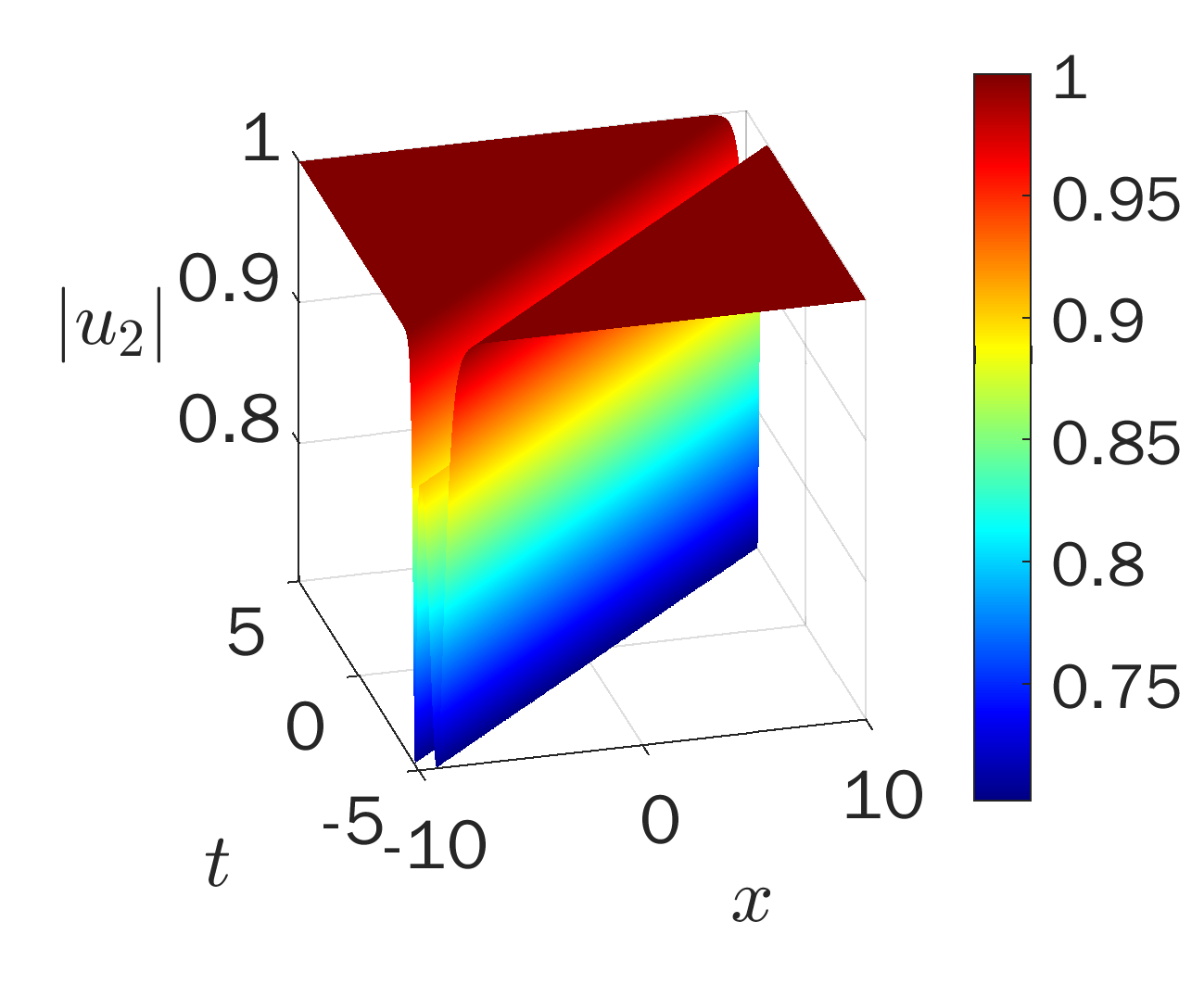}}
\subfigure[]{\label{fig:f4c}
        \includegraphics[width=30mm]{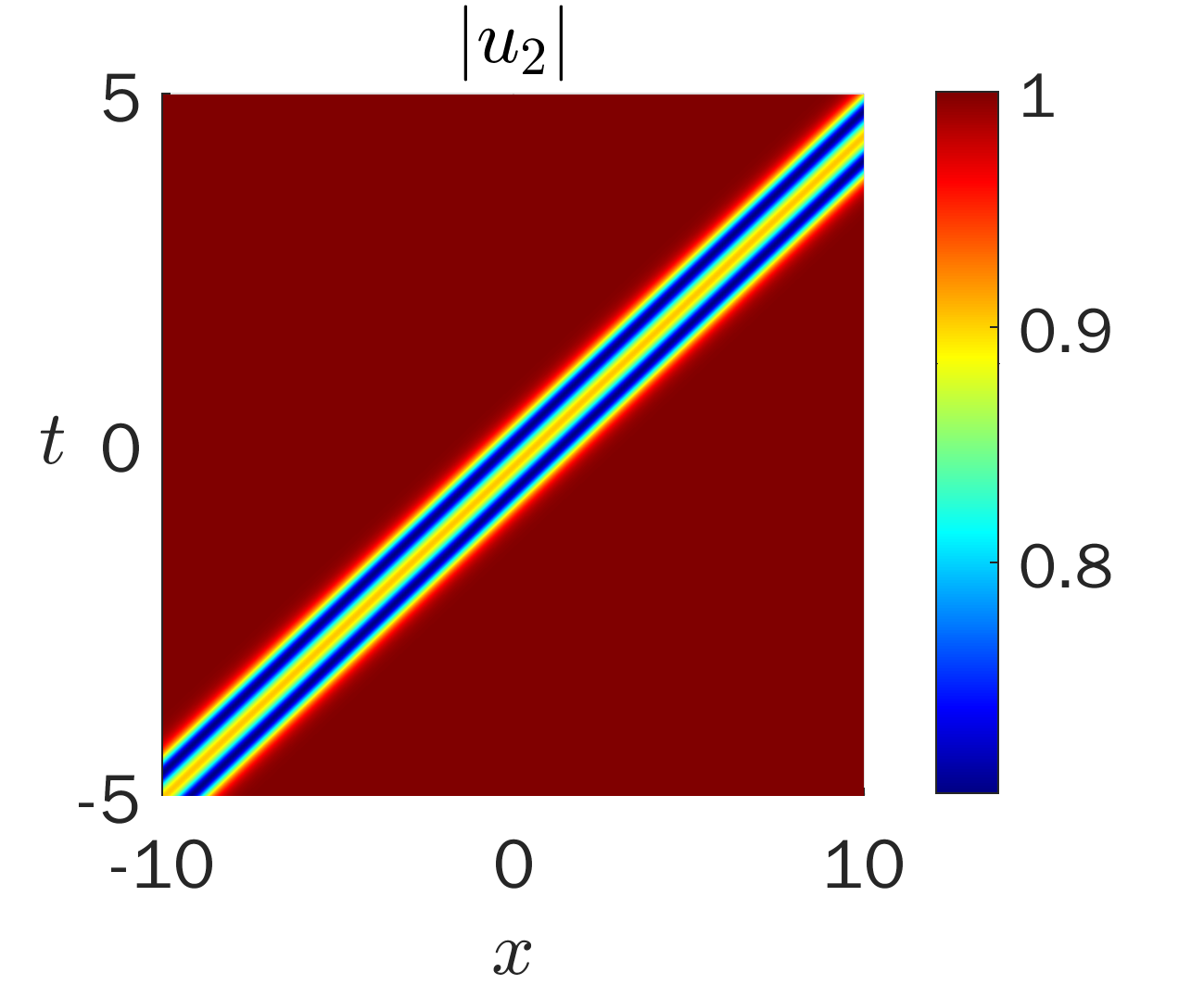}}
\subfigure[]{\label{fig:f4e}
        \includegraphics[width=30mm]{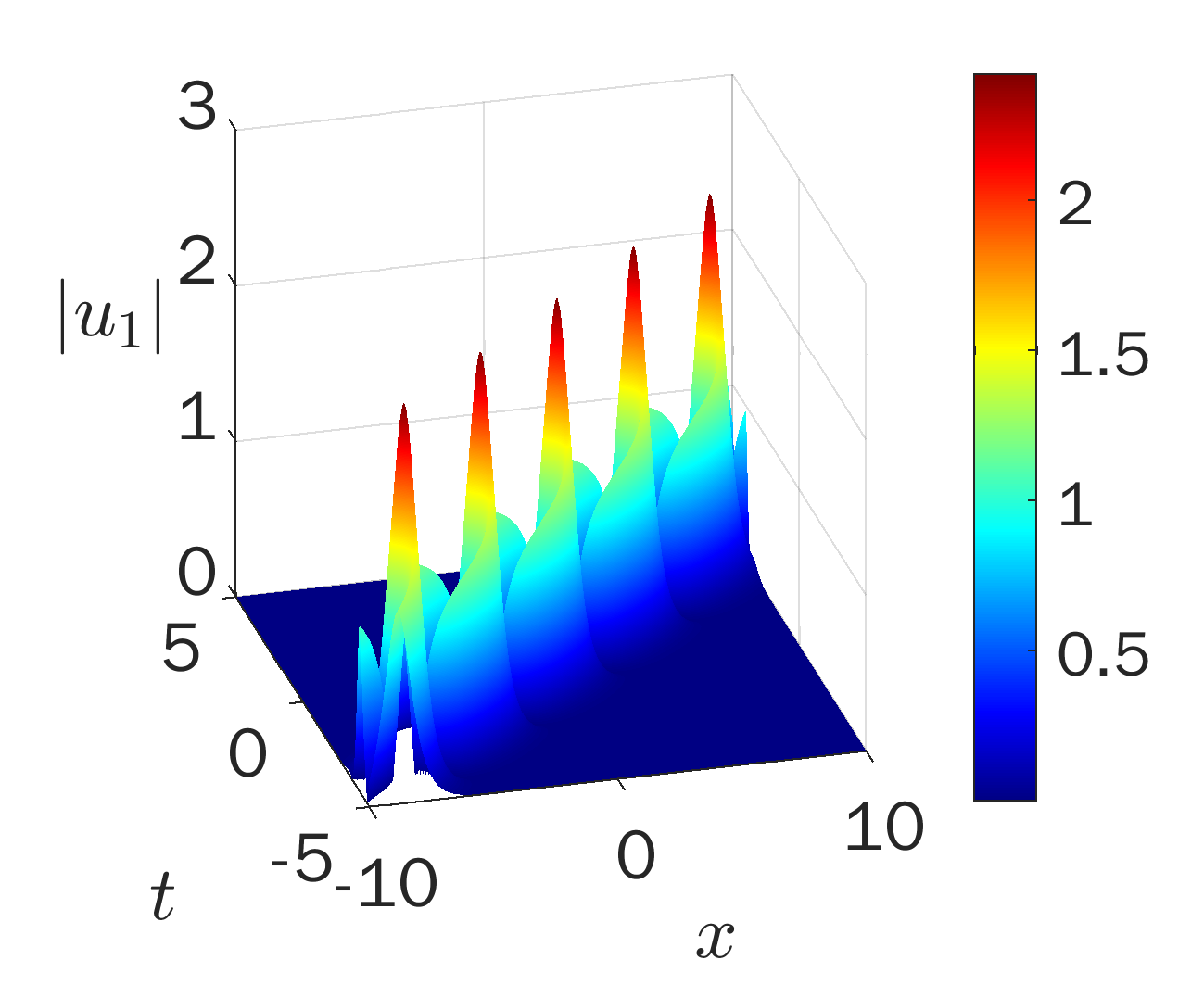}}
\subfigure[]{ 
        \includegraphics[width=30mm]{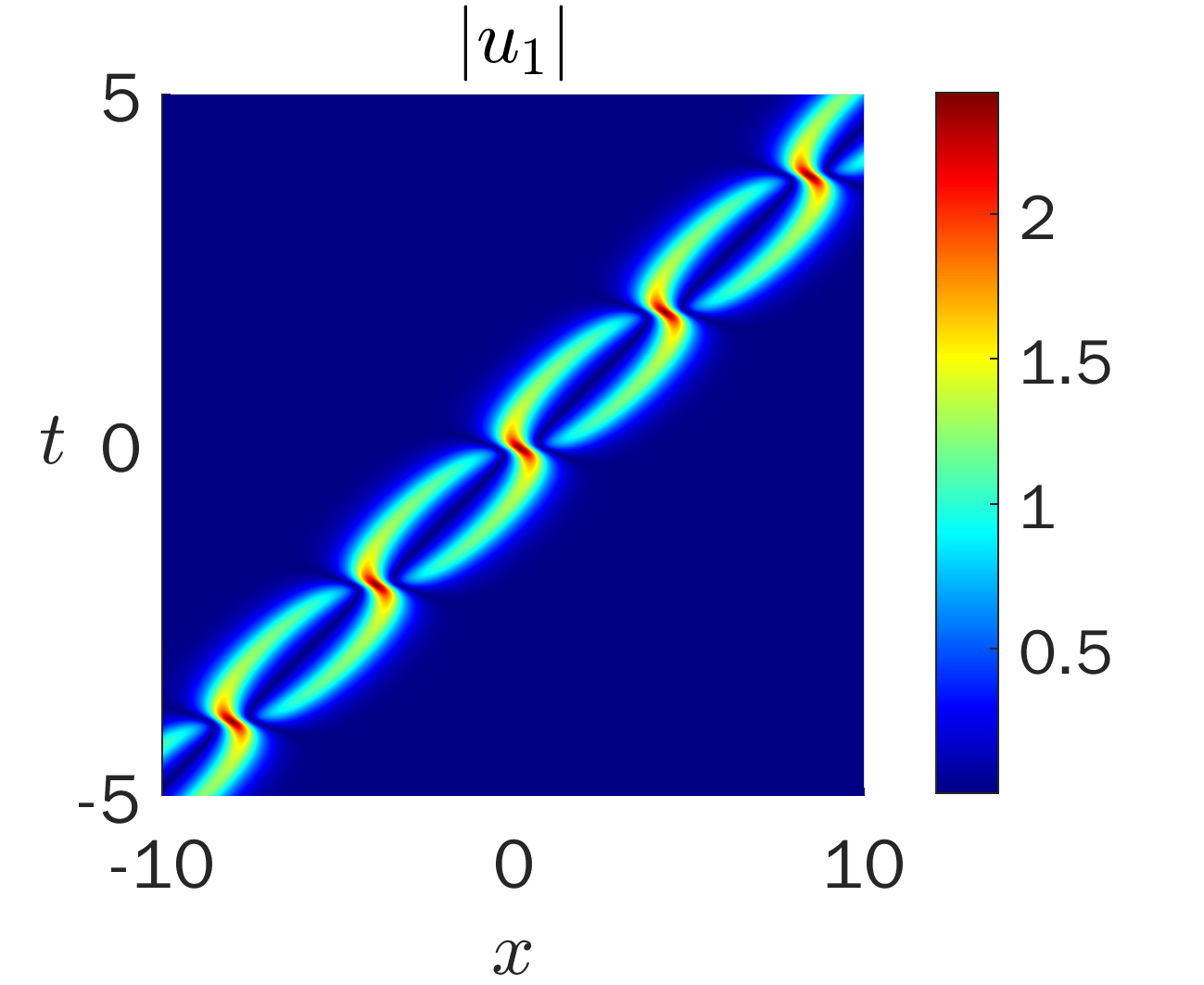}}
\subfigure[]{\label{fig:f4g}
        \includegraphics[width=30mm]{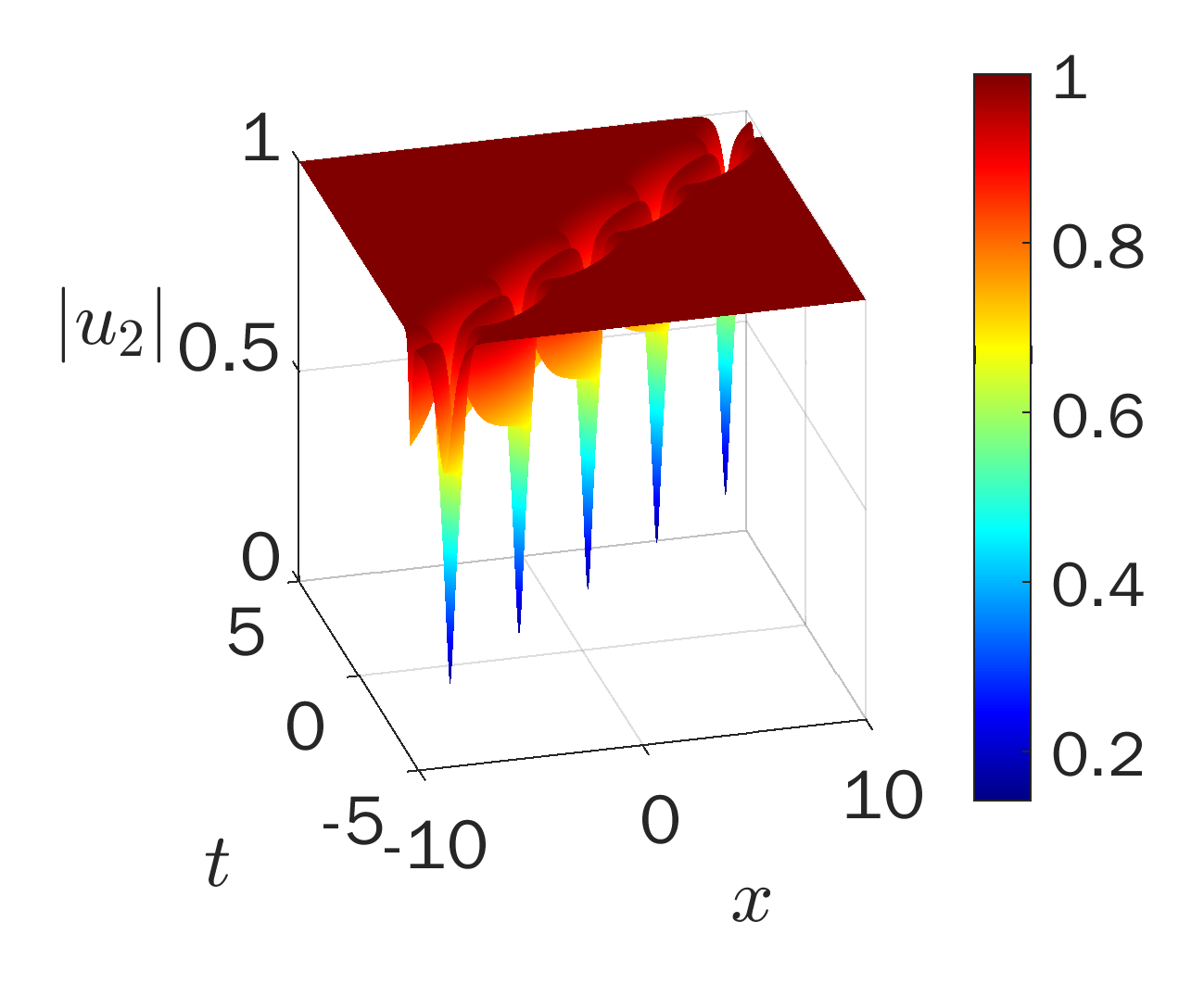}}
\subfigure[]{ 
        \includegraphics[width=30mm]{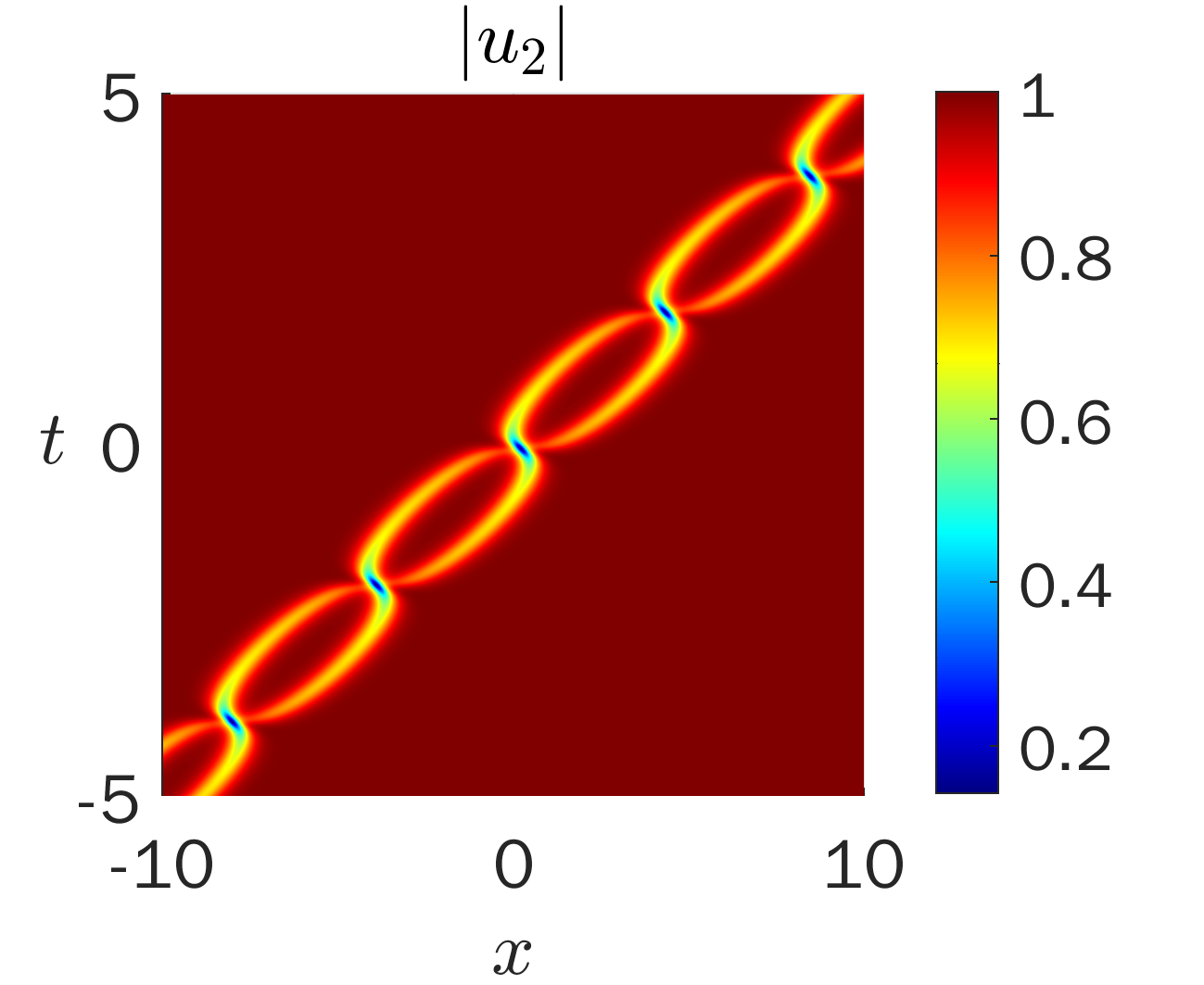}}
    \caption{One bright-dark soliton solution to the coupled Sasa-Satsuma equation under parameters \( N=2, \alpha=-2, \rho=1,  \epsilon_1=-1, \epsilon_2=1, p_1=2+0.2\mathrm{i}\) with the first row \( C_1=1,C_2=0,\xi_{1,0}=\xi_{2,0}=0\) and the second row \(C_1=C_2=1\). The second column and the forth column show the density plots of the first and third columns, respectively.}
    \label{fig:one bright-dark soliton for N=2 case2}
\end{figure}

To avoid periodicity and obtain a soliton solution, we set $C_1 \neq 0$, $C_2 = 0$, or $C_1 = 0$, $C_2 \neq 0$ in \eqref{n2form}. For example, consider the case $C_1 \neq 0$, $C_2 = 0$, then the solution \eqref{css_bd_sln} becomes,
\begin{equation}\label{n2formC2=0}
\begin{aligned}
    f & = \frac{b^2 }{4 a^2 \left(a^2+b^2\right)}e^{ 2 \xi_1^* +  2\xi_1 }+\frac{d_{1,1} }{a}e^{ \xi_1^* +  \xi_1}+d_{1,1}^2,\\
    g & = C_1 e^{\xi _1} \left(d_{1,1}+\frac{b e^{\xi _1^*+\xi _1}}{2 a (b-\mathrm{i} a)}\right), \\
    h & = \frac{b^2 \left(b^2+(a-\mathrm{i} \alpha )^2\right)}{4 a^2 \left(a^2+b^2\right) \left(b^2+(a+\mathrm{i} \alpha )^2\right)} e^{2 \xi_1^* +  2\xi_1} -\frac{d_{1,1} \left(a^2+\alpha ^2-b^2\right)}{a \left(b^2+(a+\mathrm{i} \alpha )^2\right)} e^{\xi_1^* + \xi_1} + d_{1,1}^2.\\
\end{aligned}
\end{equation}

For the following analysis, we assume \(a > 0\) and \(b > 0\) without loss of generality. Under this assumption, the condition \(d_{1,1} > 0\) must be satisfied to ensure \(f \neq 0\), thereby guaranteeing the regularity of the solution.

Assuming $y = \exp(\xi_1 + \xi_1^*)$ in \eqref{n2formC2=0}, we can solve \(y\) from equations
\begin{equation*}
    \partial_y|u_1|^2 = 0, \quad \partial_y|u_2|^2 = 0.
\end{equation*}
This yields four possible roots
{\allowdisplaybreaks\begin{equation}
\begin{aligned}
    & y_1 = -\frac{2a}{b}\sqrt{a^2 + b^2}|d_{1,1}|, 
    \quad y_2 = \frac{2a}{b}\sqrt{a^2 + b^2}|d_{1,1}|,\\
    &y_3  = \frac{2\left(a(a^2 - b^2)  - a^2 \sqrt{a^2 - 3b^2}\right)|d_{1,1}|}{b^2}, \quad y_4  = \frac{2\left(a(a^2 - b^2)  + a^2 \sqrt{a^2 - 3b^2}\right)|d_{1,1}|}{b^2}.
\end{aligned}
\end{equation}}
Each of these \(y_i > 0, i = 1, 2, 3, 4\) could indicate an extreme of \(|u_1|\) and \(|u_2|\). Therefore, the soliton solutions of \eqref{n2formC2=0} can be classified into two cases.

\begin{itemize}
    \item [1.] When \(a^2 > 3b^2\), the roots satisfy \(y_1 < 0\) and \(0 < y_3 < y_2 < y_4\), and 
    {\allowdisplaybreaks\begin{equation}
    \begin{aligned}
        |u_1(y_2)|^2 &= \frac{4 a^2 c_1^2 y \left(a^2+b^2\right) \left(4 a d_{11} \left(a d_{11} \left(a^2+b^2\right)+b^2 y\right)+b^2 y^2\right)}{\left(4 a d_{11} \left(a^2+b^2\right) \left(a d_{11}+y\right)+b^2 y^2\right)^2},\\
        |u_1(y_3)|^2 &= |u_1(y_4)|^2 = \frac{c_1^2 \left(a^2+b^2\right)}{4 a d_{11}},\\
        |u_2(y_2)|^2 & = \frac{\rho^2}{d_{1,1} \left(a^2+(b-\alpha )^2\right) \left(a^2+(\alpha +b)^2\right)}\left(-4 b \left(a^2+\alpha ^2+b^2\right) \sqrt{d_{1,1}^2 \left(a^2+b^2\right)}\right.\\
        &\left.+d_{1,1} \left(a^4+\alpha ^4+2 \alpha ^2 \left(a^2+b^2\right)+6 a^2 b^2+5 b^4\right)\right)
        \\
        |u_2(y_3)|^2 & = |u_2(y_4)|^2 = \rho^2 \frac{\alpha ^2 \left(a^2+\alpha ^2-3 b^2\right)}{\alpha ^4+\left(a^2+b^2\right)^2+2 \alpha ^2 (a^2 - b^2)}.
    \end{aligned}  
    \end{equation}
    Further, it can be proven that
    \begin{equation}
        0 < |u_1(y_2)| < |u_1(y_3)| = |u_1 (y_4)|,\quad
        \rho > |u_2(y_2)| > |u_2(y_3)| = |u_2(y_4)|, 
    \end{equation}}
    As a consequence, this case represents a two-hump soliton for both \(u_1\) and \(u_2\), with the humps located at \(\xi_1 + \xi_1^* = \log y_3\) and \(\xi_1 + \xi_1^* = \log y_4\),  as illustrated in  \cref{fig:f4a,fig:f4c}.
    \item [2.] When \(a^2 < 3b^2\), only \(y_2>0\) is valid, which means the solution is a single-hump soliton with the hump located at \(\xi_1 + \xi_1^* = \log y_2\) for both of the components \(u_1\) and \(u_2\), as shown in \cref{fig:f3a,fig:f3c}.
\end{itemize} 


\begin{figure}[H]
    \centering
\subfigure[]{\label{fig:f3a}
        \includegraphics[width=30mm]{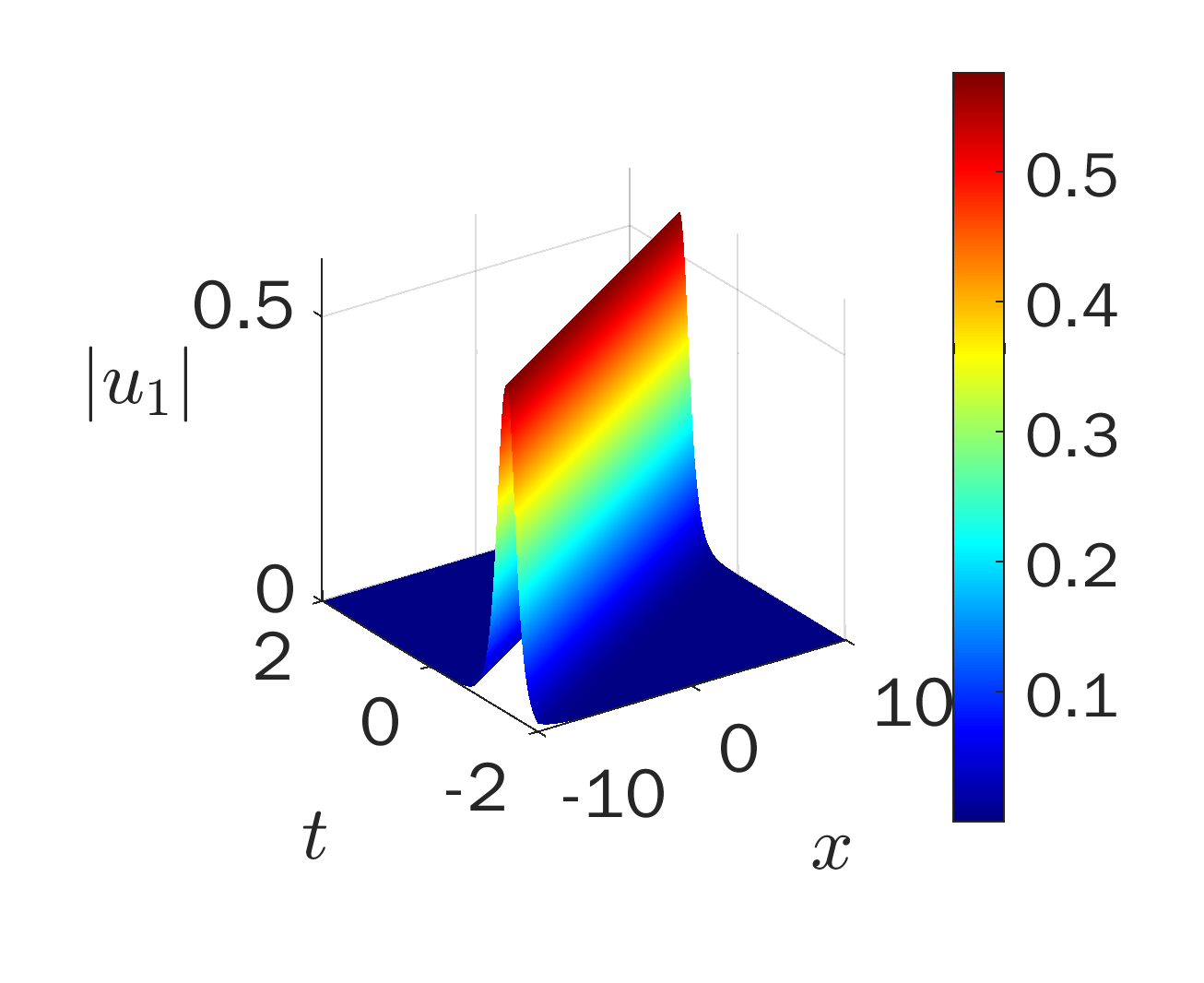}}
\subfigure[]{ 
        \includegraphics[width=30mm]{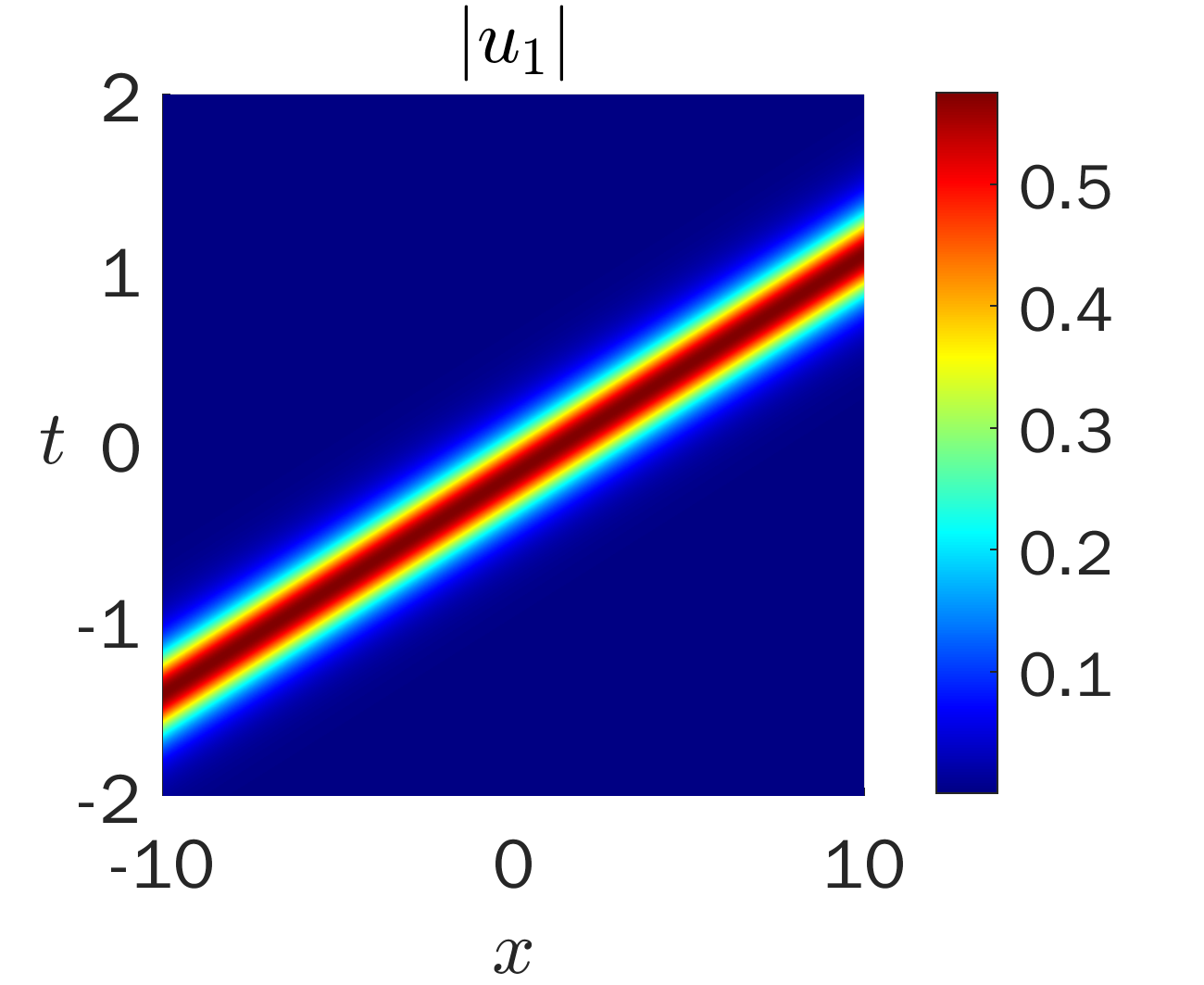}}
\subfigure[]{\label{fig:f3c}
        \includegraphics[width=30mm]{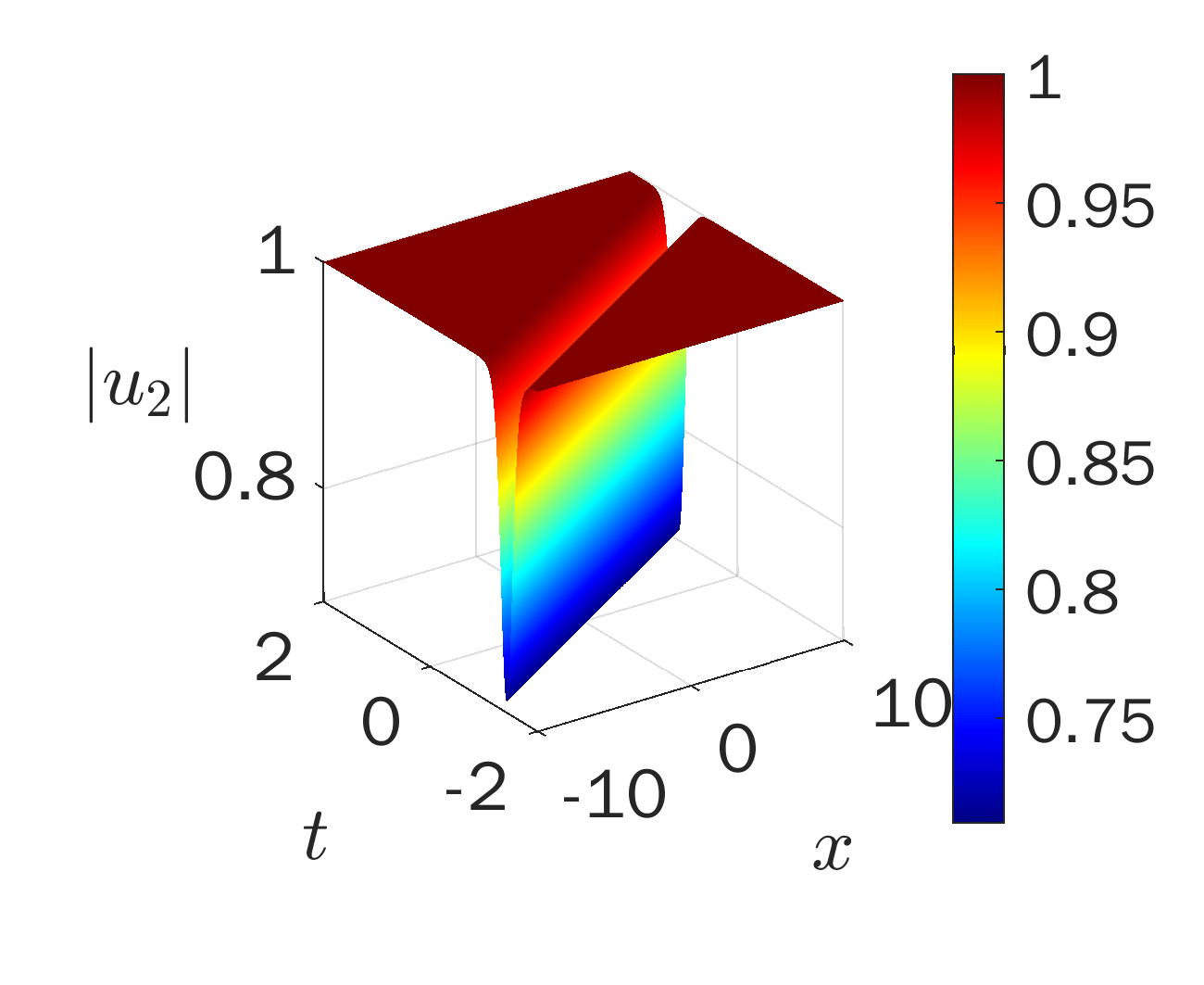}}
\subfigure[]{%
        \includegraphics[width=30mm]{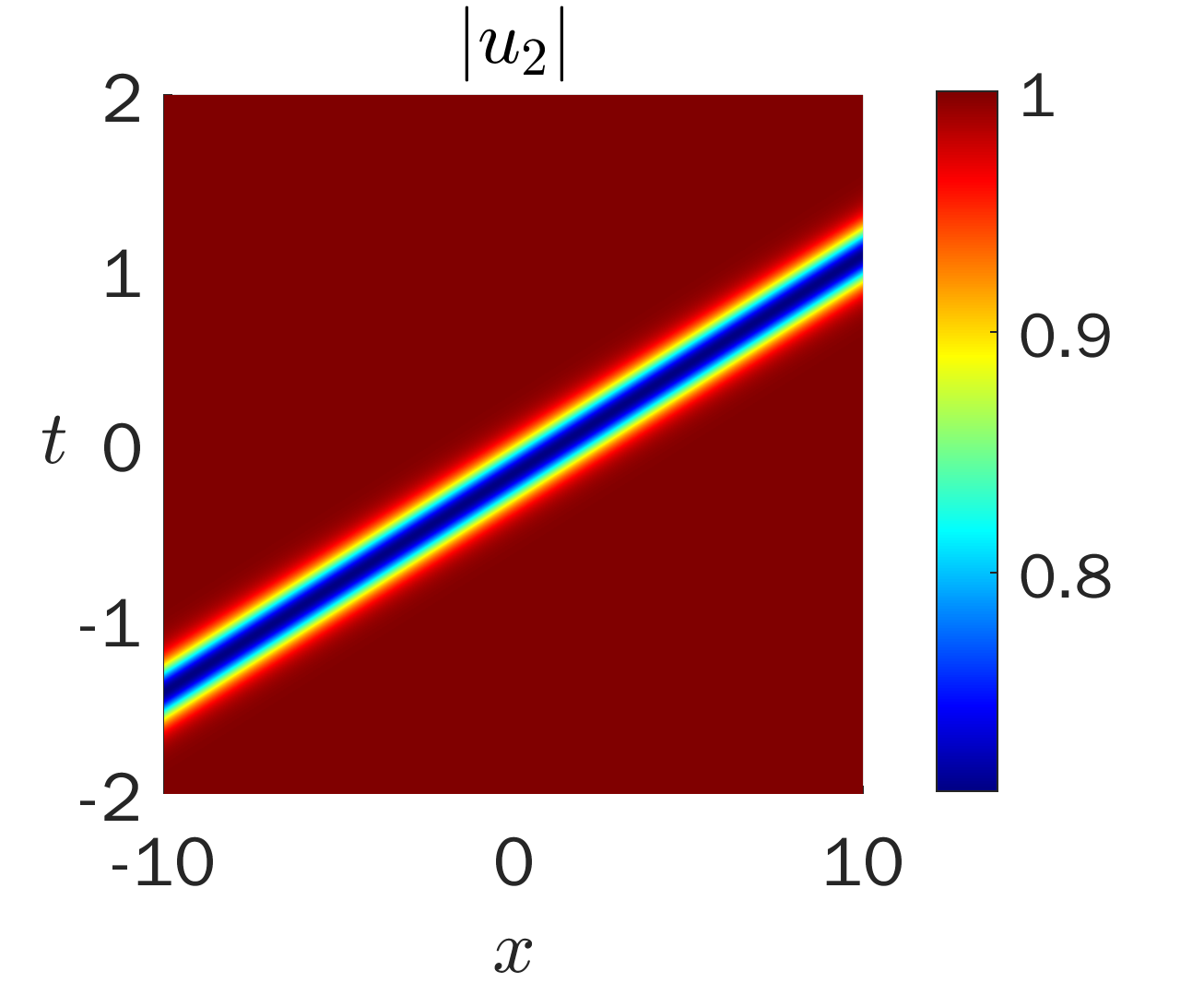}}
\subfigure[]{\label{fig:f3e}
        \includegraphics[width=30mm]{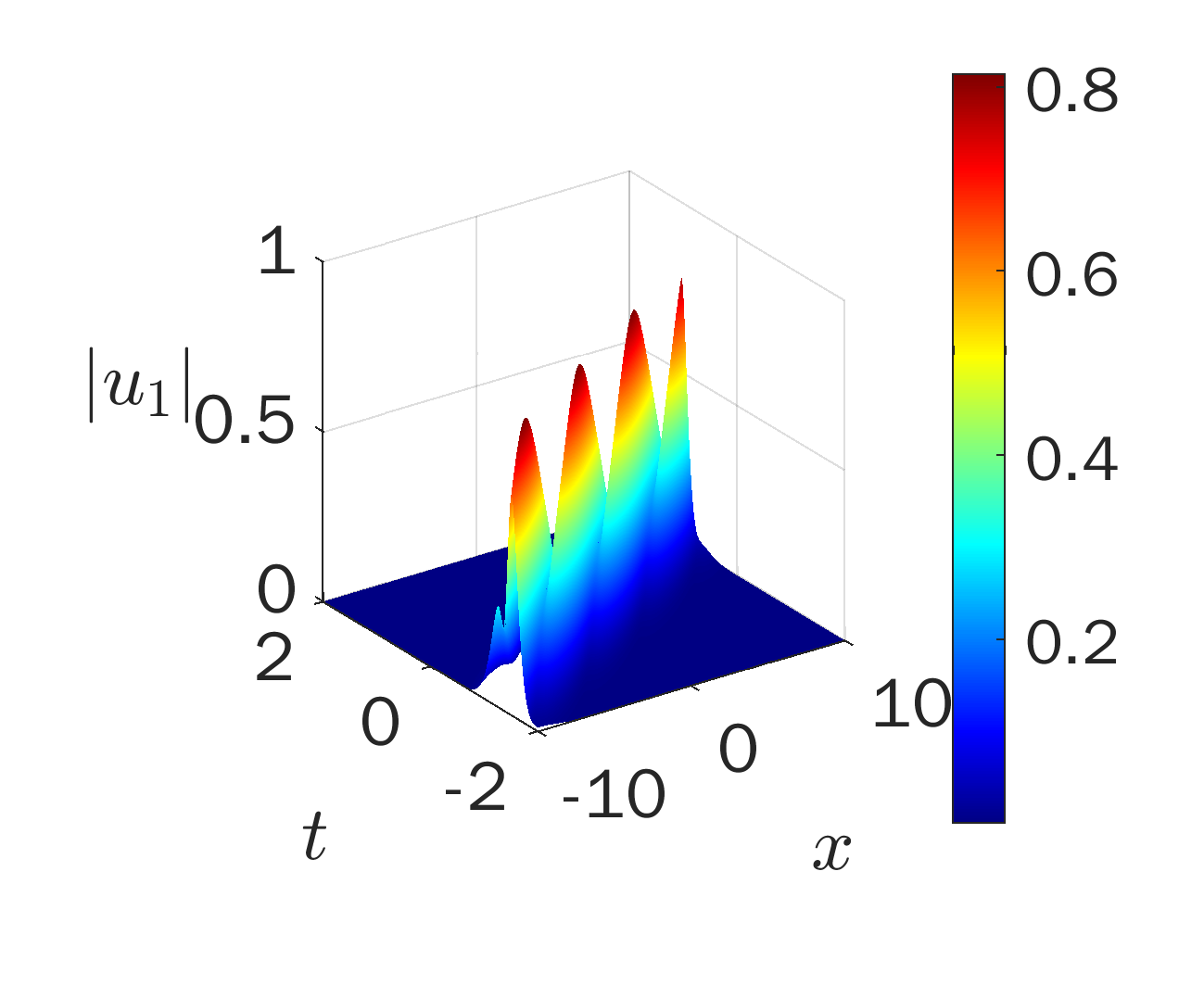}}
\subfigure[]{ 
        \includegraphics[width=30mm]{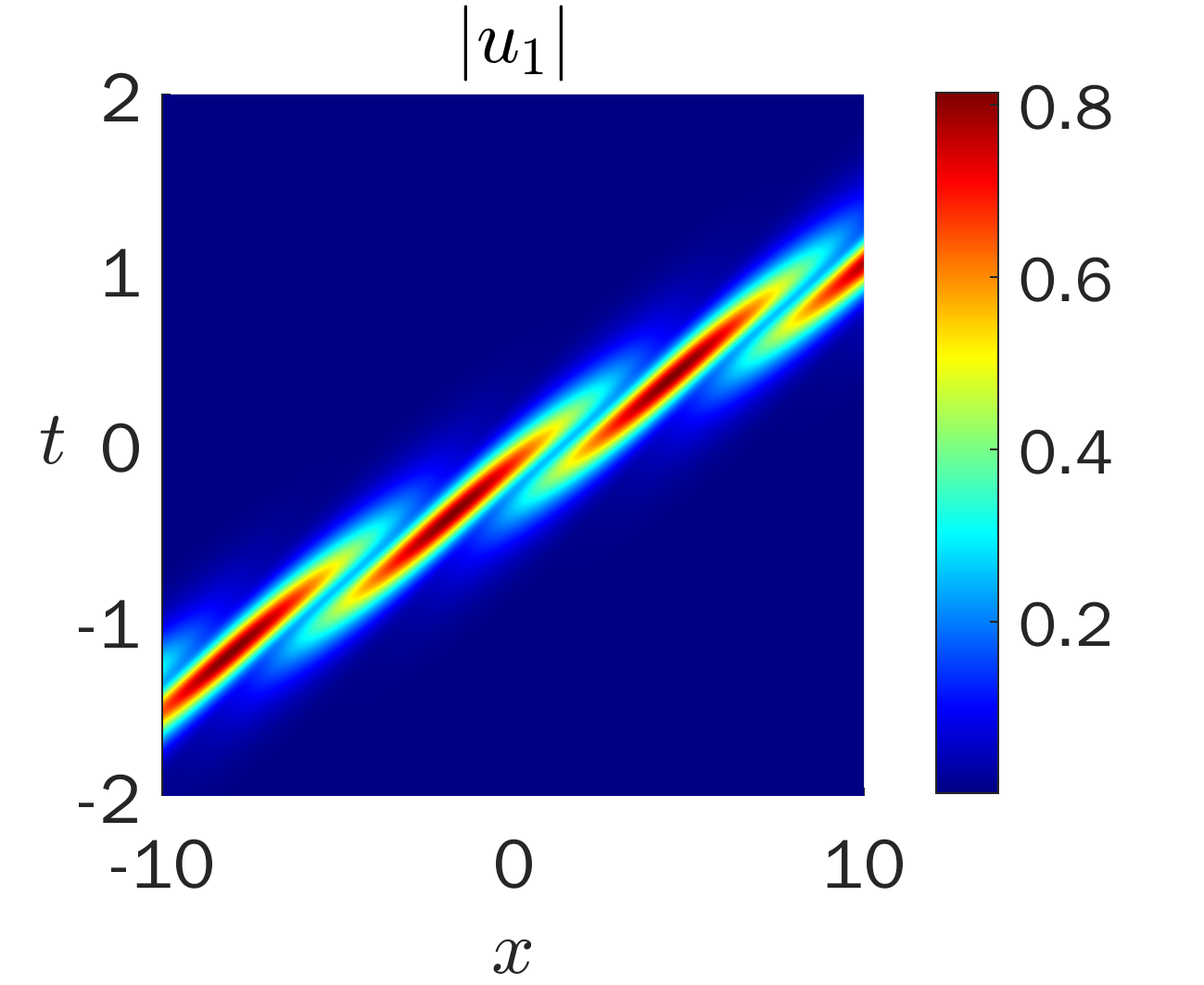}}
\subfigure[]{\label{fig:f3g}
        \includegraphics[width=30mm]{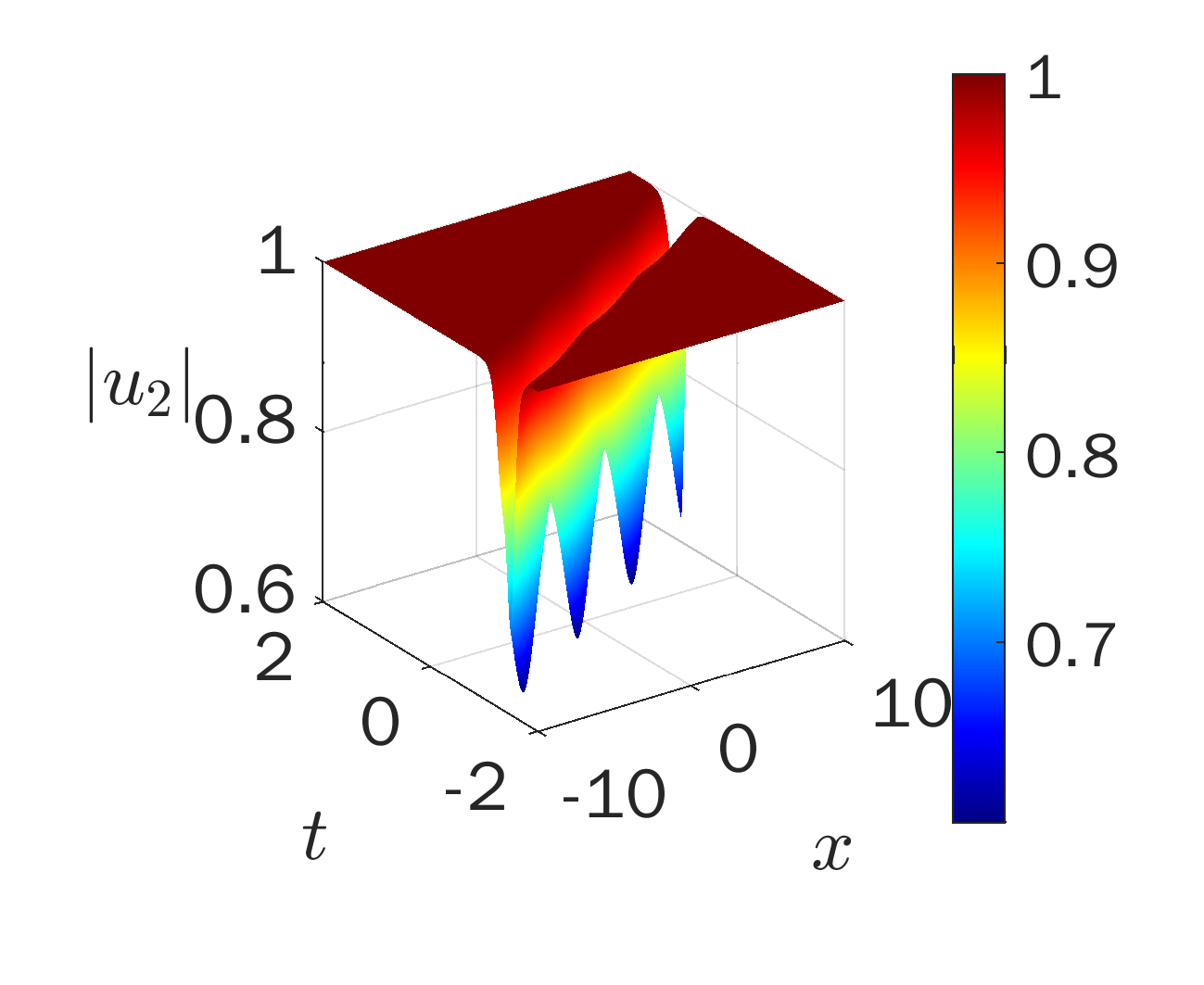}}
\subfigure[]{
        \includegraphics[width=30mm]{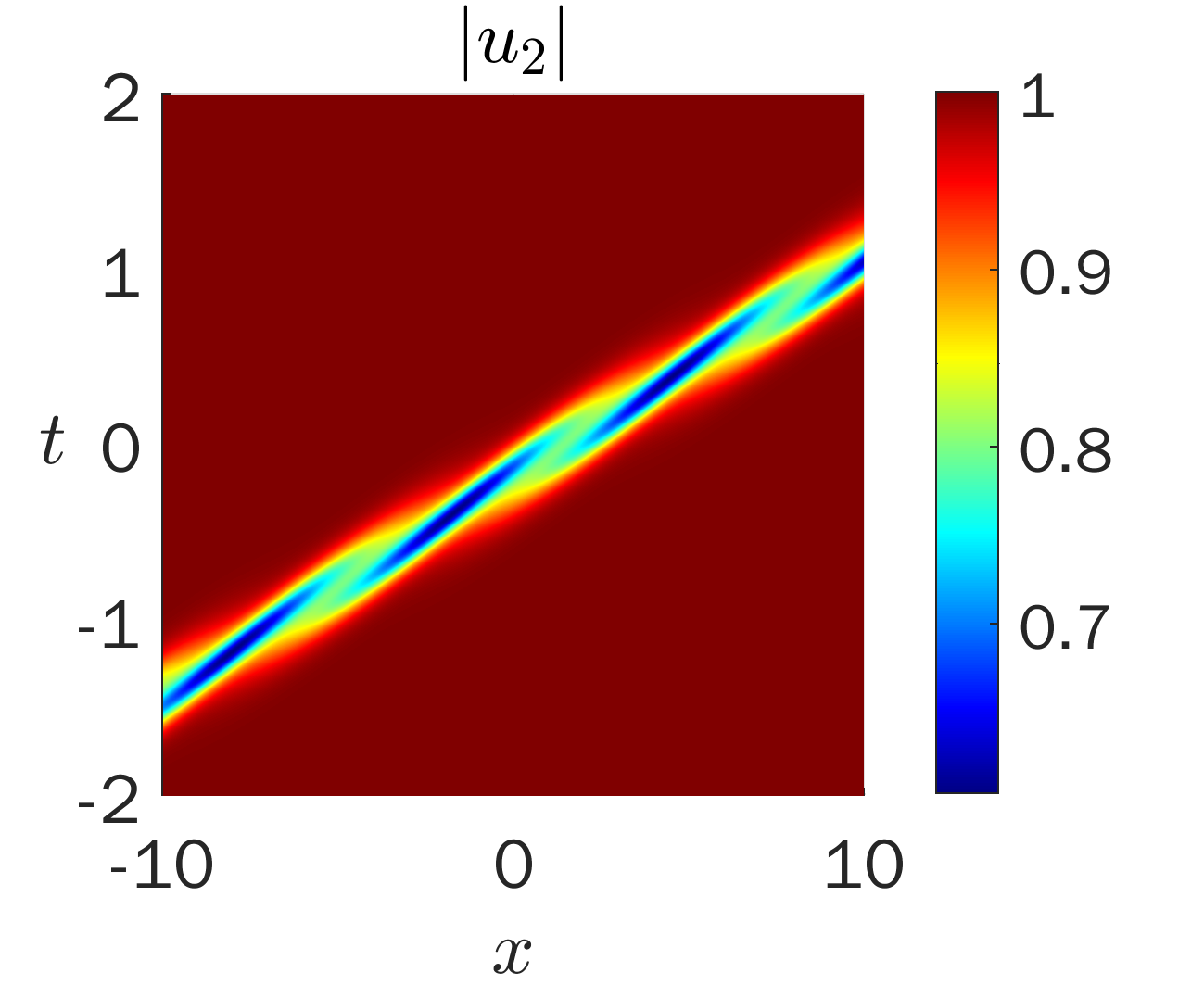}}
    \caption{
    One bright-dark solutions to the coupled Sasa-Satsuma equation under parameters \( N = 2 \), \( \alpha = 2 \), \( \rho = 1 \), \( \epsilon_1 = -1 \), \( \epsilon_2 = 1 \), \( p_1 = 1 + \mathrm{i},\xi_{1,0}=\xi_{2,0}=0 \) and \(C_1 = 2\) with the first row \(C_2=0\) and the second row \(C_2=1\).  (b) and (d) are the corresponding density plots of (a) and (c), respectively. (f) and (h) are the corresponding density plots of (e) and (g), respectively.
    }
    \label{fig:one bright-dark soliton for N=2 case1}
\end{figure}

\subsection{Two bright-dark solitons for $N=3$}
For \(N=3\), solutions in \cref{thm:bd_css} undergo a parameter restriction \(p_1 = p_3^*, p_2 \in \mathbb{R}\) and \(\xi_{1,0} = \xi_{3,0}^*, \xi_{2,0} \in \mathbb{R}\). In this case, two bright-dark solitons can be obtained. By taking
\begin{equation}
        \alpha=1,\quad \rho=1, \quad \epsilon_1=-1,\quad \epsilon_2=1, \quad p_1=2+\mathrm{i},\quad p_2=2,
\end{equation}
and varying values of $C_1, C_2, C_3$ and $C_4$, we can observe a soliton interacting with another soliton or a breather, as illustrated in Fig. \ref{fig:two bright-dark soliton for N=3, case1}. 
The nature of the interactions is dictated by the value of \(C_3\): 
\begin{itemize}
\item Figs. \ref{fig:N=3_alpha=1_rho=1_eps=[-1,1]_p=[2+1.0i,2,2-1.0i]_C=[1,1,0]_u1_2D} and \ref{fig:N=3_alpha=1_rho=1_eps=[-1,1]_p=[2+1.0i,2,2-1.0i]_C=[1,1,0]_u2_2D} correspond to \(C_3=0\), where inelastic collisions occur between a single soliton and a breather. In this scenario, the breather transitions into a soliton after the collision.
\item Figs. \ref{fig:N=3_alpha=1_rho=1_eps=[-1,1]_p=[2+1.0i,2,2-1.0i]_C=[1,1,1]_u1_2D} and \ref{fig:N=3_alpha=1_rho=1_eps=[-1,1]_p=[2+1.0i,2,2-1.0i]_C=[1,1,1]_u2_2D} correspond to \(C_3=1\), depicting an elastic collision between a single soliton and a breather, where both structures retain their forms after the interaction. In this case, the collision does not alter the shapes of the soliton or the breather, and only a phase shift occurs as a result of the interaction.
\end{itemize}
\begin{figure}[H]
    \centering
\subfigure[]{
        \includegraphics[width=30mm]{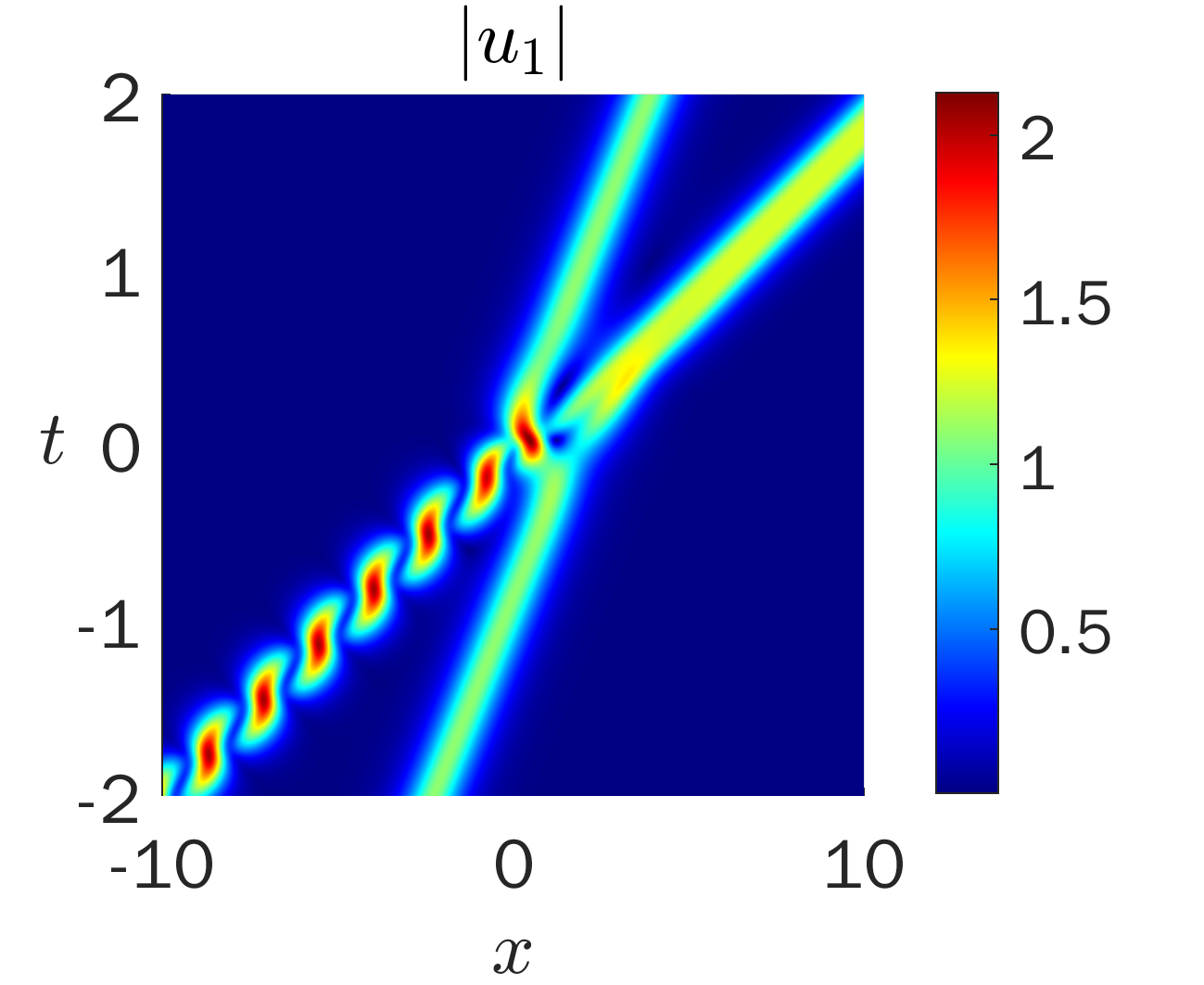}\label{fig:N=3_alpha=1_rho=1_eps=[-1,1]_p=[2+1.0i,2,2-1.0i]_C=[1,1,0]_u1_2D}}
\subfigure[]{ 
        \includegraphics[width=30mm]{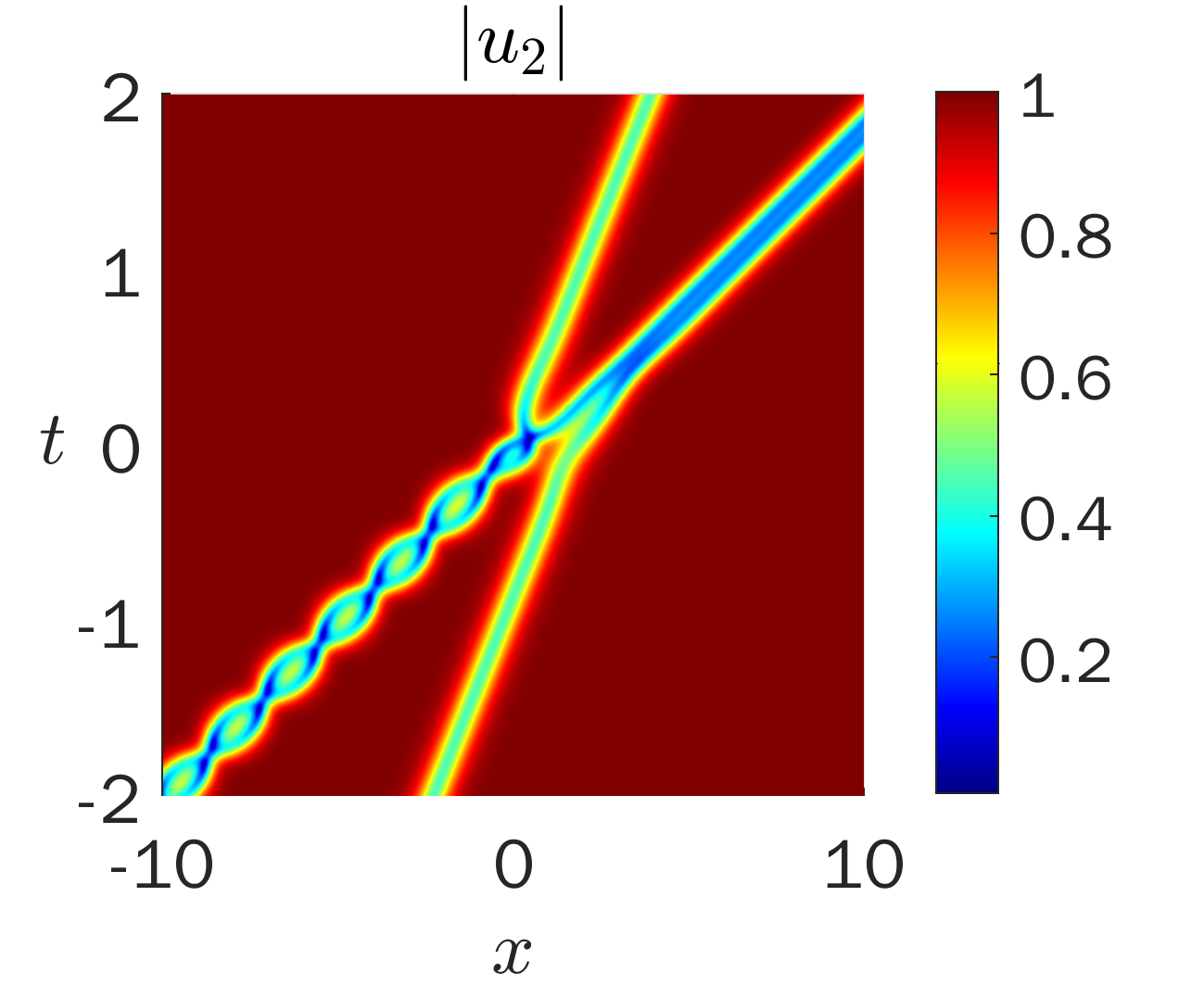}\label{fig:N=3_alpha=1_rho=1_eps=[-1,1]_p=[2+1.0i,2,2-1.0i]_C=[1,1,0]_u2_2D}}
\subfigure[]{
        \includegraphics[width=30mm]{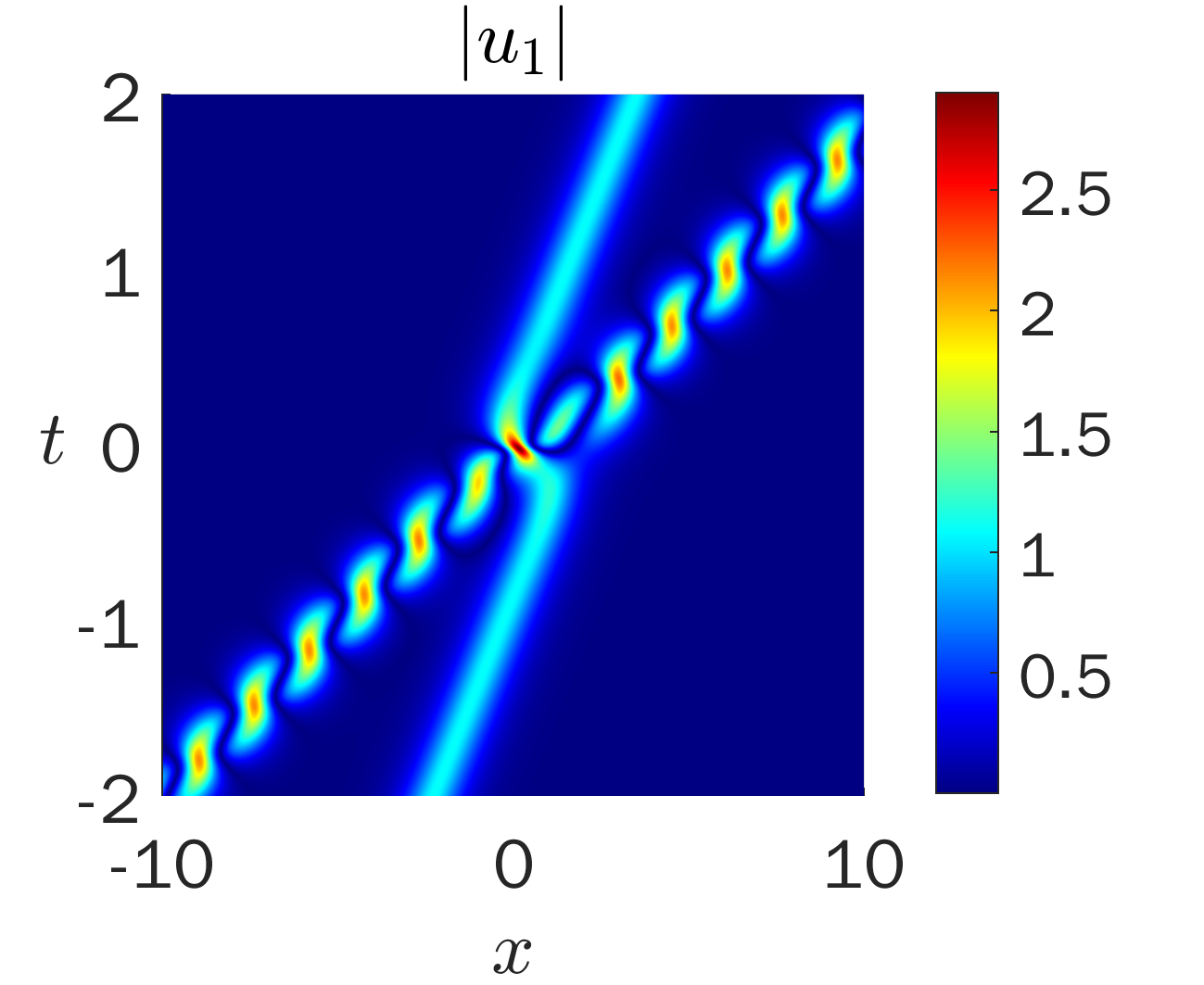}\label{fig:N=3_alpha=1_rho=1_eps=[-1,1]_p=[2+1.0i,2,2-1.0i]_C=[1,1,1]_u1_2D}}
\subfigure[]{ 
        \includegraphics[width=30mm]{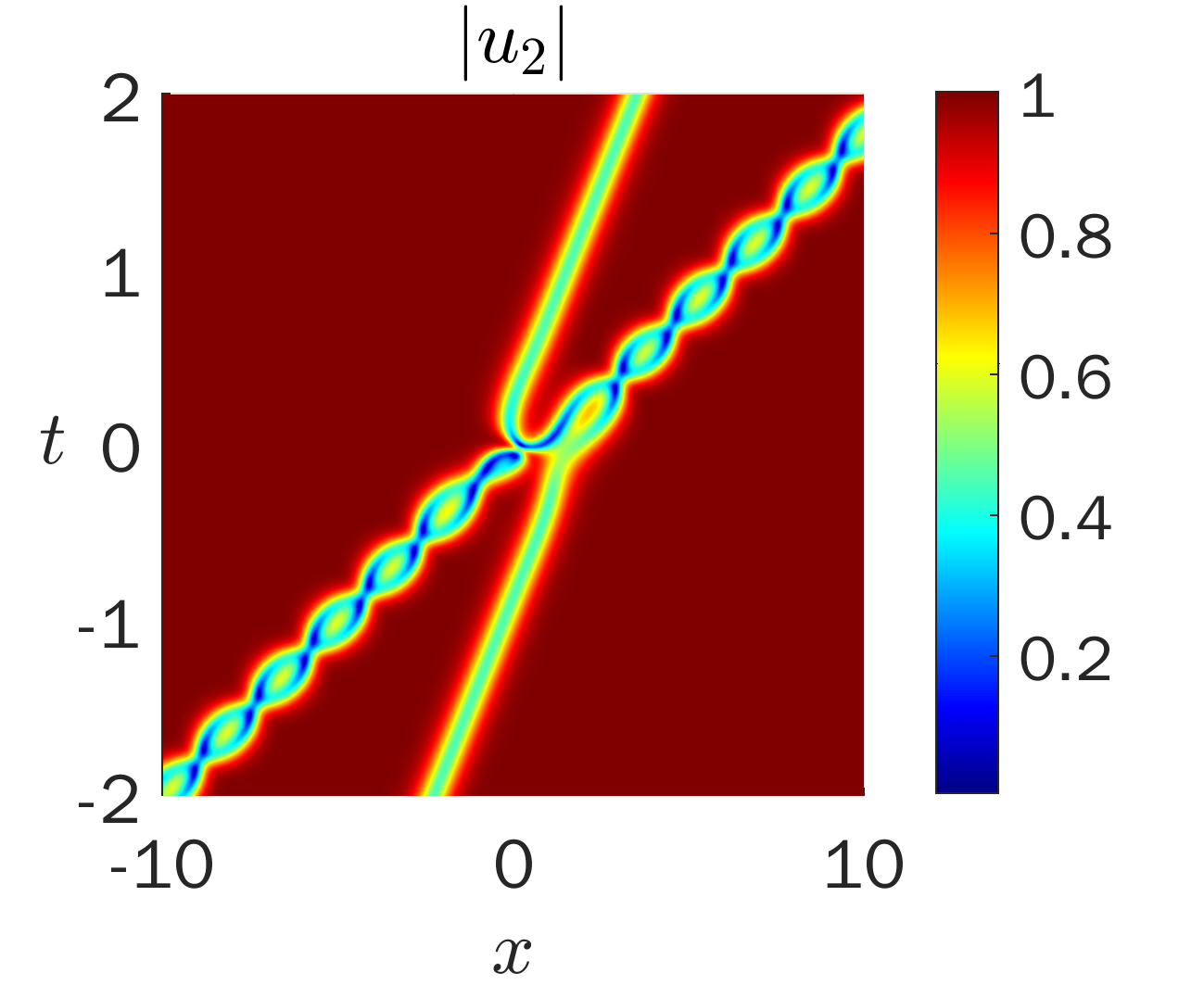}\label{fig:N=3_alpha=1_rho=1_eps=[-1,1]_p=[2+1.0i,2,2-1.0i]_C=[1,1,1]_u2_2D}}
        \caption{Two bright-dark soliton solutions to the coupled Sasa-Satsuma equation under parameters \( N=3, \alpha=1, \rho=1,  \epsilon_1=-1, \epsilon_2=1, p_1=2+\mathrm{i}, p_2=2\), \(C_1=C_2=1,\xi_{1,0}=\xi_{2,0}=\xi_{3,0}=0\) with (a)-(b) \(C_3=0\), and (c)-(d) \(C_3=1\).}
    \label{fig:two bright-dark soliton for N=3, case1}
\end{figure}

Moreover, we can also observe two parallel solitons. Note that solitons are indicated by lines \(\xi_1 + \xi_3 = 0\) and \(\xi_2 = 0\), where
\begin{equation}
    \xi_1+\xi_3=2 \Re\left(p_1\right) x+ \left(p_1^3+p_3^3 -12 \epsilon_2 \rho^2 \Re\left(p_1\right)\right)t+ \xi_{1,0}+\xi_{3,0}, 
\end{equation}
and
\begin{equation}
    \xi_2=p_2 x+ \left(p_2^3 -6 \epsilon_2 \rho^2 p_2\right)t+ \xi_{2,0}.
\end{equation}
Hence, if we take
\begin{equation}
    \alpha=2,\quad \rho=1, \quad \epsilon_1=\epsilon_2=-1, \quad p_1=2+\mathrm{i},\quad p_2=1,
\end{equation}
we have \(\xi_1 + \xi_3 = 4 \xi_2\), which means the solitons have the same speed. This kind of solution is known as the bound state soliton. As illustrated in Fig. \ref{fig:two bright-dark soliton for N=3, case2}, with \(C_1 = 0, C_2 = 1, C_3 = 1\), there are bound states between two bright-dark solitons 
(see \cref{fig:n3c2a,fig:n3c2b}).

\begin{figure}[H]
    \centering
\subfigure[]{\label{fig:n3c2a}  
        \includegraphics[width=30mm]{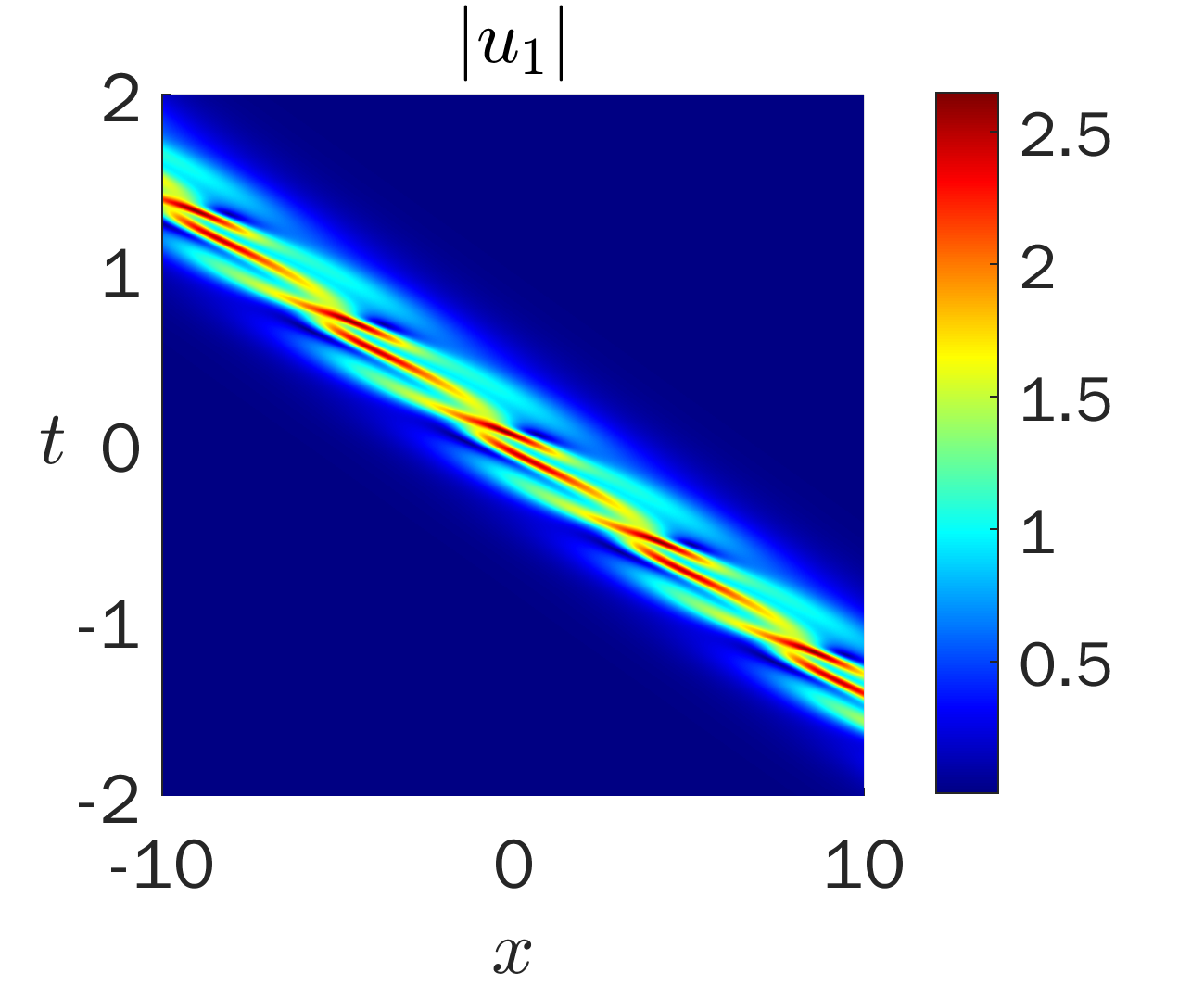}}
\subfigure[]{\label{fig:n3c2a3D}  
        \includegraphics[width=30mm]{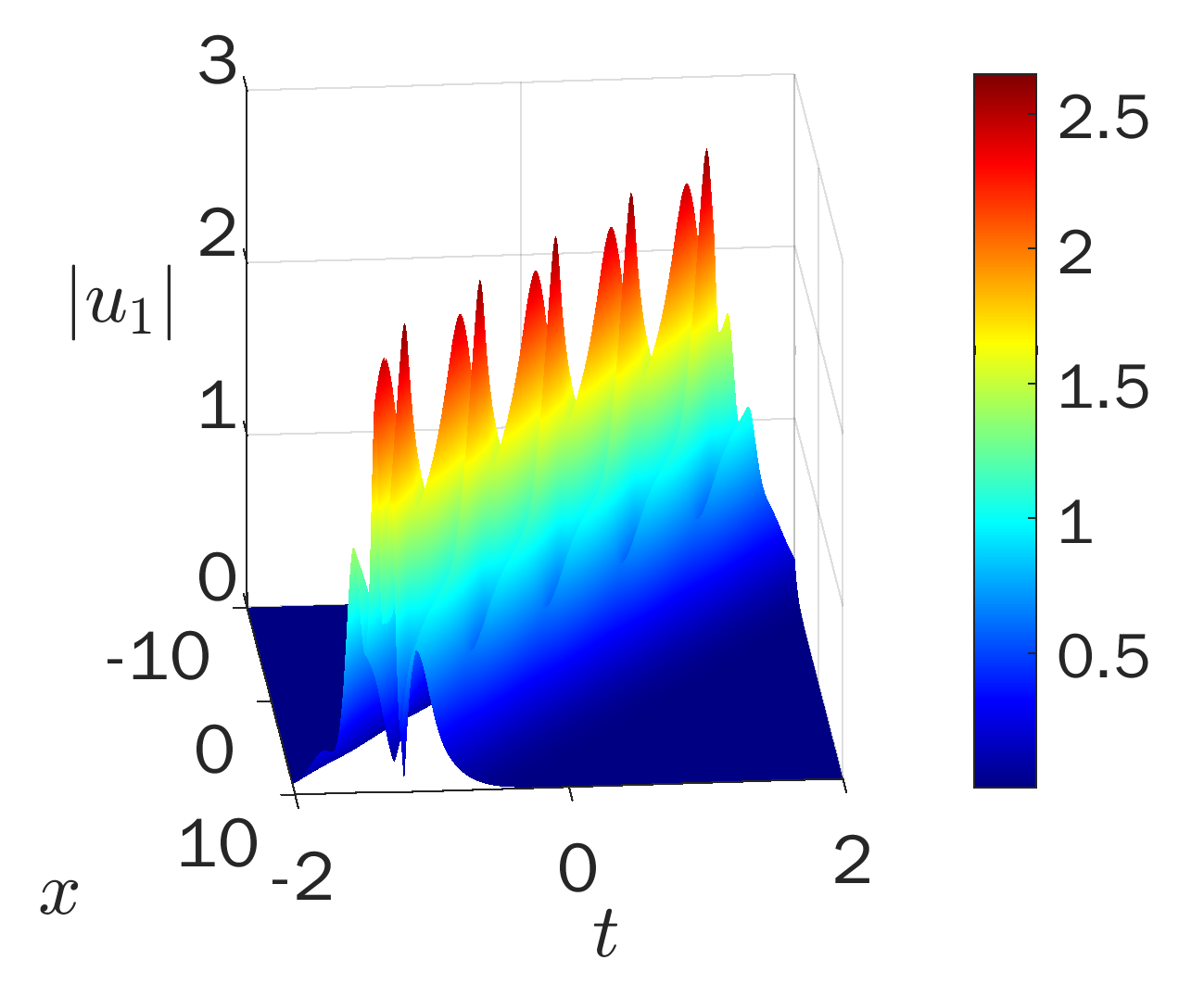}}
\subfigure[]{ \label{fig:n3c2b}     
        \includegraphics[width=30mm]{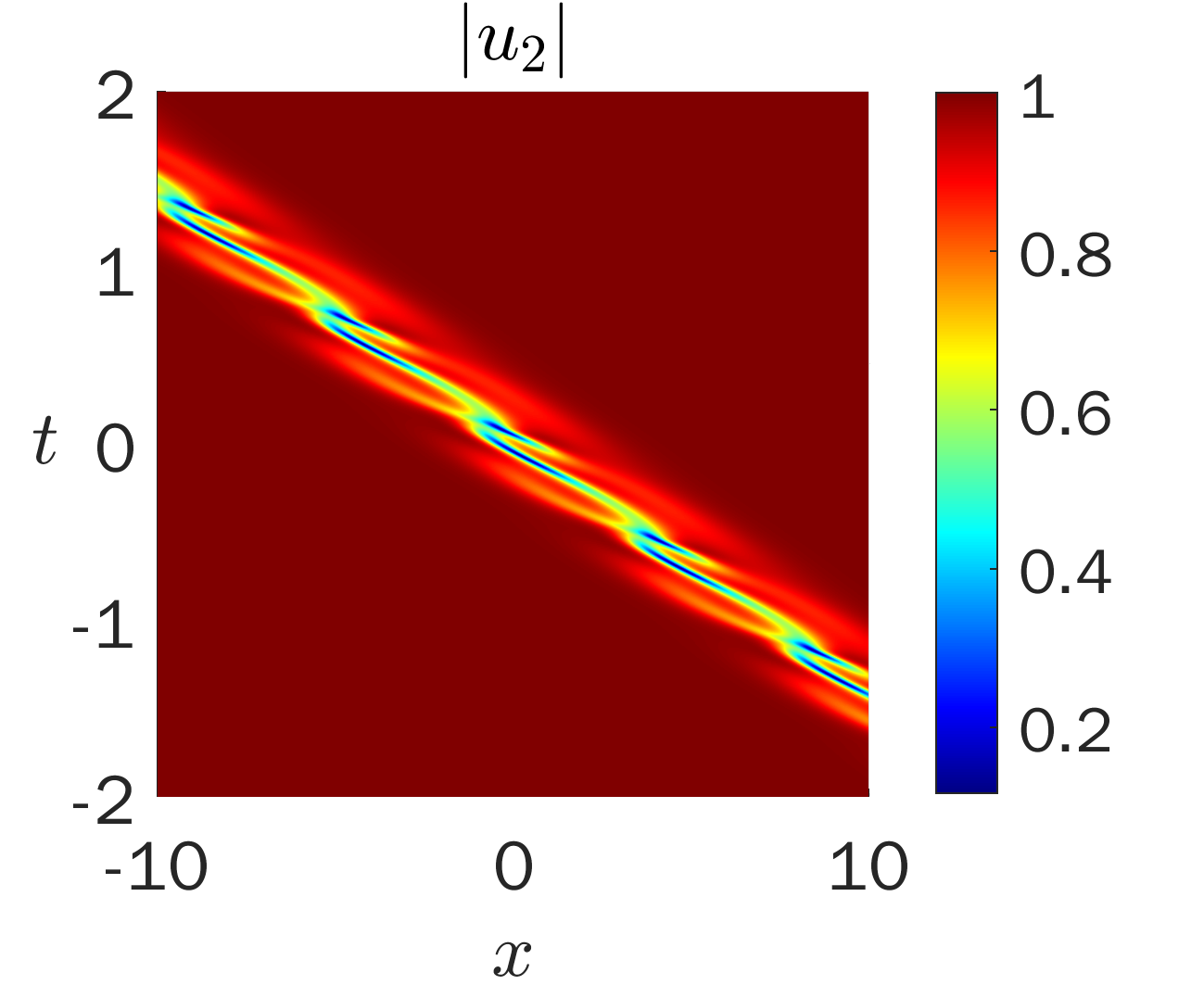}}
\subfigure[]{ \label{fig:n3c2b3D}     
        \includegraphics[width=30mm]{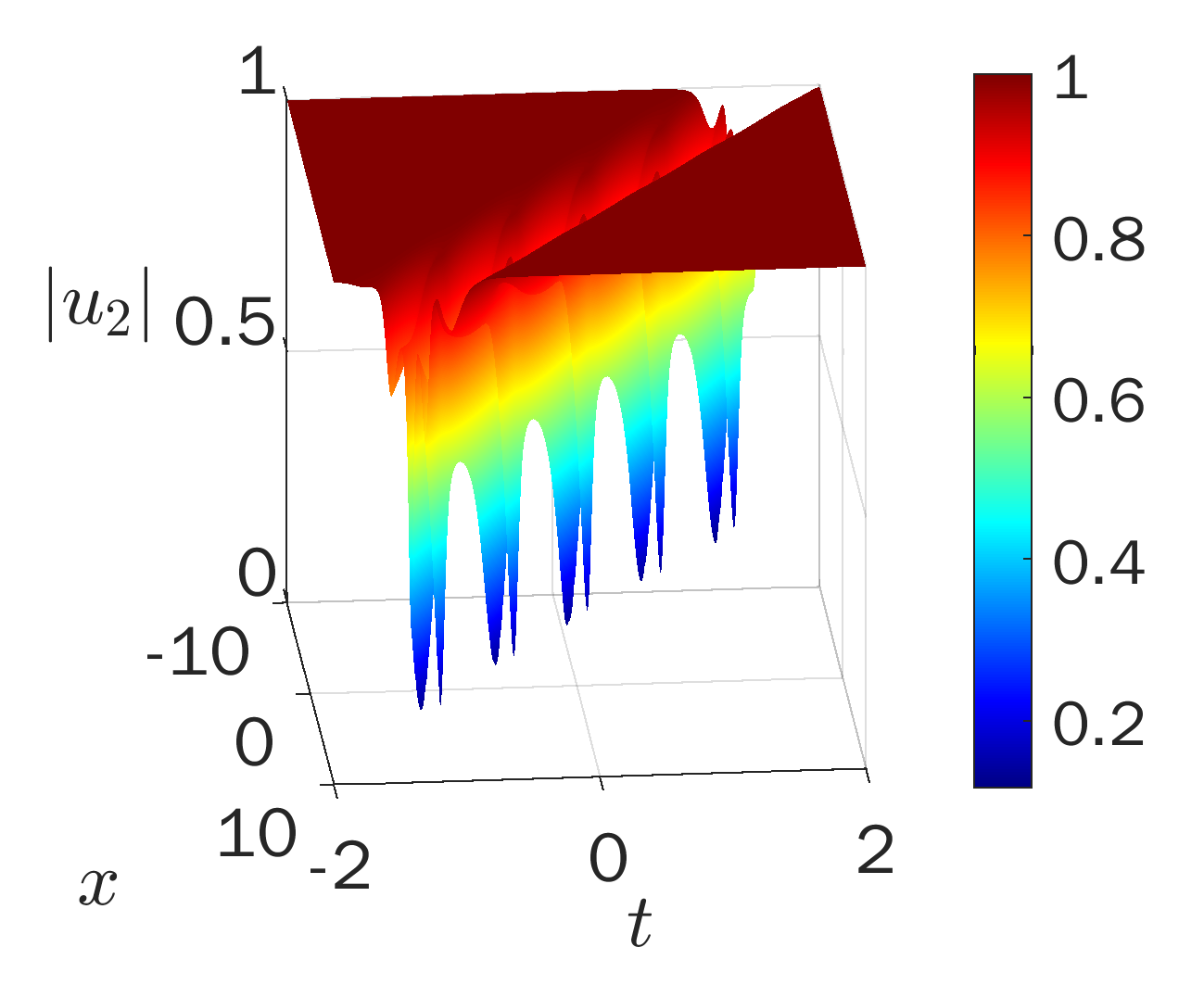}}
\caption{Bound states between two bright-dark soliton solutions to the coupled Sasa-Satsuma equation under parameters \( N=3, \alpha=2, \rho=1,  \epsilon_1= \epsilon_2=-1, p_1=2+\mathrm{i}, p_2=1\), 
\(C_1=0, C_2 = 1, C_3=1,\xi_{1,0}=\xi_{2,0}=\xi_{3,0}=0\).
}
\label{fig:two bright-dark soliton for N=3, case2}
\end{figure}

\subsection{Two bright-dark solitons for $N=4$}
When $N=4$ in \cref{thm:bd_css}, two bright-dark solitons to the coupled Sasa-Satsuma equation can be obtained, and the parameter restrictions \eqref{cmp_conj_rest} are
\begin{equation}\label{n4res}
    p_1 = p_4^*, \quad p_2 = p_3^*, \quad \xi_{1,0} = \xi_{4,0}^*, \quad \xi_{2,0} = \xi_{3,0}^*.
\end{equation}
Specifically, we set
\begin{equation}
    \alpha=1,\quad \rho=1, \quad \epsilon_1=-1, \quad \epsilon_2=1, \quad p_1=2.5-\mathrm{i},\quad p_2=3-\mathrm{i},
\end{equation}
then some figures can be obtained by setting different values of $C_i, i=1,2,3,4,$ as depicted in Fig. \ref{fig:two bright-dark soliton for N=4, case1}.

\begin{figure}[H]
    \centering
\subfigure[]{
        \includegraphics[width=30mm]{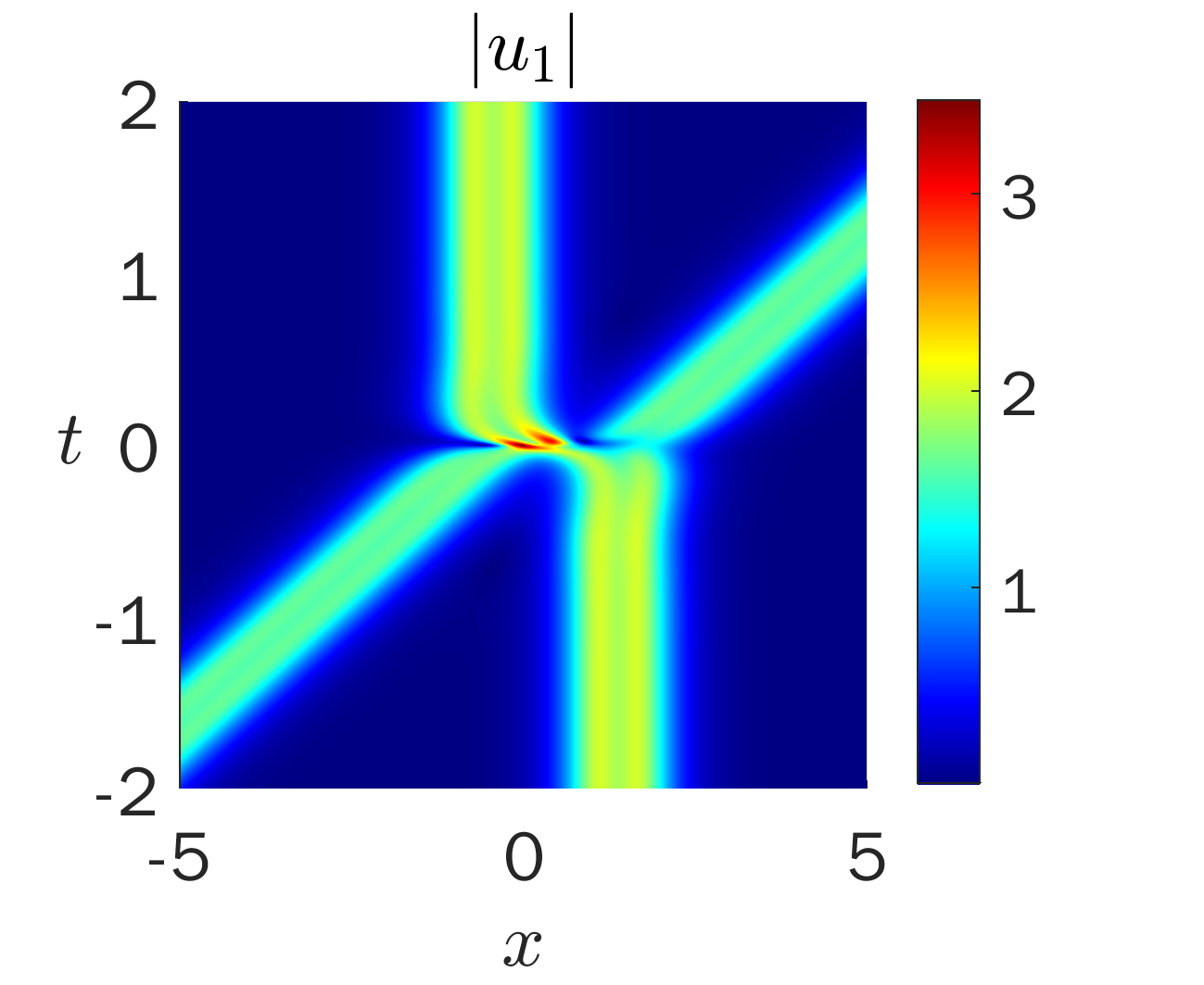}\label{fig:N=4_alpha=1_rho=1_eps=[-1,1]_p=[2.5+1.0i,3-1.0i,3+1.0i,2.5-1.0i]_C=[0,1,0,1]_u1_2D}}
\subfigure[]{ 
        \includegraphics[width=30mm]{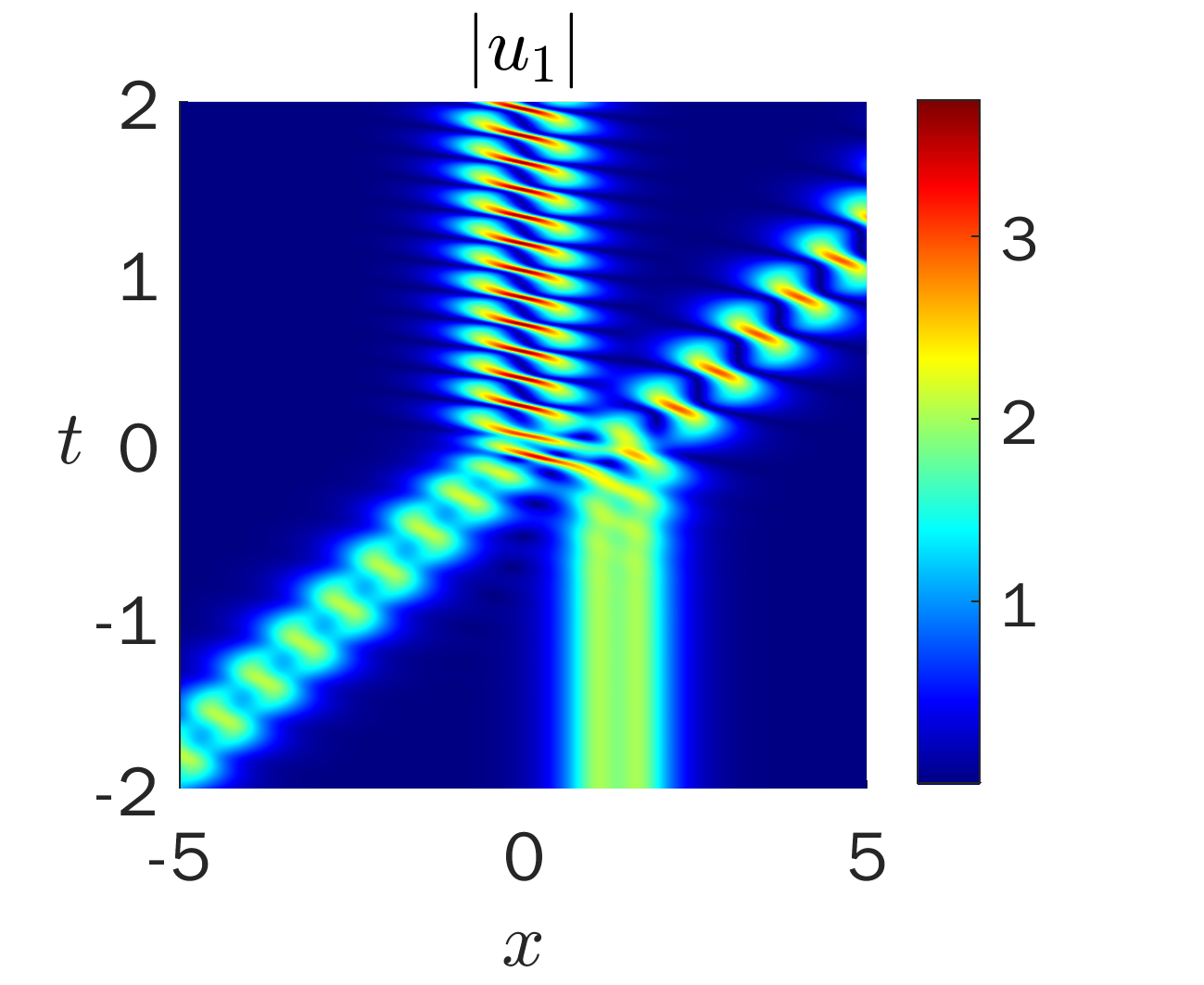}\label{fig:N=4_alpha=1_rho=1_eps=[-1,1]_p=[2.5+1.0i,3-1.0i,3+1.0i,2.5-1.0i]_C=[1,1,0,1]_u1_2D}}
\subfigure[]{
        \includegraphics[width=30mm]{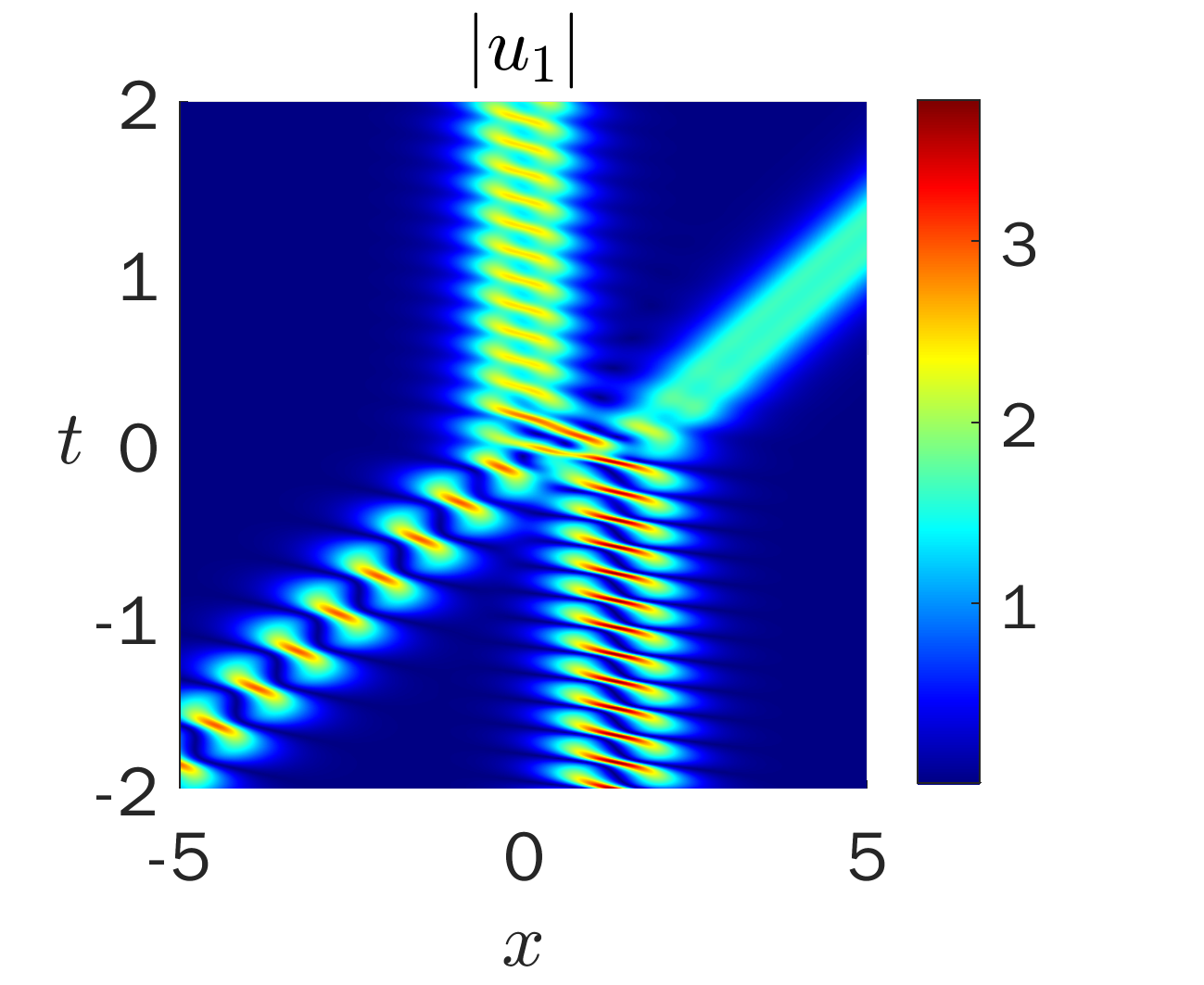}\label{fig:N=4_alpha=1_rho=1_eps=[-1,1]_p=[2.5+1.0i,3-1.0i,3+1.0i,2.5-1.0i]_C=[0,1,1,1]_u1_2D}}
\subfigure[]{ 
        \includegraphics[width=30mm]{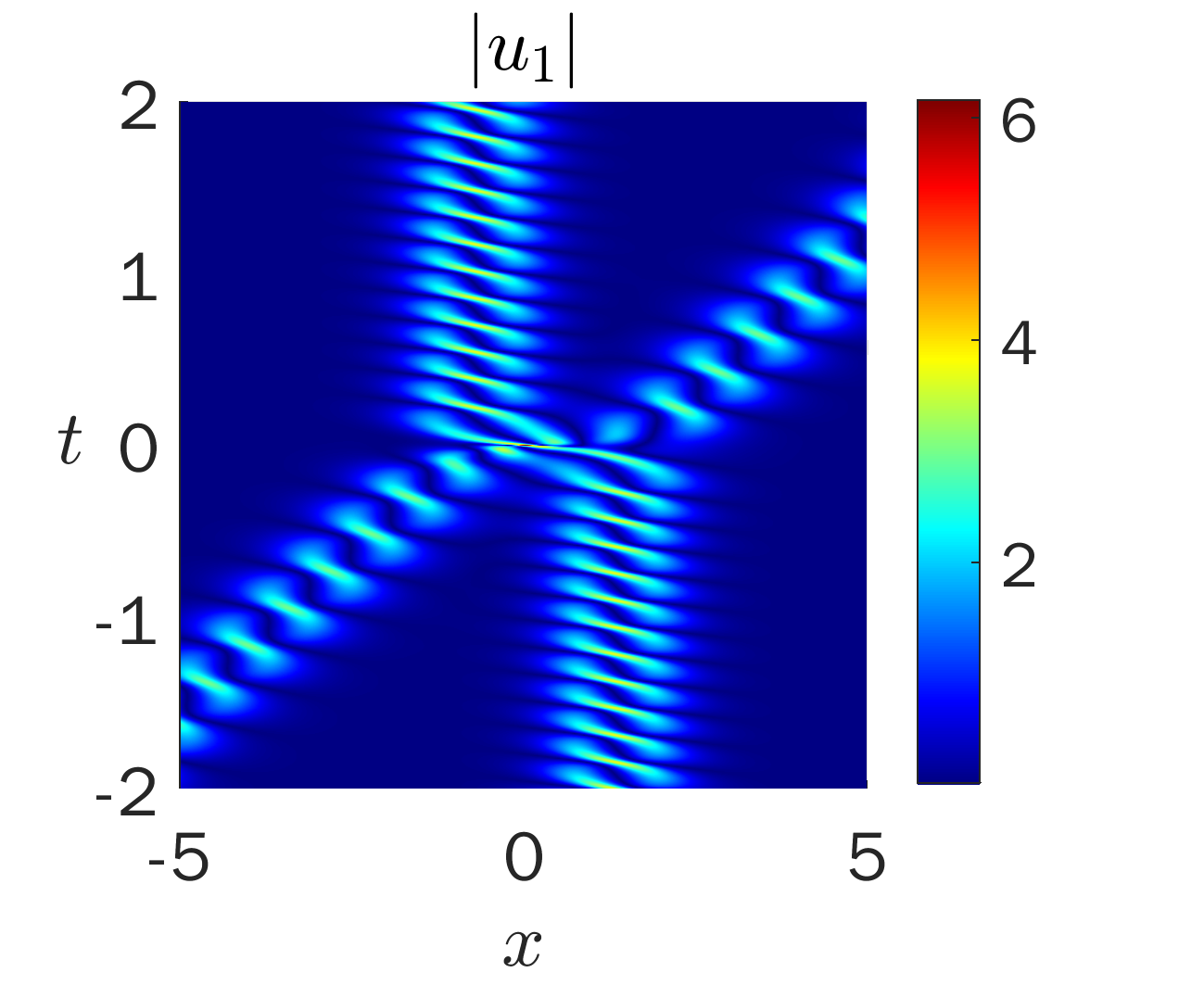}\label{fig:N=4_alpha=1_rho=1_eps=[-1,1]_p=[2.5+1.0i,3-1.0i,3+1.0i,2.5-1.0i]_C=[1,1,1,1]_u1_2D}}
\subfigure[]{
        \includegraphics[width=30mm]{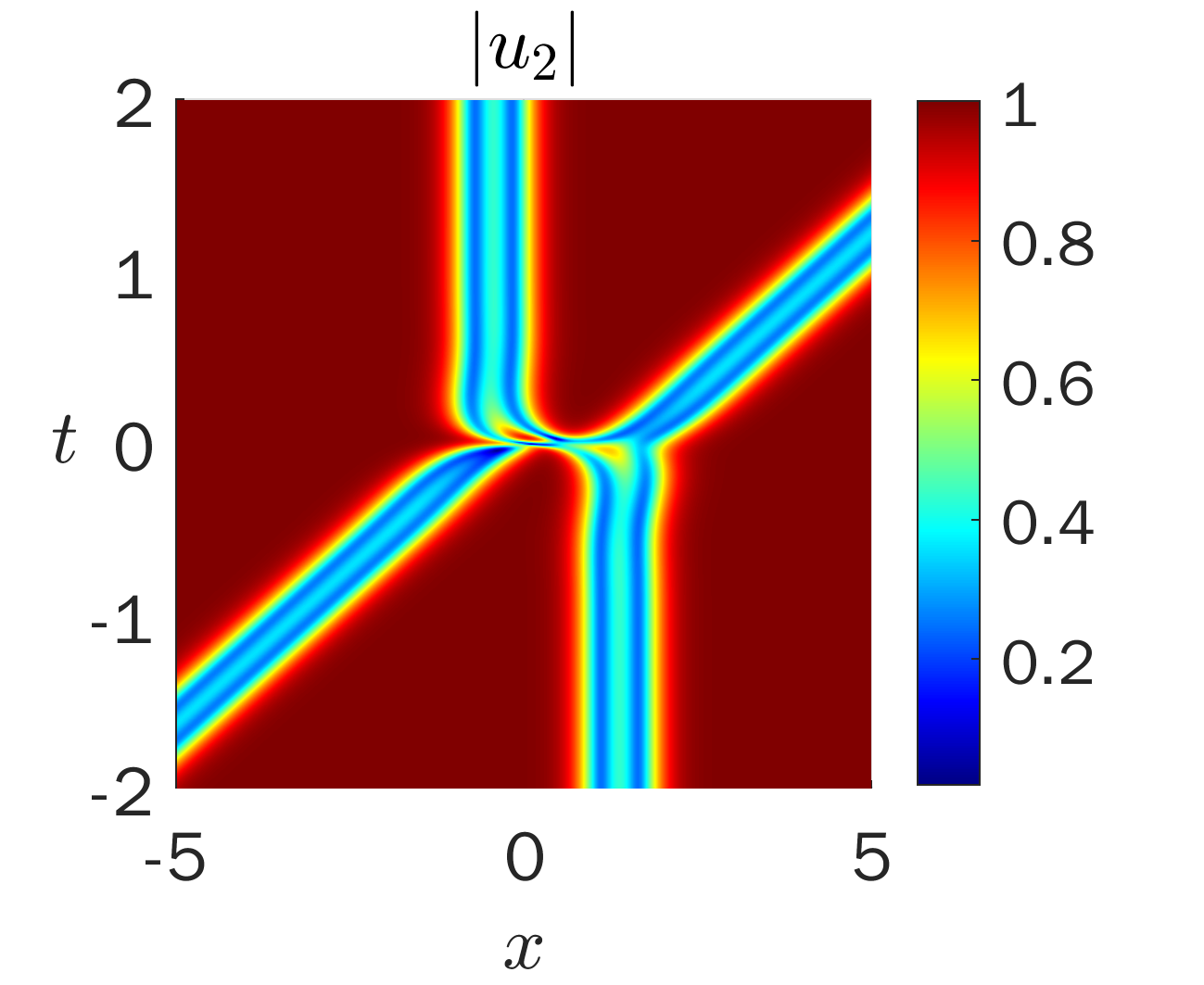}\label{fig:N=4_alpha=1_rho=1_eps=[-1,1]_p=[2.5+1.0i,3-1.0i,3+1.0i,2.5-1.0i]_C=[0,1,0,1]_u2_2D}}
\subfigure[]{ 
        \includegraphics[width=30mm]{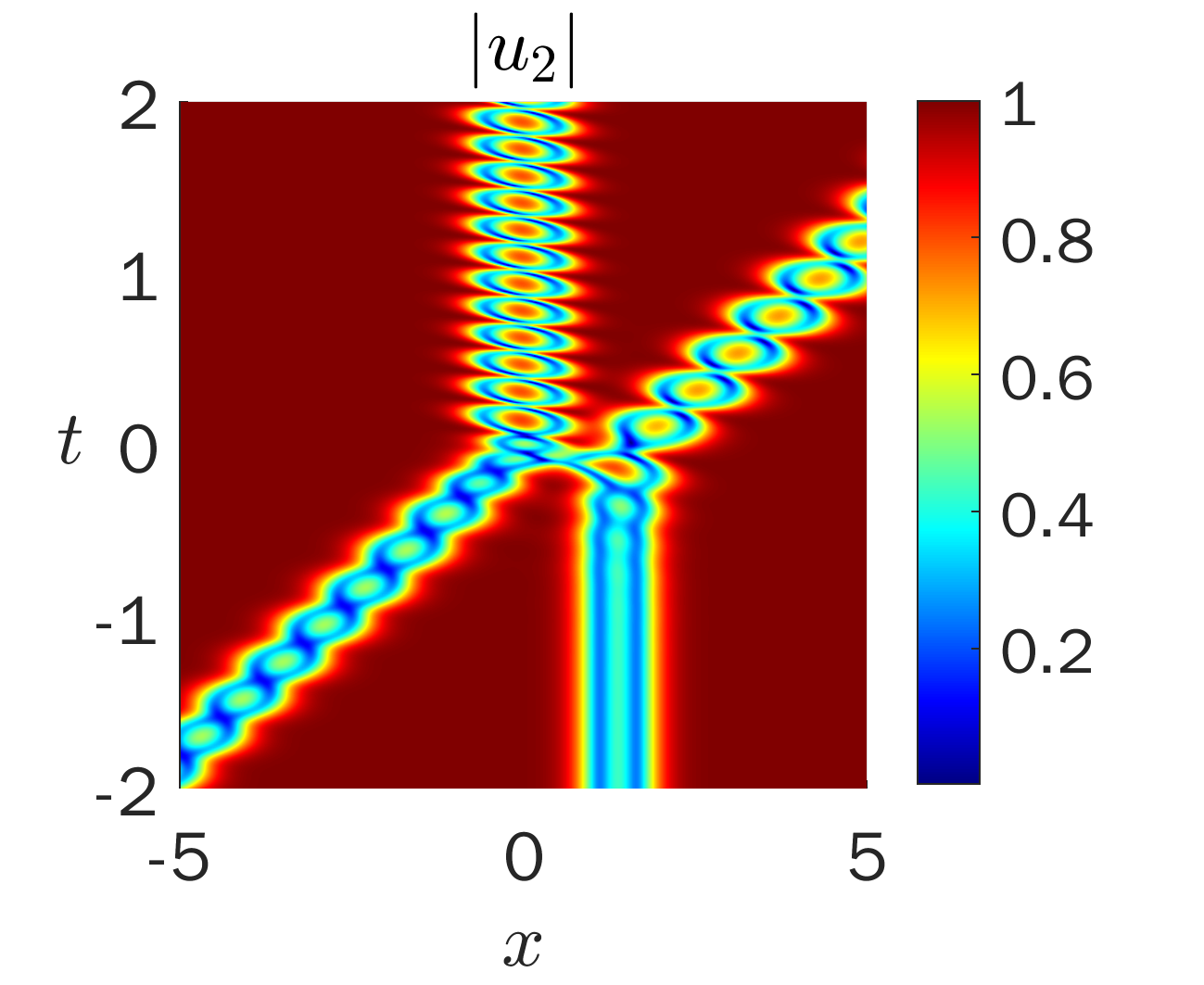}\label{fig:N=4_alpha=1_rho=1_eps=[-1,1]_p=[2.5+1.0i,3-1.0i,3+1.0i,2.5-1.0i]_C=[1,1,0,1]_u2_2D}}
\subfigure[]{
        \includegraphics[width=30mm]{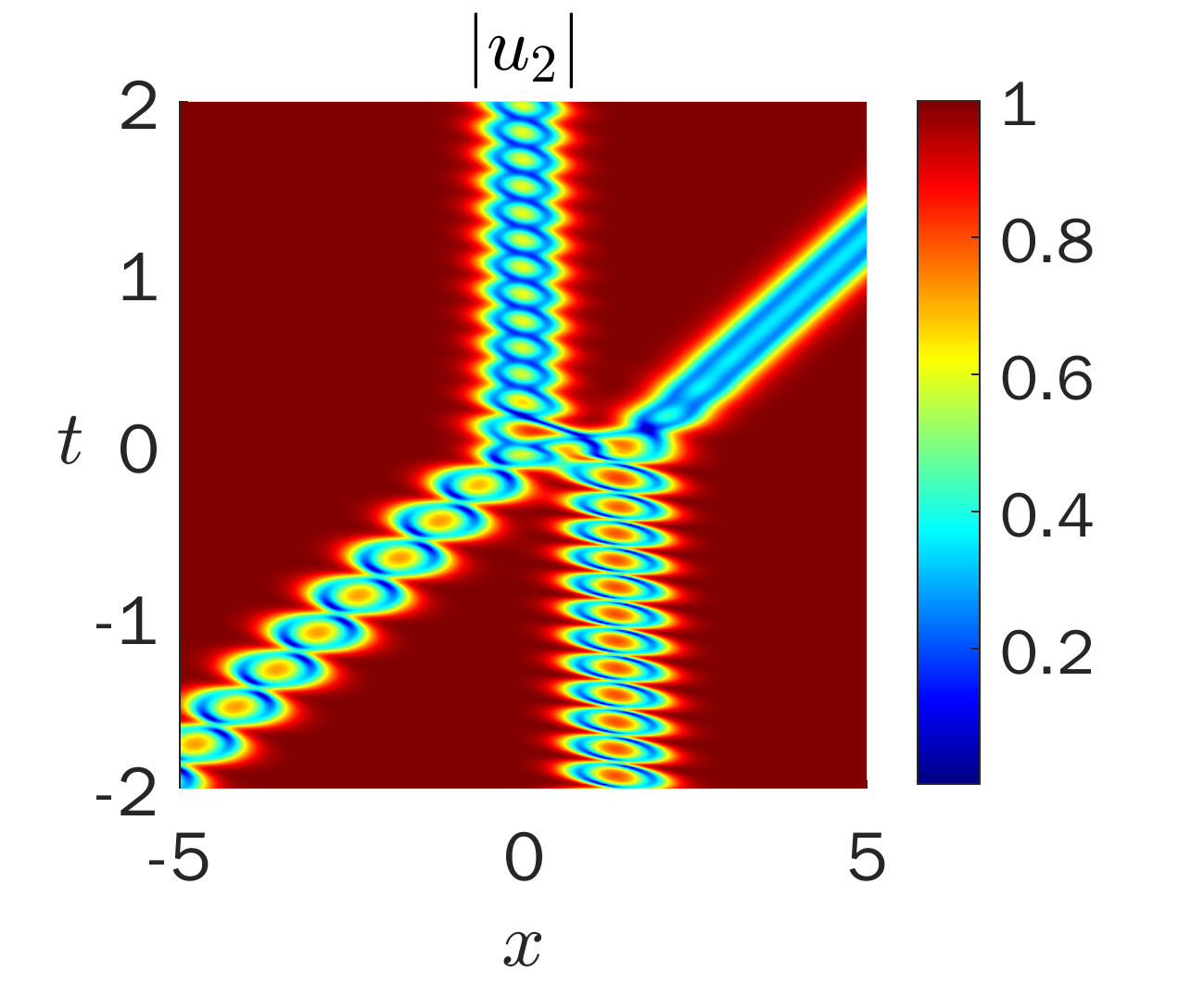}\label{fig:N=4_alpha=1_rho=1_eps=[-1,1]_p=[2.5+1.0i,3-1.0i,3+1.0i,2.5-1.0i]_C=[0,1,1,1]_u2_2D}}
\subfigure[]{ 
        \includegraphics[width=30mm]{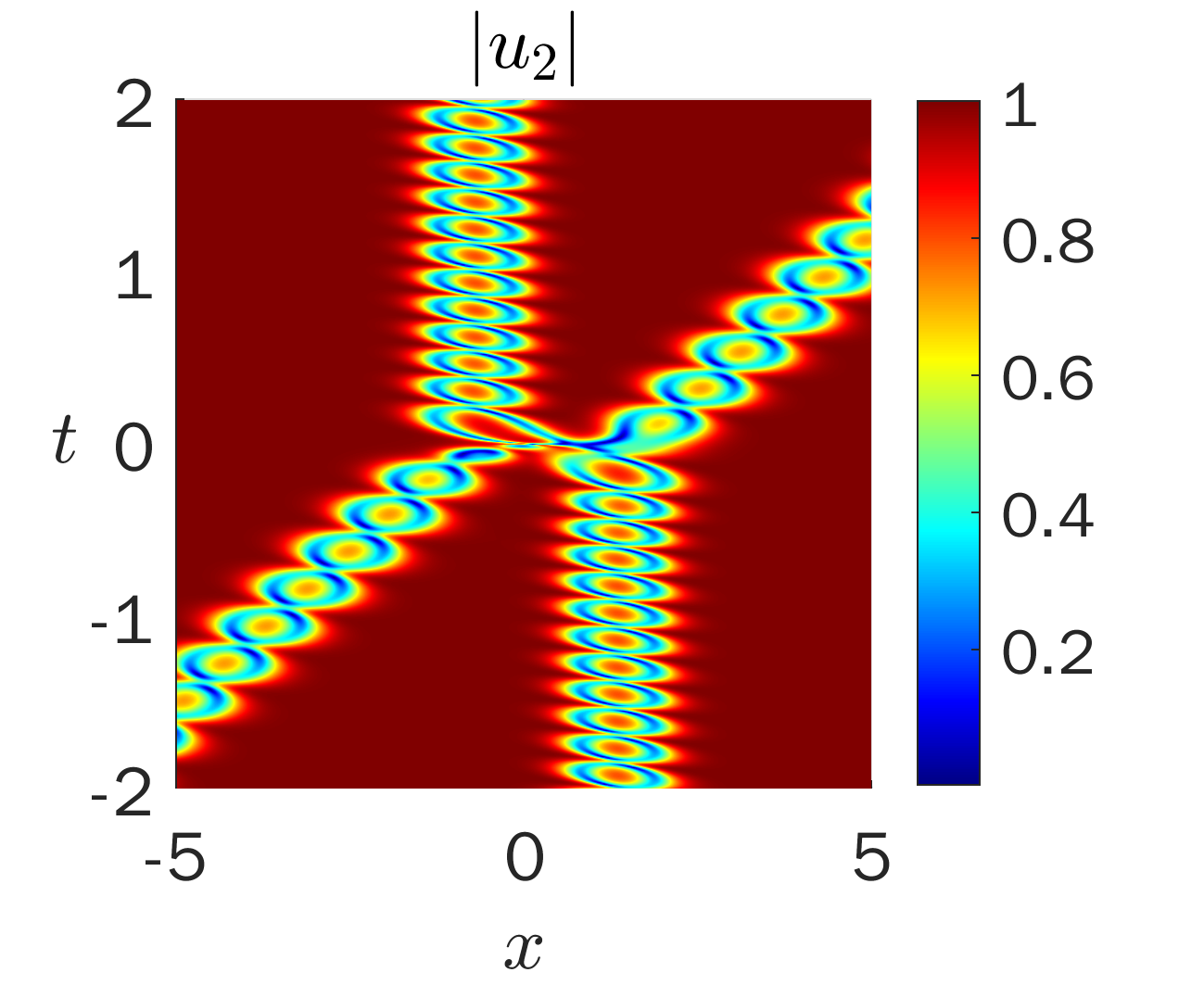}\label{fig:N=4_alpha=1_rho=1_eps=[-1,1]_p=[2.5+1.0i,3-1.0i,3+1.0i,2.5-1.0i]_C=[1,1,1,1]_u2_2D}}
    \caption{Two bright-dark soliton solutions to the coupled Sasa-Satsuma equation under parameters \( N=4, \alpha=1, \rho=1,  \epsilon_1=-1, \epsilon_2=1, p_1=2.5-\mathrm{i}, p_2=3-\mathrm{i}\), \(C_2=C_4=1,\xi_{1,0}=\xi_{2,0}=\xi_{3,0}=\xi_{4,0}=0\) with (a) and (e): \(C_1=C_3=0\); (b) and (f): \(C_1=1,C_3=0\); (c) and (g): \(C_1=0,C_3=1\); (d) and (h): \(C_1=C_3=1\).}
    \label{fig:two bright-dark soliton for N=4, case1}
\end{figure}

\begin{itemize}
\item Figs. \ref{fig:N=4_alpha=1_rho=1_eps=[-1,1]_p=[2.5+1.0i,3-1.0i,3+1.0i,2.5-1.0i]_C=[0,1,0,1]_u1_2D} and \ref{fig:N=4_alpha=1_rho=1_eps=[-1,1]_p=[2.5+1.0i,3-1.0i,3+1.0i,2.5-1.0i]_C=[0,1,0,1]_u2_2D}: \(C_1=C_3=0\), demonstrating elastic collisions between two single solitons. The solitons retain their shapes after the interaction.  
\item Figs. \ref{fig:N=4_alpha=1_rho=1_eps=[-1,1]_p=[2.5+1.0i,3-1.0i,3+1.0i,2.5-1.0i]_C=[1,1,0,1]_u1_2D} and \ref{fig:N=4_alpha=1_rho=1_eps=[-1,1]_p=[2.5+1.0i,3-1.0i,3+1.0i,2.5-1.0i]_C=[1,1,0,1]_u2_2D}: \(C_1=1, C_3=0\), illustrating shape-changing collisions between a single soliton and a breather. After the interaction, the system evolves into two distinct new breathers.  
\item Figs. \ref{fig:N=4_alpha=1_rho=1_eps=[-1,1]_p=[2.5+1.0i,3-1.0i,3+1.0i,2.5-1.0i]_C=[0,1,1,1]_u1_2D} and \ref{fig:N=4_alpha=1_rho=1_eps=[-1,1]_p=[2.5+1.0i,3-1.0i,3+1.0i,2.5-1.0i]_C=[0,1,1,1]_u2_2D}: \(C_1=0, C_3=1\), showing shape-changing collisions between two breathers. The interaction leads to the formation of a single bright-dark soliton and a new breather.  
\item Figs. \ref{fig:N=4_alpha=1_rho=1_eps=[-1,1]_p=[2.5+1.0i,3-1.0i,3+1.0i,2.5-1.0i]_C=[1,1,1,1]_u1_2D} and \ref{fig:N=4_alpha=1_rho=1_eps=[-1,1]_p=[2.5+1.0i,3-1.0i,3+1.0i,2.5-1.0i]_C=[1,1,1,1]_u2_2D}: \(C_2=C_3=1\), highlighting elastic collisions between two breathers, where both breathers retain their original forms after the collision.  
\end{itemize}

Furthermore, for the \(N=4\) case, we can also identify bound states of two bright-dark solitons or bound states formed by a bright-dark soliton and a breather. Fig. \ref{fig:two bright-dark soliton for N=4, case2} presents two such examples, generated using the following parameters
\begin{equation}\label{n4bdstate_parameters}
\begin{aligned}
    &N=4, \quad \alpha=2, \quad\rho=1, \quad\epsilon_1=\epsilon_2=-1, \quad p_1=1-\mathrm{i}, \quad p_2=2-\sqrt{2}\mathrm{i}, \\
    &C_1 = 0, \quad C_2 = 1,
\quad C_4 = 1.
\end{aligned}
\end{equation}
These examples demonstrate different dynamics depending on the values of \(C_3\). Specifically, \cref{n4bdstate_a,n4bdstate_b} illustrate the dynamics when \(C_3 = 0\), corresponding to the bound state of two bright-dark solitons, while \cref{n4bdstate_c,n4bdstate_d} show the dynamics when \(C_3 = 1\), representing the bound state of the soliton-breather interaction.

\begin{figure}[H]
    \centering
    \subfigure[]{\label{n4bdstate_a}
        \includegraphics[width=30mm]{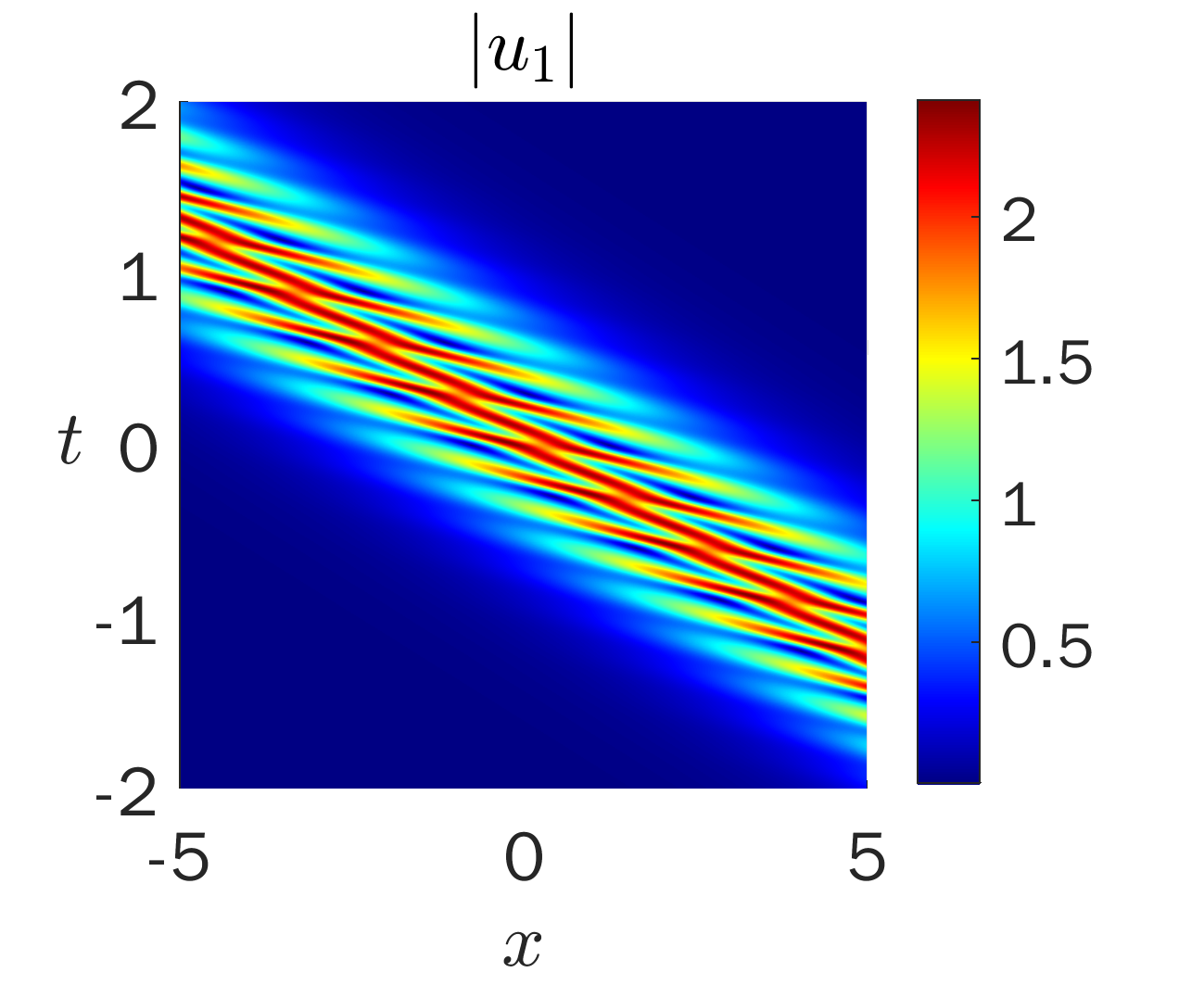}}
    \subfigure[]{\label{n4bdstate_b}
        \includegraphics[width=30mm]{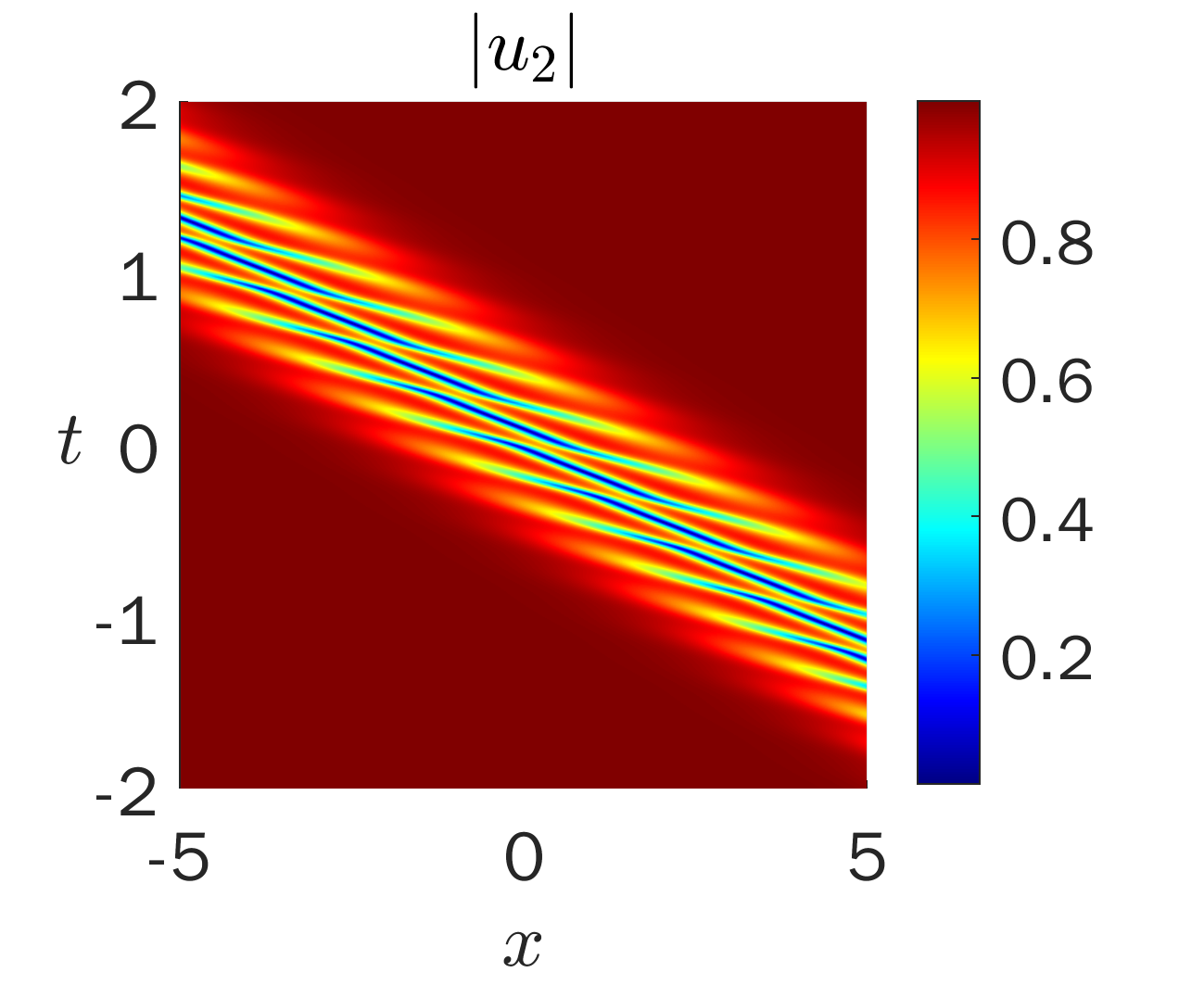}}
    \subfigure[]{\label{n4bdstate_c}
        \includegraphics[width=30mm]{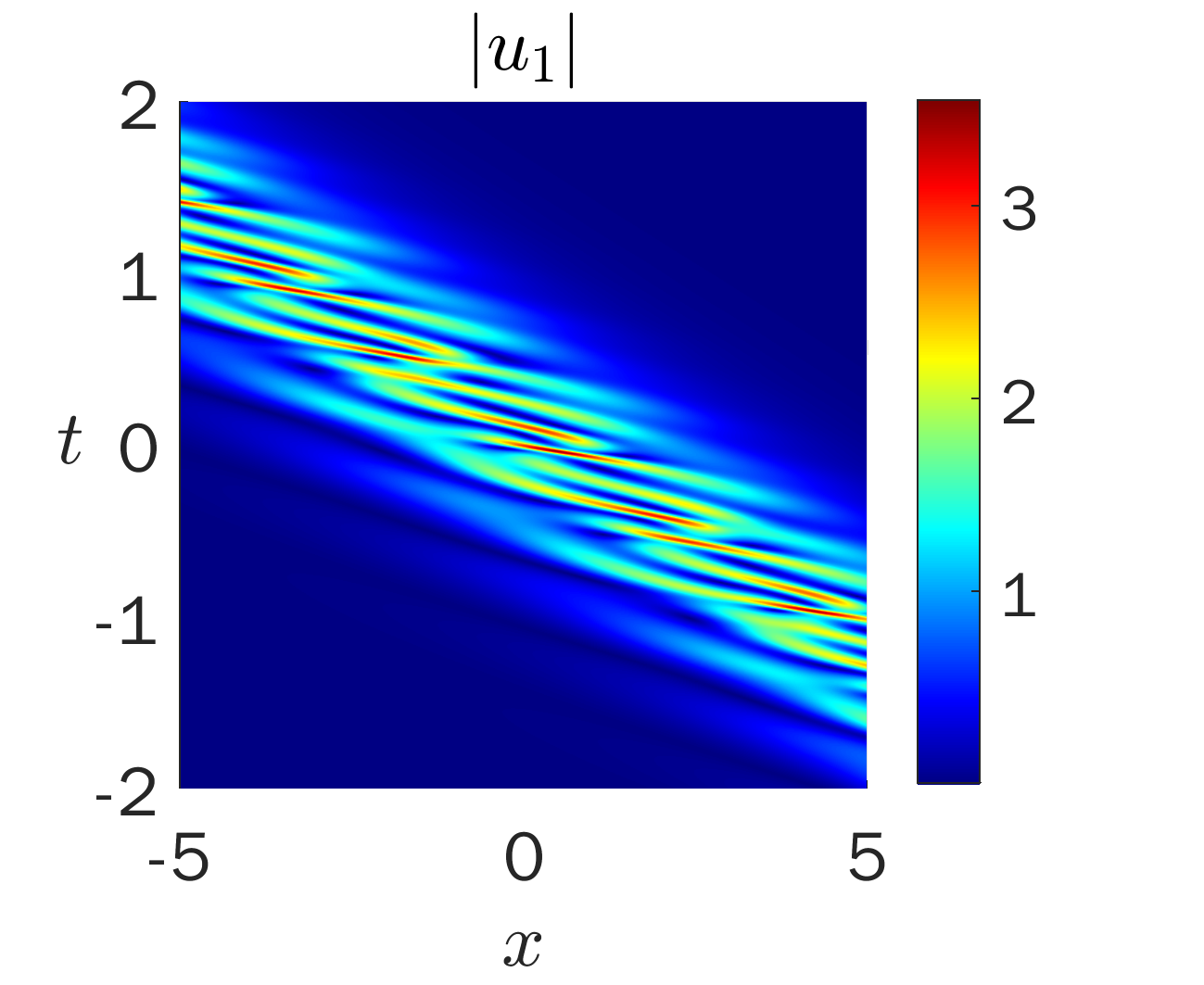}}
    \subfigure[]{\label{n4bdstate_d}
        \includegraphics[width=30mm]{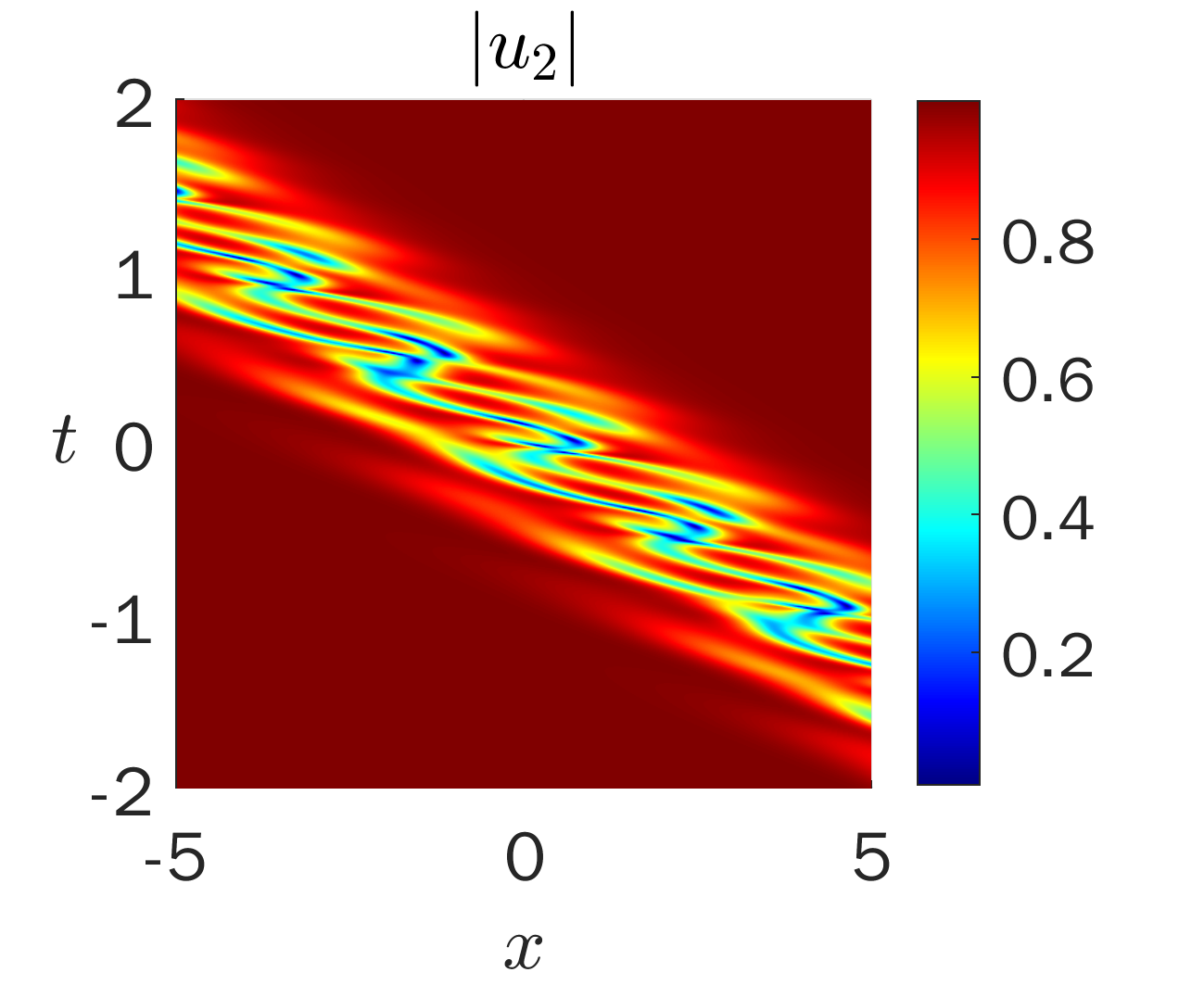}}
    \caption{Bound states between two bright-dark soliton solutions or soliton-breather interaction to the coupled Sasa-Satsuma equation under parameters \eqref{n4bdstate_parameters} with \(\xi_{1,0}=\xi_{2,0}=\xi_{3,0}=\xi_{4,0}=0\).  Panels (a) and (b)  correspond to the case \(C_3=0\), while panels  (c) and (d) illustrate the case \(C_3=1\).
    }
    \label{fig:two bright-dark soliton for N=4, case2}
\end{figure}

\section{Derivation of the bright-dark soliton solution}\label{sect:bd}

\subsection{Proof of \cref{thm:2b2d_4cmkdv}} \label{sec:deriva_2b2d_4cmkdv}

We begin by introducing a set of bilinear equations from the KP-Toda hierarchy (\cref{kp_toda}) and subsequently reduce them to the bilinear equations corresponding to the four-component Hirota equation (\cref{dim_reduction_4cmkdv,cmp_reduction_4cmkdv}). This process ultimately leads to the derivation of two-bright-two-dark soliton solutions for \eqref{4cmkdv} with $n=4$.

\begin{lemma}\label{kp_toda}
The bilinear equations
{\allowdisplaybreaks
\begin{align}
    & \left(D_{x_1}^3+3D_{x_1}D_{x_2}-4D_{x_3}\right)\tau_{k,l}^{(1)} \cdot \tau_{k,l}^{(0)}=0,
    \label{KP_bright_dark_1}\\
    \newsubeqblock
    \mysubeq & \left(D_{x_{-1}^{(1)}}\left(D_{x_1}^2 - D_{x_2}\right) - 4 (D_{x_1} -a )\right)\tau_{k,l}^{(1)}\cdot \tau_{k,l}^{(0)} - 4 a \tau_{k+1,l}^{(1)} \tau_{k-1,l}^{(0)} = 0,\label{KP_bright_dark_2_1}\\
    \mysubeq & \left(D_{x_{-1}^{(2)}}\left(D_{x_1}^2 - D_{x_2}\right) - 4 (D_{x_1} -b )\right)\tau_{k,l}^{(1)}\cdot \tau_{k,l}^{(0)} - 4 b \tau_{k,l+1}^{(1)} \tau_{k,l-1}^{(0)} = 0,\label{KP_bright_dark_2_2}\\
    \mysubeq & D_{y_1^{(1)}}\left(D_{x_1}^2-D_{x_2}\right)\tau_{k,l}^{(1)} \cdot \tau_{k,l}^{(0)}=0,\label{KP_bright_dark_2_3}\\
    \mysubeq & D_{y_1^{(2)}}\left(D_{x_1}^2-D_{x_2}\right)\tau_{k,l}^{(1)} \cdot \tau_{k,l}^{(0)}= -4 \tau_{k,l}^{(1,2)} \bar{\tau}_{k,l}^{(2)},\label{KP_bright_dark_2_4}\\
    & \left(D_{x_1}^3+3D_{x_1}D_{x_2}-4D_{x_3}\right)\tau_{k,l}^{(2)} \cdot \tau_{k,l}^{(0)}=0,
    \label{KP_bright_dark_3}\\
    \newsubeqblock 
    \mysubeq & \left(D_{x_{-1}^{(1)}}\left(D_{x_1}^2 - D_{x_2}\right) - 4 (D_{x_1} -a )\right)\tau_{k,l}^{(2)}\cdot \tau_{k,l}^{(0)} - 4 a \tau_{k+1,l}^{(2)} \tau_{k-1,l}^{(0)} = 0,\label{KP_bright_dark_4_1}\\
    \mysubeq & \left(D_{x_{-1}^{(2)}}\left(D_{x_1}^2 - D_{x_2}\right) - 4 (D_{x_1} -b )\right)\tau_{k,l}^{(2)}\cdot \tau_{k,l}^{(0)} - 4 b \tau_{k,l+1}^{(2)} \tau_{k,l-1}^{(0)} = 0,\label{KP_bright_dark_4_2}\\
    \mysubeq & D_{y_1^{(1)}}\left(D_{x_1}^2-D_{x_2}\right)\tau_{k,l}^{(2)} \cdot \tau_{k,l}^{(0)}= 4 \tau_{k,l}^{(1,2)} \bar{\tau}_{k,l}^{(2)},\label{KP_bright_dark_4_3}\\
    \mysubeq & D_{y_1^{(2)}}\left(D_{x_1}^2-D_{x_2}\right)\tau_{k,l}^{(2)} \cdot \tau_{k,l}^{(0)}=0,\label{KP_bright_dark_4_4}\\
    & \left(D_{x_1}^2-D_{x_2}+2 a D_{x_1}\right) \tau_{k+1,l}^{(0)} \cdot \tau_{k,l}^{(0)}=0, \label{KP_bright_dark_5}\\
    & \left(D_{x_1}^3+3 D_{x_1} D_{x_2}-4 D_{x_3}+3 a\left(D_{x_1}^2+D_{x_2}\right)+6 a^2 D_{x_1}\right) \tau_{k+1,l}^{(0)} \cdot \tau_{k,l}^{(0)}=0, \label{KP_bright_dark_6}\\
    \newsubeqblock
    \mysubeq & \left(D_{x_{-1}^{(1)}}\left(D_{x_1}^2-D_{x_2}+2 a D_{x_1}\right)-4D_{x_1}\right) \tau_{k+1,l}^{(0)} \cdot \tau_{k,l}^{(0)} =0, \label{KP_bright_dark_7_1}\\
     & \left(D_{x_{-1}^{(2)}}\left(D_{x_1}^2-D_{x_2}+2 a D_{x_1}\right)-4\left(D_{x_1} + a - b\right) \right) \tau_{k+1,l}^{(0)} \cdot \tau_{k,l}^{(0)} \nonumber\\
    \mysubeq &\qquad+ 4(a-b)\tau_{k+1,l+1}^{(0)} \cdot \tau_{k,l-1}^{(0)} =0, \label{KP_bright_dark_7_2}\\
    \mysubeq & \left(D_{y_1^{(1)}}\left(D_{x_1}^2 - D_{x_2} + 2a D_{x_1}\right)\right)\tau_{k+1,l}^{(0)} \cdot \tau_{k,l}^{(0)} + 4 a \tau_{k+1,l}^{(1)} \bar{\tau}_{k,l}^{(1)} = 0, \label{KP_bright_dark_7_3}\\
    \mysubeq & \left(D_{y_1^{(2)}}\left(D_{x_1}^2 - D_{x_2} + 2a D_{x_1}\right)\right)\tau_{k+1,l}^{(0)} \cdot \tau_{k,l}^{(0)} + 4 a \tau_{k+1,l}^{(2)} \bar{\tau}_{k,l}^{(2)} = 0, \label{KP_bright_dark_7_4}\\
    & \left(D_{x_1}^2-D_{x_2}+2 b D_{x_1}\right) \tau_{k,l+1}^{(0)} \cdot \tau_{k,l}^{(0)}=0, \label{KP_bright_dark_8}\\
    & \left(D_{x_1}^3+3 D_{x_1} D_{x_2}-4 D_{x_3}+3 b\left(D_{x_1}^2+D_{x_2}\right)+6 b^2 D_{x_1}\right) \tau_{k,l+1}^{(0)} \cdot \tau_{k,l}^{(0)}=0, \label{KP_bright_dark_9}\\
    \newsubeqblock
     & \left(D_{x_{-1}^{(1)}}\left(D_{x_1}^2-D_{x_2}+2 b D_{x_1}\right)-4\left(D_{x_1} + b - a\right) \right) \tau_{k,l+1}^{(0)} \cdot \tau_{k,l}^{(0)} \nonumber\\
    \mysubeq & \qquad+ 4(b-a)\tau_{k+1,l+1}^{(0)} \cdot \tau_{k-1,l}^{(0)} =0, \label{KP_bright_dark_10_1}\\
    \mysubeq & \left(D_{x_{-1}^{(2)}}\left(D_{x_1}^2-D_{x_2}+2 b D_{x_1}\right)-4D_{x_1}\right) \tau_{k,l+1}^{(0)} \cdot \tau_{k,l}^{(0)} =0, \label{KP_bright_dark_10_2}\\
    \mysubeq & \left(D_{y_1^{(1)}}\left(D_{x_1}^2 - D_{x_2} + 2b D_{x_1}\right)\right)\tau_{k,l+1}^{(0)} \cdot \tau_{k,l}^{(0)} + 4 b \tau_{k,l+1}^{(1)} \bar{\tau}_{k,l}^{(1)} = 0, \label{KP_bright_dark_10_3}\\
    \mysubeq & \left(D_{y_1^{(2)}}\left(D_{x_1}^2 - D_{x_2} + 2b D_{x_1}\right)\right)\tau_{k,l+1}^{(0)} \cdot \tau_{k,l}^{(0)}  + 4 b \tau_{k,l+1}^{(2)} \bar{\tau}_{k,l}^{(2)} = 0, \label{KP_bright_dark_10_4}\\
    \newsubeqblock
    \mysubeq & D_{y_1^{(1)}}D_{x_1}\tau_{k,l}^{(0)} \cdot \tau_{k,l}^{(0)}=-2\tau_{k,l}^{(1)} \bar{\tau}_{k,l}^{(1)},\label{KP_bright_dark_11_1}\\
    \mysubeq & D_{y_1^{(2)}}D_{x_1}\tau_{k,l}^{(0)} \cdot \tau_{k,l}^{(0)}=-2\tau_{k,l}^{(2)} \bar{\tau}_{k,l}^{(2)},\label{KP_bright_dark_11_2}\\
    \mysubeq & \left(D_{x_{-1}^{(1)}} D_{x_1}-2\right) \tau_{k,l}^{(0)} \cdot \tau_{k,l}^{(0)}=-2 \tau_{k+1,l}^{(0)} \tau_{k-1,l}^{(0)}, \label{KP_bright_dark_11_3} \\
    \mysubeq & \left(D_{x_{-1}^{(2)}} D_{x_1}-2\right) \tau_{k,l}^{(0)} \cdot \tau_{k,l}^{(0)}=-2 \tau_{k,l+1}^{(0)} \tau_{k,l-1}^{(0)}, \label{KP_bright_dark_11_4} \\
    & D_{x_1}\tau_{k,l}^{(1)} \cdot \tau_{k,l}^{(2)} = \tau_{k,l}^{(1,2)} \tau_{k,l}^{(0)}, \label{KP_bright_dark_12}\\
    & \left(D_{x_1} + a\right)\tau_{k+1,l}^{(0)} \cdot \tau_{k,l}^{(1)} = a \tau_{k+1,l}^{(1)} \tau_{k,l}^{(0)},\label{KP_bright_dark_13}\\
    & \left(D_{x_1} + a\right)\tau_{k+1,l}^{(0)} \cdot \tau_{k,l}^{(2)} = a \tau_{k+1,l}^{(2)} \tau_{k,l}^{(0)},\label{KP_bright_dark_14}\\
    & \left(D_{x_1} + b\right)\tau_{k,l+1}^{(0)} \cdot \tau_{k,l}^{(1)} = b \tau_{k,l+1}^{(1)} \tau_{k,l}^{(0)},\label{KP_bright_dark_15}\\
    & \left(D_{x_1} + b\right)\tau_{k,l+1}^{(0)} \cdot \tau_{k,l}^{(2)} = b \tau_{k,l+1}^{(2)} \tau_{k,l}^{(0)},\label{KP_bright_dark_16}\\
    & \left(D_{x_1} + a-b\right)\tau_{k+1,l}^{(0)} \cdot \tau_{k,l+1}^{(0)} = (a-b) \tau_{k+1,l+1}^{(0)} \tau_{k,l}^{(0)},\label{KP_bright_dark_17}
\end{align}
}
where \(k,l=1,\ldots, N\), are satisfied by the following \(\tau\)-functions
\begin{align}\label{tau_bright_dark}
    \begin{split}
        &\tau_{k,l}^{(0)} = \left|M_{k,l}\right|,
        \quad
        \tau_{k,l}^{(1)} = \begin{vmatrix}
            M_{k,l} & \Phi_{k,l} \\
            -\left(\bar{\Psi}\right)^T & 0
        \end{vmatrix},
        \quad 
        \bar{\tau}_{k,l}^{(1)} = \begin{vmatrix}
            M_{k,l} & \Psi \\
            -\left(\bar{\Phi}_{k,l}\right)^T & 0
        \end{vmatrix},
        \\
        &\tau_{k,l}^{(2)} = \begin{vmatrix}
            M_{k,l} & \Phi_{k,l} \\
            -\left(\bar{\Upsilon}\right)^T & 0
        \end{vmatrix},
        \quad 
        \bar{\tau}_{k,l}^{(2)} = \begin{vmatrix}
            M_{k,l} & \Upsilon \\
            -\left(\bar{\Phi}_{k,l}\right)^T & 0
        \end{vmatrix},
        \quad
        \tau_{k,l}^{(1,2)} = \begin{vmatrix}
            M_{k,l} & \Phi_{k,l} & \partial_{x_1} \Phi_{k,l} \\
            -\bar{\Upsilon}^T & 0 & 0 \\
            -\bar{\Psi}^T & 0 & 0 
        \end{vmatrix},
    \end{split}
\end{align}
where \(M_{k,l} = \left(m_{ij}^{k,l}\right)\) is an \(N \times N\) matrix, and \(\Phi_{k,l}\), \(\bar{\Phi}_{k,l}\), \(\Psi\), \(\bar{\Psi}\), \(\Upsilon\), \(\bar{\Upsilon}\) are vectors defined as
{\allowdisplaybreaks\begin{align}
    m_{ij}^{k,l}&=\frac{1}{p_i+\bar{p}_j}\left(-\frac{p_i - a}{\bar{p}_j + a}\right)^k \left(-\frac{p_i - b}{\bar{p}_j + b}\right)^l e^{\xi_i+\bar{\xi}_j}+\frac{\tilde{C}_i \bar{C}_j}{q_i+\bar{q}_j}e^{\eta_i+\bar{\eta}_j}+\frac{\tilde{D}_i \bar{D}_j}{r_i+\bar{r}_j}e^{\chi_i+\bar{\chi}_j},\\
    \Phi_{k,l} &= \left( \left(1-p_1 a^{-1}\right)^k \left(1-p_1 b^{-1}\right)^l e^{\xi _{1}}, \left(1-p_2 a^{-1}\right)^k \left(1-p_2 b^{-1}\right)^l e^{\xi _2},\ldots, \right. \nonumber\\
    &\left.\left(1-p_N a^{-1}\right)^k \left(1-p_N b^{-1}\right)^l e^{\xi _{N}}\right)^T, \\
    \bar{\Phi}_{k,l} &= \left( \left(1+\bar{p}_1 a^{-1}\right)^{-k} \left(1+\bar{p}_1 b^{-1}\right)^{-l} e^{\bar{\xi}_{1}}, \left(1+\bar{p}_2 a^{-1}\right)^{-k} \left(1+\bar{p}_2 b^{-1}\right)^{-l} e^{\bar{\xi}_2},\ldots ,\right. \nonumber\\
    &\left. \left(1+\bar{p}_N a^{-1}\right)^{-k} \left(1+\bar{p}_N b^{-1}\right)^{-l} e^{\bar{\xi}_{N}}\right)^T, \quad\\
    \Psi &=\left(\tilde{C}_1 e^{\eta _{1}}, \tilde{C}_2 e^{\eta _2},\ldots, \tilde{C}_N e^{\eta _{N}}\right)^T,\quad
    \bar{\Psi}=\left(\bar{C}_1 e^{\bar{\eta}_1}, \bar{C}_2 e^{\bar{\eta}_2},\ldots, \bar{C}_N e^{\bar{\eta}_{N}}\right)^T,\\
    \Upsilon &=\left(\tilde{D}_1 e^{\chi _{1}}, \tilde{D}_2 e^{\chi _2},\ldots, \tilde{D}_N e^{\chi _{N}}\right)^T,\quad
    \bar{\Upsilon}=\left(\bar{D}_1 e^{\bar{\chi}_1}, \bar{D}_2 e^{\bar{\chi}_2},\ldots, \bar{D}_N e^{\bar{\chi}_{N}}\right)^T,\\
    \xi_i&=p_i x_1 + p_i^2 x_2 + p_i^3 x_3 +\frac{1}{p_i - a} x_{-1}^{(1)} +\frac{1}{p_i - b} x_{-1}^{(2)} +\xi_{i0},\\
    \bar{\xi}_i&=\bar{p}_i x_1 - \bar{p}_i^2 x_2 + \bar{p}_i^3 x_3 +\frac{1}{\bar{p}_i + a} x_{-1}^{(1)}+\frac{1}{\bar{p}_i + b} x_{-1}^{(2)}+\bar{\xi}_{i0},\\
    \eta_i&=q_i y_1^{(1)},\quad \bar{\eta}_i=\bar{q}_i y_1^{(1)}, \quad \chi_i=r_i y_1^{(2)},\quad \bar{\chi}_i=\bar{r}_i y_1^{(2)}.
\end{align}}
\end{lemma}

We remark here that among bilinear equations in Lemma 2, the bilinear equations \eqref{KP_bright_dark_1}, \eqref{KP_bright_dark_2_3}, \eqref{KP_bright_dark_2_4}, 
\eqref{KP_bright_dark_3}, 
\eqref{KP_bright_dark_4_3}, 
\eqref{KP_bright_dark_4_4}, 
\eqref{KP_bright_dark_11_1}, 
\eqref{KP_bright_dark_11_2} 
\eqref{KP_bright_dark_12}
are the ones used in \cite{gilson2003sasa} to derive bright soliton solutions to the SS equation. 
Eqs. \eqref{KP_bright_dark_5}, \eqref{KP_bright_dark_6}, \eqref{KP_bright_dark_7_1}, 
\eqref{KP_bright_dark_7_2}, 
\eqref{KP_bright_dark_8}, 
\eqref{KP_bright_dark_9}, 
\eqref{KP_bright_dark_10_1}, 
\eqref{KP_bright_dark_10_2} 
\eqref{KP_bright_dark_10_3}, 
\eqref{KP_bright_dark_10_4} 
\eqref{KP_bright_dark_11_3}, 
\eqref{KP_bright_dark_11_4} \eqref{KP_bright_dark_17}
are used in \cite{ohta2010dark} to derive dark soliton solutions to the SS equation.
The rest of bilinear equations are new and used to derive bright-dark soliton solutions to the CSS equation.

\begin{lemma}\label{dim_reduction_4cmkdv}
    By requiring the condition
    \begin{align}\label{dm_para_br_dk}
        \begin{split}
            q_i = \frac{1}{c_1} \left(-p_i + \frac{c_3 \rho_3^2}{p_i - a} + \frac{c_4 \rho_4^2}{p_i - b}\right), \quad \bar{q}_i = \frac{1}{c_1} \left(-\bar{p}_i + \frac{c_3 \rho_3^2}{\bar{p}_i + a} + \frac{c_4 \rho_4^2}{\bar{p}_i + b}\right),\\
            r_i = \frac{1}{c_2} \left(-p_i + \frac{c_3 \rho_3^2}{p_i - a} + \frac{c_4 \rho_4^2}{p_i - b}\right), \quad \bar{r}_i = \frac{1}{c_2} \left(-\bar{p}_i + \frac{c_3 \rho_3^2}{\bar{p}_i + a} + \frac{c_4 \rho_4^2}{\bar{p}_i + b}\right),\\
        \end{split}
    \end{align}
    where \(i=1,2,\ldots,N\), the tau functions given in \eqref{tau_bright_dark} satisfy the differential relation
    \begin{equation}\label{dm_cond_br_dk}
       \left( c_1 \partial_{y^{(1)}} + c_2 \partial_{y^{(2)}} + c_3 \rho_3^2 \partial_{x_{-1}^{(1)}} + c_4 \rho_4^2 \partial_{x_{-1}^{(2)}} \right) \tau_{k,l}^{(m)}= \partial_{x_1} \tau_{k,l}^{(m)}, \quad m=0,1,2.
    \end{equation}
\end{lemma}
\begin{proof}
    For \(\tau_{k,l}^{(0)}\), we have
    \begin{align}
        &\tau_{k,l}^{(0)}=\prod_{n=1}^{N} e^{\xi_n+\bar{\xi}_n} \det\left(\mathfrak{m}_{i j}^{k,l} \right),
        \nonumber\\ 
        &\mathfrak{m}_{i j}^{k,l}= \frac{1}{p_i+\bar{p}_j}\left(-\frac{p_i - a}{\bar{p}_j + a}\right)^k\left(-\frac{p_i - b}{\bar{p}_j + b}\right)^l+\frac{\tilde{C}_i \bar{C}_j}{q_i+\bar{q}_j}e^{\eta_i+\bar{\eta}_j-\xi_i-\bar{\xi}_j}+\frac{\tilde{D}_i \bar{D}_j}{r_i+\bar{r}_j}e^{\chi_i+\bar{\chi}_j-\xi_i-\bar{\xi}_j}. \label{frakm}
    \end{align}
    The term \(\prod_{n=1}^{N} e^{\xi_n+\bar{\xi}_n}\) can be dropped in \eqref{KP_bright_dark_1}-\eqref{KP_bright_dark_17} due to the property of the \(D\)-operator in the bilinear equations. 
    Note that 
    \begin{align}
        \begin{split}\label{diff_rel_bright_dark}
        &\left(c_1 \partial_{y^{(1)}} + c_2 \partial_{y^{(2)}} + c_3 \rho_3^2 \partial_{x_{-1}^{(1)}} + c_4 \rho_4^2 \partial_{x_{-1}^{(2)}} - \partial_{x_1}\right)\mathfrak{m}_{i j}^{k}\\
        &=\left(c_1 (q_i + \bar{q}_j) - c_3 \rho_3^2 \left(\frac{1}{p_i-a} + \frac{1}{\bar{p}_j+a}\right) \right.\\
        &\quad \left.- c_4 \rho_4^2 \left(\frac{1}{p_i-b} + \frac{1}{\bar{p}_j+b}\right) + (p_i + \bar{p}_j)\right)\frac{\tilde{C}_i \bar{C}_j}{q_i+\bar{q}_j}e^{\eta_i+\bar{\eta}_j-\xi_i-\bar{\xi}_j}\\
        &\quad + \left(c_2 (r_i + \bar{r}_j) - c_3 \rho_3^2 \left(\frac{1}{p_i-a} + \frac{1}{\bar{p}_j+a}\right) - \right.\\
        &\quad \left.c_4 \rho_4^2 \left(\frac{1}{p_i-b} + \frac{1}{\bar{p}_j+b}\right) + (p_i + \bar{p}_j)\right)\frac{\tilde{D}_i \bar{D}_j}{r_i+\bar{r}_j}e^{\chi_i+\bar{\chi}_j-\xi_i-\bar{\xi}_j}.
        \end{split}
    \end{align}
    By substituting \eqref{dm_para_br_dk}, it follows that 
    \[\left(c_1 \partial_{y^{(1)}} + c_2 \partial_{y^{(2)}} + c_3 \rho_3^2 \partial_{x_{-1}^{(1)}} + c_4 \rho_4^2 \partial_{x_{-1}^{(2)}}\right)\mathfrak{m}_{i j}^{k} = \partial_{x_1}\mathfrak{m}_{i j}^{k} \]
    and
    \begin{eqnarray*}
        \left(c_1 \partial_{y^{(1)}} + c_2 \partial_{y^{(2)}} + c_3 \rho_3^2 \partial_{x_{-1}^{(1)}} + c_4 \rho_4^2 \partial_{x_{-1}^{(2)}}\right)\tau_k^0&=&\sum_{i,j=1}^{N} \Delta_{ij}\left(\rho_1^2 \partial_r+\rho_2^2 \partial_s\right)\mathfrak{m}_{ij}^{\mathrm{n}}\\&=&\sum_{i,j=1}^{N} \frac{1}{c}\partial_x\mathfrak{m}_{ij}^{\mathrm{n}}=\frac{1}{c}\partial_x\tau_k^0.
    \end{eqnarray*}
    
    For \(\tau_{k,l}^{(1)}\), it is found that
    \begin{align}
        \tau_{k,l}^{(1)} = \prod_{n=1}^{N} e^{\xi_n+\bar{\xi}_n} \begin{vmatrix}
        \mathfrak{m}_{i j}^{k} & \left(\dfrac{p_i - a}{-a}\right)^k\left(\dfrac{p_i - b}{-b}\right)^l \\
        -\bar{C}_j e^{\bar{\eta}_j - \bar{\xi}_j} & 0
    \end{vmatrix},
    \end{align}
    where \(\mathfrak{m}_{i j}^{k}\) is defined as \eqref{frakm}.
    Similar to the previous case, by \eqref{dm_para_br_dk}, we have 
    \[\left(c_1 \partial_{y^{(1)}} + c_2 \partial_{y^{(2)}} + c_3 \rho_3^2 \partial_{x_{-1}^{(1)}} + c_4 \rho_4^2 \partial_{x_{-1}^{(2)}}\right)e^{\bar{\eta}_j - \bar{\xi}_j} = \partial_{x_1}e^{\bar{\eta}_j - \bar{\xi}_j}.\]
    It follows that
    \begin{equation*}
        \left(c_1 \partial_{y^{(1)}} + c_2 \partial_{y^{(2)}} + c_3 \rho_3^2 \partial_{x_{-1}^{(1)}} + c_4 \rho_4^2 \partial_{x_{-1}^{(2)}}\right)\tau_{k,l}^{(1)} = \frac{1}{c}\partial_x\tau_{k,l}^{(1)}.
    \end{equation*}
    The proofs for \(\bar{\tau}_{k,l}^{(1)},\ \tau_{k,l}^{(2)},\ \bar{\tau}_{k,l}^{(2)}\) and \( \tau_{k,l}^{(1,2)}\) go analogously.
\end{proof}
The above dimension reduction allows us to combine \eqref{KP_bright_dark_2_1}-\eqref{KP_bright_dark_2_4}, \eqref{KP_bright_dark_4_1}-\eqref{KP_bright_dark_4_4}, \eqref{KP_bright_dark_7_1}-\eqref{KP_bright_dark_7_4}, \eqref{KP_bright_dark_10_1}-\eqref{KP_bright_dark_10_4}, and \eqref{KP_bright_dark_11_1}-\eqref{KP_bright_dark_11_4}, which give rise to
{\allowdisplaybreaks\begin{align}
    &\left(D_{x_1}^3 - D_{x_1} D_{x_2} - 4 ((c_3\rho_3^2 + c_4\rho_4^2)D_{x_1} - c_3\rho_3^2 a - c_4\rho_4^2 b)\right)\tau_{k,l}^{(1)}\cdot \tau_{k,l}^{(0)}\nonumber\\
    &\qquad - 4 c_3 \rho_3^2 a \tau_{k+1,l}^{(1)} \tau_{k-1,l}^{(0)} - 4 c_4 \rho_4^2 b \tau_{k,l+1}^{(1)} \tau_{k,l-1}^{(0)} +4c_2 \tau_{k,l}^{(1,2)} \bar{\tau}_{k,l}^{(2)}  = 0,\label{KP_bright_dark_2.1}\\
    &\left(D_{x_1}^3 - D_{x_1} D_{x_2} - 4 ((c_3\rho_3^2 + c_4\rho_4^2)D_{x_1} - c_3\rho_3^2 a - c_4\rho_4^2 b)\right)\tau_{k,l}^{(2)}\cdot \tau_{k,l}^{(0)}\nonumber\\
    &\qquad - 4 c_3 \rho_3^2 a \tau_{k+1,l}^{(1)} \tau_{k-1,l}^{(0)} - 4 c_4 \rho_4^2 b \tau_{k,l+1}^{(1)} \tau_{k,l-1}^{(0)} +4c_1 \tau_{k,l}^{(2,1)} \bar{\tau}_{k,l}^{(1)}  = 0,\label{KP_bright_dark_4.1}\\
    &\left(D_{x_1}^3-D_{x_1}D_{x_2}+2 a D_{x_1}^2-4(c_3\rho_3^2+c_4\rho_4^2) D_{x_1}+c_4\rho_4^2(a-b)\right) \tau_{k+1,l}^{(0)} \cdot \tau_{k,l}^{(0)} \nonumber\\
    &\qquad + 4c_4\rho_4^2(a-b)\tau_{k+1,l+1}^{(0)} \cdot \tau_{k,l-1}^{(0)} + 4c_1 a \tau_{k+1,l}^{(1)} \bar{\tau}_{k,l}^{(1)} + 4c_1 a \tau_{k+1,l}^{(2)} \bar{\tau}_{k,l}^{(2)}=0,\label{KP_bright_dark_7.1}\\
    &\left(D_{x_1}^3-D_{x_1}D_{x_2}+2 a D_{x_1}^2-4(c_3\rho_3^2+c_4\rho_4^2) D_{x_1}+c_3\rho_3^2(b-a)\right) \tau_{k,l+1}^{(0)} \cdot \tau_{k,l}^{(0)} \nonumber\\
    &\qquad + 4c_3\rho_3^2(b-a)\tau_{k+1,l+1}^{(0)} \cdot \tau_{k-1,l}^{(0)} + 4c_1 a \tau_{k,l+1}^{(1)} \bar{\tau}_{k,l}^{(1)} + 4c_1 a \tau_{k,l+1}^{(2)} \bar{\tau}_{k,l}^{(2)}=0,\label{KP_bright_dark_10.1}\\
    &D_{x_1} \tau_{k,l}^{(0)} \cdot \tau_{k,l}^{(0)} - 2 (c_3 \rho_3^2+c_4\rho_4^2) \tau_{k,l}^{(0)} \cdot \tau_{k,l}^{(0)} \nonumber \\
    &\qquad = -2c_1 \tau_{k,l}^{(1)} \tau_{k,l}^{(1)} -2c_2 \tau_{k,l}^{(2)} \tau_{k,l}^{(2)} - 2c_3\rho_2^2 \tau_{k+1,l}^{(0)} \tau_{k-1,l}^{(0)} - 2c_4\rho_4^2 \tau_{k,l+1}^{(0)} \tau_{k,l-1}^{(0)},\label{KP_bright_dark_11.1}
\end{align}}
respectively. 
Next, to eliminate the terms involving \(D_{x_1} D_{x_2}\), we consider the combinations
\begin{equation*}
\frac{3\times\eqref{KP_bright_dark_2.1}+\eqref{KP_bright_dark_1}}{4}, \quad \frac{3\times\eqref{KP_bright_dark_4.1}+\eqref{KP_bright_dark_3}}{4}, 
\end{equation*}
and
\begin{equation*}
\frac{3\times\eqref{KP_bright_dark_7.1}+3a\times\eqref{KP_bright_dark_5}+\eqref{KP_bright_dark_6}}{4}, \quad \frac{3\times\eqref{KP_bright_dark_10.1}+3a\times\eqref{KP_bright_dark_8}+\eqref{KP_bright_dark_9}}{4}.
\end{equation*}
Each operation results in 
{\allowdisplaybreaks\begin{align}
    &\left(D_{x_1}^3 - D_{x_3} -3(c_3\rho_3^2 + c_4\rho_4^2)D_{x_1} + 3c_3\rho_3^2 a + 3c_4\rho_4^2 b\right)\tau_{k,l}^{(1)}\cdot \tau_{k,l}^{(0)}\nonumber\\
    &\qquad - 3 c_3 \rho_3^2 a \tau_{k+1,l}^{(1)} \tau_{k-1,l}^{(0)} - 3 c_4 \rho_4^2 b \tau_{k,l+1}^{(1)} \tau_{k,l-1}^{(0)} + 3c_2 \tau_{k,l}^{(1,2)} \bar{\tau}_{k,l}^{(2)}  = 0,\label{KP_bright_dark_2.2}\\
    &\left(D_{x_1}^3 - D_{x_3} - 3(c_3\rho_3^2 + c_4\rho_4^2)D_{x_1} + 3c_3\rho_3^2 a + 3c_4\rho_4^2 b\right)\tau_{k,l}^{(2)}\cdot \tau_{k,l}^{(0)}\nonumber\\
    &\qquad - 3 c_3 \rho_3^2 a \tau_{k+1,l}^{(1)} \tau_{k-1,l}^{(0)} - 3 c_4 \rho_4^2 b \tau_{k,l+1}^{(1)} \tau_{k,l-1}^{(0)} + 3 c_1 \tau_{k,l}^{(2,1)} \bar{\tau}_{k,l}^{(1)}  = 0,\label{KP_bright_dark_4.2}\\
    &\left(D_{x_1}^2 - D_{x_3} + 3a D_{x_1}^2 - 3(a^2+c_3\rho_3^2+c_4\rho_4^2) D_{x_1} + c_4\rho_4^2(a-b)\right) \tau_{k+1,l}^{(0)} \cdot \tau_{k,l}^{(0)} \nonumber\\
    &\qquad + 3c_4\rho_4^2(a-b)\tau_{k+1,l+1}^{(0)} \cdot \tau_{k,l-1}^{(0)} + 3c_1 a \tau_{k+1,l}^{(1)} \bar{\tau}_{k,l}^{(1)} + 3c_1 a \tau_{k+1,l}^{(2)} \bar{\tau}_{k,l}^{(2)}=0,\label{KP_bright_dark_7.2}\\
    &\left(D_{x_1}^2 - D_{x_3} + 3b D_{x_1}^2 - 3(b+c_3\rho_3^2+c_4\rho_4^2) D_{x_1} + c_3\rho_3^2(b-a)\right) \tau_{k,l+1}^{(0)} \cdot \tau_{k,l}^{(0)} \nonumber\\
    &\qquad + 3c_3\rho_3^2(b-a)\tau_{k+1,l+1}^{(0)} \cdot \tau_{k-1,l}^{(0)} + 3c_1 a \tau_{k,l+1}^{(1)} \bar{\tau}_{k,l}^{(1)} + 3c_1 a \tau_{k,l+1}^{(2)} \bar{\tau}_{k,l}^{(2)}=0.\label{KP_bright_dark_10.2}
\end{align}}
By applying the transformation \(x_1 \to x_1 - 3(c_3\rho_3^2+c_4\rho_4^2)x_3, x_3\to x_3\), equations \eqref{KP_bright_dark_2.2}-\eqref{KP_bright_dark_10.2} align with the desired bilinear forms \eqref{4H_BL_1}-\eqref{4H_BL_4}. This transformation allows us to eliminate the variables \(x_2, x_{-1}^{(1)}, x_{-1}^{(2)}, y_1^{(1)}, y_1^{(2)}\) and hence we may set them to be zero. Moreover, let us set \(x_1 = x, x_3 = t\).

Finally, we address the complex conjugate reduction.
\begin{lemma}\label{cmp_reduction_4cmkdv}
    Let \( a = \i\alpha_3 \in \i \mathbb{R}, b = \i\alpha_4 \in \i \mathbb{R} \), and impose the following parameter constraints
    \begin{equation}\label{cmp_reduction_4cmkdv_cond}
       \bar{p}_i^* = p_i, \quad \bar{\xi}_{i,0}^* = \xi_{i0}, \quad C_i = \left(\tilde{C}_i\right)^* = \bar{C}_i, \quad D_i = \left(\tilde{D}_i\right)^* = \bar{D}_i,
    \end{equation}
    where \( i = 1, 2, \ldots, N \). Under these constraints, the tau functions given in \eqref{tau_bright_dark} satisfy the following complex conjugate relations:
    \begin{equation}\label{bd_complex_reduction_res}
        \tau_{-k,-l}^{(i)} = \left(\tau_{k,l}^{(i)}\right)^*, \quad i = 0,1,2. 
    \end{equation}
\end{lemma}
\begin{proof}
    With the restrictions \eqref{cmp_reduction_4cmkdv_cond}, we have
    \begin{align*}
        &\bar{\xi}_i=\bar{p}_i \left(x - 3(c_3\rho_3^2+c_4\rho_4^2)t\right)+\bar{p}_i^3 t+\bar{\xi}_{i,0}=p_i^* \left(x - 3(c_3\rho_3^2+c_4\rho_4^2)t\right)+(p_i^*)^3 t+\xi_{i,0}^*=\xi_i^*,\\
        &\bar{q}_i = \frac{1}{c_1} \left(-\bar{p}_i + \frac{c_3 \rho_3^2}{\bar{p}_i + \i \alpha_3} + \frac{c_4 \rho_4^2}{\bar{p}_i + \i \alpha_4}\right) = \frac{1}{c_1} \left(-p_i^* + \frac{c_3 \rho_3^2}{(p_i - \i \alpha_3)^*} + \frac{c_4 \rho_4^2}{(p_i - \i \alpha_4)^*}\right) = q_i^*.
    \end{align*}
    Similarly, we can show \(\bar{r}_i = r_i^*\). Consequently,
    \begin{align*}
        \left(m_{i,j}^{k,l}\right)^*&=\frac{1}{p_i^*+\bar{p}_j^*}\left(-\frac{p_i^* + \i \alpha_3}{\bar{p}_j^* - \i \alpha_3}\right)^k \left(-\frac{p_i^* + \i \alpha_4}{\bar{p}_j^* - \i \alpha_4}\right)^l e^{\xi_i^*+\bar{\xi}_j^*}+\frac{\tilde{C}_i^* \bar{C}_j^*}{q_i^*+\bar{q}_j^*}+\frac{\tilde{D}_i^* \bar{D}_j^*}{r_i^*+\bar{r}_j^*} \\
        &= \frac{1}{\bar{p}_i+p_j}\left(-\frac{\bar{p}_i - \i \alpha_3}{\bar{p}_j + \i \alpha_3}\right)^k \left(-\frac{\bar{p}_i - \i \alpha_4}{\bar{p}_j + \i \alpha_4}\right)^l e^{\xi_i+\bar{\xi}_j}+\frac{\bar{C}_i \tilde{C}_j}{\bar{q}_i+q_j} + \frac{\bar{D}_i \tilde{D}_j}{\bar{r}_i+r_j}= m_{j,i}^{-k,-l}.
    \end{align*}
    Based on this, we deduce that 
    \begin{align*}
        \bar{\tau}_{-k,-l}^{(1)}&=\begin{vmatrix}
            m_{i,j}^{k,l} & \tilde{C}_i \\ -\left(\dfrac{\i \alpha_3}{\bar{p}_i + \i \alpha_3}\right)^{-k}\left(\dfrac{\i \alpha_4}{\bar{p}_i + \i \alpha_4}\right)^{-l} e^{\bar{\xi}_j} & 0
        \end{vmatrix} \\&= \begin{vmatrix}
            m_{i,j}^{k,l} & -\tilde{C}_i \\ \left(\dfrac{\i \alpha_3}{\bar{p}_i + \i \alpha_3}\right)^{-k}\left(\dfrac{\i \alpha_4}{\bar{p}_i + \i \alpha_4}\right)^{-l} e^{\bar{\xi}_j} & 0
        \end{vmatrix} \\
        &= \begin{vmatrix}
            m_{j,i}^{k,l} & \left(\dfrac{\i \alpha_3}{\bar{p}_j + \i \alpha_3}\right)^{-k}\left(\dfrac{\i \alpha_4}{\bar{p}_j + \i \alpha_4}\right)^{-l} e^{\bar{\xi}_j} \\ -\tilde{C}_i & 0
        \end{vmatrix} \\&= \begin{vmatrix}
            \left(m_{i,j}^{-k,-l}\right)^* &  \left(\left(\dfrac{p_j - \i \alpha_3}{-\i \alpha_3}\right)^k\right)^*\left(\left(\dfrac{p_j - \i \alpha_4}{-\i \alpha_4}\right)^l\right)^* e^{\xi_i^*} \\ -\left(\bar{C}_i\right)^* & 0
        \end{vmatrix}\\
        &= \left(\tau_{-k,-l}^{(1)}\right)^*.
    \end{align*}
    Similarly, it can be shown that \(\bar{\tau}_{-k,-l}^{(2)}=\left(\tau_{-k,-l}^{(2)}\right)^*\) and 
    \[\tau_{-k,-l}^{(0)} = \det\left(m_{i,j}^{-k}\right) = \det\left(m_{j,i}^{-k,-l}\right) = \det\left(\left(m_{i,j}^{k,l}\right)^*\right) = \left(\tau_{k,l}^{(0)}\right)^*.\] 
    Thus, the relation \eqref{bd_complex_reduction_res} is satisfied.

\end{proof}
Now, we set
\begin{align*}
    &\tau_{0,0}^{(1)} = g_1 = \bar{\tau}_{0,0}^{(1)} = g_1^*, \quad \tau_{0,0}^{(2)} = g_2 = \bar{\tau}_{0,0}^{(2)} = g_2^*, \\ &\tau_{1,0}^{(0)} = h_3 = \bar{\tau}_{1,0}^{(0)} = h_3^*, \quad \tau_{0,1}^{(0)} = h_4 = \bar{\tau}_{0,1}^{(0)} = h_4^*, \\
    &\tau_{0,0}^{(1,2)} = s_{12} = -s_{21}, \quad \tau_{1,0}^{(1)} = r_{13} = r_{31},\quad \tau_{0,1}^{(1)} = r_{14} = r_{41}, \\
    &\tau_{1,0}^{(2)} = r_{23} = r_{32},\quad \tau_{0,1}^{(2)} = r_{24} = r_{42}, \quad \tau_{1,1}^{(0)} = r_{34} = r_{43},\quad \tau_{0,0}^{(0)} = f,
\end{align*}
where \(f\) is a real-valued function. Then the bilinear equations \eqref{KP_bright_dark_2.2}-\eqref{KP_bright_dark_10.2} reduce to \eqref{4H_BL_1}-\eqref{4H_BL_4}, \eqref{KP_bright_dark_11.1} reduces to \eqref{4H_BL_5}, \eqref{KP_bright_dark_12}-\eqref{KP_bright_dark_17} reduce to \eqref{4H_BL_6}-\eqref{4H_BL_11}, respectively. As a consequence, the two-bright-two-dark soliton solutions of the four-component Hirota equation are derived.

\subsection{Proof of \cref{thm:bd_css}} \label{sec:deriva_bd_css}
\begin{lemma}\label{lma:4cmkdv_to_css}
    Let \(\epsilon_1 = c_1 = c_2\), \(\epsilon_2 = c_3 = c_4\),  \(\alpha = \alpha_3 = -\alpha_4\) and \( \rho = \rho_3 = \rho_4\), and impose the following parameter constraints
    \begin{align}
        p_i = p_{N+1-i}^*, \quad \xi_{i,0} = \xi_{N+1-i,0}, \quad C_i = D_{N+1-i}^*,
    \end{align}
  then we have \(g_1 = g_2^*, h_3 = h_4^*\).
\end{lemma}
\begin{proof}
    Under the assumptions, we have 
    \begin{equation*}
        \omega_1 = \omega_2 = 3\rho^2 \epsilon_2 \alpha_3 -3\rho^2 \epsilon_2 \alpha_3 = 0, \quad \omega_3 = \alpha_3^3 + 6\epsilon_2 \rho^2 \alpha_3 = -\omega_4
    \end{equation*}
and
    {\allowdisplaybreaks\begin{align*}
        \xi_i^* ={}& p_i^* \left(x - 3(c_3\rho_3^2+c_4\rho_4^2)t\right)+(p_i^*)^3 t+\xi_{i,0}^* \\
        ={}& p_{N+1-i} \left(x - 3(c_3\rho_3^2+c_4\rho_4^2)t\right)+(p_{N+1-i})^3 t+\xi_{N+1-i,0} = \xi_{N+1-i}, \\
        q_i^* ={}& \frac{1}{c_1} \left(-p_i^* + \frac{c_3 \rho_3^2}{p_i^* + \i \alpha_3} + \frac{c_3 \rho_3^2}{p_i^* - \i \alpha_3}\right) \\
        ={}& \frac{1}{c_2} \left(-p_{N+1-i} + \frac{c_3 \rho_3^2}{p_{N+1-i} + \i \alpha_3} + \frac{c_3 \rho_3^2}{p_{N+1-i} - \i \alpha_3}\right) = r_{N+1-i},\\
        \left(m_{i,j}^{k,l}\right)^*={}&\frac{1}{p_i^*+p_j}\left(-\frac{p_i^* + \i \alpha_3}{p_j - \i \alpha_3}\right)^k \left(-\frac{p_i^* - \i \alpha_3}{p_j + \i \alpha_3}\right)^l e^{\xi_i^*+\xi_j}+\frac{C_i C_j^*}{q_i^*+q_j}+\frac{D_i D_j^*}{r_i^*+r_j} \\
        ={}& \frac{1}{p_{N+1-i}+p_{N+1-j}^*}\left(-\frac{p_{N+1-i} + \i \alpha_3}{p_{N+1-j}^* - \i \alpha_3}\right)^k \left(-\frac{p_{N+1-i} - \i \alpha_3}{p_{N+1-j}^*+ \i \alpha_3}\right)^l  e^{\xi_{N+1-i}+\xi_{N+1-j}^*}\\
        &+\frac{D_{N+1-i}^* D_{N+1-j}}{d_{N+1-i}+d_{N+1-j}^*}+\frac{C_{N+1-i} C_{N+1-j}^*}{q_{N+1-i}+q_{N+1-j}^*} \\
        ={}& m_{N+1-i,N+1-j}^{l,k}.
    \end{align*}}
    Thus, we have the equalities
    \begin{align*}
        &g_1 = \tau_{0,0}^{(1)} = \begin{vmatrix}
            m_{i,j}^{0,0} & e^{\xi_i} \\ -C_j  & 0
        \end{vmatrix} = \begin{vmatrix}
            \left(m_{N+1-i,N+1-j}^{0,0}\right)^* & e^{\xi_{N+1-j}^*} \\ -D_{N+1-j}^*  & 0
        \end{vmatrix} = \begin{vmatrix}
            \left(m_{i,j}^{0,0}\right)^* & e^{\xi_{j}^*} \\ -D_j^*  & 0
        \end{vmatrix} = g_2^*,\\
        &h_3 = \tau_{1,0}^{(0)} = \begin{vmatrix}
            m_{i,j}^{1,0}
        \end{vmatrix} = \begin{vmatrix}
            \left(m_{N+1-i,N+1-j}^{0,1}\right)^*
        \end{vmatrix} = \begin{vmatrix}
            \left(m_{i,j}^{0,1}\right)^*
        \end{vmatrix} = h_4^*.
    \end{align*}
\end{proof}
According to \cref{thm:2b2d_4cmkdv}, the conditions \eqref{cond_4cmkdv_to_css} and \cref{lma:4cmkdv_to_css}, the bright-dark soliton solutions to the coupled Sasa-Satsuma equation \eqref{css_1}-\eqref{css_2} can be derived by setting
\begin{align*}
    &u_1 = v_1 = \frac{g_1}{f} = \frac{g_2^*}{f} = v_2^*, \\
    &u_2 = v_3 = \rho \exp(\i (\alpha x - (\alpha^3 + 6\epsilon_2 \rho \alpha) t))\frac{h_3}{f} = \rho \exp(-\i (\alpha x - (\alpha^3 + 6\epsilon_2 \rho \alpha) t))\frac{h_4^*}{f} = v_4^*.
\end{align*}
\section{Conclusion}\label{sect:conclusion}
In this paper, we derived various soliton solutions to the CSS equation under mixed boundary conditions which were not considered in our previous papers. The solutions are given in terms of $N \times N $ determinants. The solution of $N=1$ seems not reported in the literature and the solutions of $N=2$ include both the bright-dark soliton and breather, whose types were identified. For higher order soliton interactions, the types can be changed due to collisions. In the form of the solutions given in the present paper, we haven't found resonant bright-dark soliton solutions. It remains an interesting topic to be explored.

The general $N$-soliton solutions of the multi-component NLS equation under mixed boundary conditions were derived in \cite{feng2014general} using the KP reduction method. In addition, soliton solutions under mixed boundary conditions have also been obtained for various other multi-component integrable systems. 




The stability of the bright-dark soliton solutions deserves a careful investigation, which is beyond the scope of the present paper. Nevertheless, the compact determinant form presented in this paper may facilitate such an analysis.

\section*{Funding}
B.F. Feng was partially supported by National Science Foundation (NSF) under Grant No. DMS-1715991 and U.S. Department of Defense (DoD), Air Force for Scientific Research (AFOSR) under grant No. W911NF2010276.
C.F. Wu was supported by the National Natural Science Foundation of China (Grant No. 12471077) and Shenzhen Natural Science Fund (Stable Support Project of Shenzhen, Grant No. 20231121103530003).  
G.X. Zhang was supported by Student Research Cultivation Project at the Institute for Advanced Study of Shenzhen University.

\section*{Conflict of interests}
The authors have no conflicts to disclose.

\section*{Data availability statements}
All the data are the Matlab codes for plotting the graphs and are available upon the request.

\catcode`'=9
\catcode``=9
\bibliographystyle{agsm}
\bibliography{references}

\end{document}